%% file: ksat.tex
\documentclass{article}
\usepackage{amsfonts, amsmath, amsthm, amssymb, dsfont, enumitem, mathrsfs, titlefoot}
\usepackage[margin=1in]{geometry}
\usepackage{hyperref}
\usepackage{xcolor}

\hypersetup{
  colorlinks   = true, 
  urlcolor     = blue, 
  linkcolor    = blue, 
  citecolor   = red 
}
\numberwithin{equation}{section}

\makeatletter
\newcommand\RedeclareMathOperator{%
  \@ifstar{\def\rmo@s{m}\rmo@redeclare}{\def\rmo@s{o}\rmo@redeclare}%
}
\newcommand\rmo@redeclare[2]{%
  \begingroup \escapechar\m@ne\xdef\@gtempa{{\string#1}}\endgroup
  \expandafter\@ifundefined\@gtempa
     {\@latex@error{\noexpand#1undefined}\@ehc}%
     \relax
  \expandafter\rmo@declmathop\rmo@s{#1}{#2}}
\newcommand\rmo@declmathop[3]{%
  \DeclareRobustCommand{#2}{\qopname\newmcodes@#1{#3}}%
}
\@onlypreamble\RedeclareMathOperator
\makeatother

\makeatletter
\newtheorem*{rep@theorem}{\rep@title}
\newcommand{\newreptheorem}[2]{%
\newenvironment{rep#1}[1]{%
 \def\rep@title{#2 \ref{##1}}%
 \begin{rep@theorem}}%
 {\end{rep@theorem}}}
\makeatother

\newtheorem{theorem}{Theorem}[section]
\newtheorem{lemma}[theorem]{Lemma}
\newtheorem{proposition}[theorem]{Proposition}
\newtheorem{corollary}[theorem]{Corollary}
\newtheorem{fact}[theorem]{Fact}
\newtheorem{conjecture}[theorem]{Conjecture}
\newreptheorem{lemma}{Lemma}

\theoremstyle{definition}
\newtheorem{definition}[theorem]{Definition}
\newtheorem{example}[theorem]{Example}
\newtheorem{algorithm}[theorem]{Algorithm}

\theoremstyle{remark}
\newtheorem{remark}[theorem]{Remark}

\input{tex/commands}

\input{tex/title}

\begin{document}

\maketitle

\input{tex/0-abstract}

\setcounter{tocdepth}{1}
\tableofcontents

\input{tex/1-intro}
\input{tex/2-results}
\input{tex/3-ogp}
\input{tex/4-ldp-hardness}
\input{tex/5-multi-ogp}
\input{tex/6-ldp-stability}
\input{tex/7-local-hardness}
\input{tex/8-simulation}
\input{tex/9-achievability}
\input{tex/10-discussion}

\bibliographystyle{alpha}
\bibliography{./bib}

\appendix

\input{tex/a-kappa}

\end{document}

%% file: tex/commands.tex
\newcommand{\f}{\frac}
\newcommand{\lt}{\left}
\newcommand{\rt}{\right}

\newcommand{\vocab}{\textit}

\DeclareMathOperator*{\E}{{\mathbb{E}}}
\RedeclareMathOperator*{\P}{{\mathbb{P}}}
\DeclareMathOperator*{\Var}{{\mathrm{Var}}}

\newcommand{\eps}{\varepsilon}

\newcommand{\bN}{{\mathbb{N}}}
\newcommand{\bR}{{\mathbb{R}}}
\newcommand{\bZ}{{\mathbb{Z}}}

\newcommand{\cA}{{\mathcal{A}}}
\newcommand{\cB}{{\mathcal{B}}}

\newcommand{\cE}{{\mathcal{E}}}
\newcommand{\cG}{{\mathcal{G}}}
\newcommand{\cL}{{\mathcal{L}}}
\newcommand{\cP}{{\mathcal{P}}}
\newcommand{\cPt}{{\cP_2}}
\newcommand{\cQ}{{\mathcal{Q}}}

\newcommand{\cX}{{\mathcal{X}}}
\newcommand{\cV}{{\mathcal{V}}}

\newcommand{\sD}{{\mathsf{D}}}
\newcommand{\sE}{{\mathsf{E}}}

\newcommand{\tO}{{\widetilde O}}
\newcommand{\tOmega}{{\widetilde \Omega}}
\newcommand{\tTheta}{{\widetilde \Theta}}

\newcommand{\T}{{\mathtt{T}}}
\newcommand{\F}{{\mathtt{F}}}
\newcommand{\Q}{{\mathtt{err}}}

\newcommand{\bx}{{\bar x}}

\newcommand{\oT}{{T'}}

\newcommand{\betast}{{\beta^*}}
\newcommand{\kappast}{{\kappa^*}}
\newcommand{\kappastst}{{\kappa^{**}}}
\newcommand{\psist}{\psi^*}
\newcommand{\varphist}{\varphi_*}

\newcommand{\betamax}{{\beta_{\max}}}
\newcommand{\betamin}{{\beta_{\min}}}
\newcommand{\bp}{{\beta_{+}}}
\newcommand{\bm}{{\beta_{-}}}
\newcommand{\bz}{{\beta_0}}

\newcommand{\phip}[1]{\Phi^{(#1)}}
\newcommand{\Gp}[1]{{G^{(#1)}}}
\newcommand{\Sp}[1]{{S^{(#1)}}}
\newcommand{\vp}[1]{v^{(#1)}}
\newcommand{\xp}[1]{x^{(#1)}}
\newcommand{\yp}[1]{y^{(#1)}}
\newcommand{\zp}[1]{z^{(#1)}}
\newcommand{\varphiep}[1]{\varphi_E^{(#1)}}
\newcommand{\Psip}[1]{\Psi^{(#1)}}
\newcommand{\psip}[1]{\psi^{(#1)}}
\newcommand{\rhop}[1]{\rho^{(#1)}}

\newcommand{\kst}{{k^*}}
\newcommand{\nst}{{n^*}}
\newcommand{\tst}{{t^*}}

\newcommand{\Rst}{{R^*}}

\newcommand{\hg}{{\hat g}}

\newcommand{\tq}{{\tilde q}}
\newcommand{\tLambda}{{\tilde \Lambda}}

\newcommand{\ind}[1]{\mathds{1}\lt\{#1\rt\}}
\newcommand{\norm}[1]{\lt\|#1\rt\|}
\newcommand{\diff}[1]{{~\mathrm{d}#1}}

\newcommand{\indep}{{\perp\!\!\!\perp}}

\newcommand{\Svalid}{S_{\mathrm{valid}}}
\newcommand{\Sconsec}{S_{\mathrm{consec}}}
\newcommand{\Sindep}{S_{\mathrm{indep}}}
\newcommand{\Sogp}{S_{\mathrm{ogp}}}
\newcommand{\Snobad}{S_{\text{no-bad}}}
\newcommand{\Sconc}{S_{\mathrm{conc}}}

\newcommand{\DFG}{{\mathrm{DFG}}}
\newcommand{\DGW}{{\mathrm{DGW}}}

\newcommand{\PGW}{{\mathrm{PGW}}}
\newcommand{\Sat}{{\mathrm{Sat}}}
\newcommand{\Rep}{{\mathrm{Rep}}}

\newcommand{\Pois}{{\mathrm{Pois}}}
\newcommand{\unif}{{\mathrm{unif}}}

\newcommand{\Va}{{\mathrm{Va}}}
\newcommand{\Cl}{{\mathrm{Cl}}}

\newcommand{\pa}{{\mathrm{pa}}}
\newcommand{\gr}{{\mathrm{gr}}}

\newcommand{\round}{{\texttt{round}}}
\newcommand{\strictround}{{\texttt{strictRound}}}

\newcommand{\XOR}{{\texttt{XOR}}}

\newcommand{\poly}{{\mathrm{poly}}}
\newcommand{\polylog}{{\mathrm{polylog}}}

%% file: tex/title.tex
\title{The Algorithmic Phase Transition of Random $k$-SAT \\ for Low Degree Polynomials}
\author{
    Guy Bresler\thanks{Massachusetts Institute of Technology, Department of EECS. Email: \texttt{guy@mit.edu}. Supported by MIT-IBM Watson AI Lab and NSF CAREER award CCF-1940205.}
    \and
    Brice Huang\thanks{Massachusetts Institute of Technology, Department of EECS. Email: \texttt{bmhuang@mit.edu}. Supported by NSF Graduate Research Scholarship 1745302, a Siebel Scholarship, and NSF TRIPODS award 1740751.}
}
\date{October 29, 2021}
\keywords{Random $k$-SAT; average-case complexity; algorithmic phase transition; low degree polynomial hardness; overlap gap property.}

%% file: tex/0-abstract.tex
\begin{abstract}
    Let $\Phi$ be a uniformly random $k$-SAT formula with $n$ variables and $m$ clauses.
    We study the algorithmic task of finding a satisfying assignment of $\Phi$.
    It is known that satisfying assignments exist with high probability up to clause density $m/n = 2^k \log 2 - \f12 (\log 2 + 1) + o_k(1)$, while the best polynomial-time algorithm known, the \verb|Fix| algorithm of Coja-Oghlan \cite{Coj10}, finds a satisfying assignment at the much lower clause density $(1 - o_k(1)) 2^k \log k / k$.
    This prompts the question: is it possible to efficiently find a satisfying assignment at higher clause densities?

    We prove that the class of \emph{low degree polynomial} algorithms cannot find a satisfying assignment at clause density $(1 + o_k(1)) \kappast 2^k \log k / k$ for a universal constant $\kappast \approx 4.911$.
    This class encompasses \verb|Fix|, message passing algorithms including Belief and Survey Propagation guided decimation (with bounded or mildly growing number of rounds), and local algorithms on the factor graph.
    This is the first hardness result for any class of algorithms at clause density within a constant factor of that achieved by \verb|Fix|.
    Our proof establishes and leverages a new many-way overlap gap property tailored to random $k$-SAT.
\end{abstract}

%% file: tex/1-intro.tex
\section{Introduction}
\label{sec:intro}

The $k$-SAT problem occupies a central role in complexity theory as the first and canonical NP-complete problem \cite{Coo71}.
Its average-case analogue, random $k$-SAT, likewise has a central role in average-case computational complexity.
In this problem, we are given a $k$-CNF formula with $m$ clauses and $n$ variables whose $km$ literals are sampled uniformly and i.i.d. from the $2n$ possible literals;\footnote{In a variant of this definition, the $m$ clauses are chosen uniformly and without replacement among all $2^k \binom{n}{k}$ clauses with $k$ distinct, non-complementary literals. This definition behaves identically to ours in the large-$n$ limit, and all properties of random $k$-SAT we show in this paper apply equally to this model.} see \cite{Ach09} for a survey.
There are two natural fundamental questions for random $k$-SAT.
First, at what scalings of $(n,k,m)$ are there satisfying assignments? Second, when can they be found by efficient algorithms?

Early work showed that for fixed $k$ the interesting regime of random $k$-SAT is when $m = \Theta(n)$, and that the problem's qualitative behavior in the large-$n$ limit depends on the \vocab{clause density} $\alpha = m/n$.
Namely, \cite{FP83} showed that if $\alpha \ge 2^k \log 2$, random $k$-SAT is unsatisfiable with high probability; on the positive side, \cite{CF90} showed that if $\alpha < 2^k / k$, a simple algorithm finds a satisfying assignment with nontrivial probability, and \cite{CR92} improved the guarantee to with high probability.

As we tune $\alpha$, we encounter phase transitions separating one qualitative behavior from another.
Two phase transitions are of primary interest to us: the \vocab{satisfiability threshold}, below which random $k$-SAT admits a satisfying assignment with high probability, and the \vocab{algorithmic threshold}, below which a polynomial-time algorithm produces a satisfying assignment with high probability.

The satisfiability threshold is well understood.
\cite{KKKS98} showed that random $k$-SAT is unsatisfiable with high probability at clause density $2^k \log 2 - \f12 (\log 2 + 1) + o_k(1)$, where $o_k(1)$ denotes a term limiting to $0$ as $k\to\infty$.
\cite{CP16} showed that for a different $o_k(1)$ term, random $k$-SAT is satisfiable with high probability at clause density $2^k \log 2 - \f12 (\log 2 + 1) + o_k(1)$.
For large $k$, the landmark result of Ding, Sly, and Sun \cite{DSS15} precisely identified the satisfiability threshold $\alpha_s(k)$ within this $o_k(1)$ range, proving that with high probability, random $k$-SAT is satisfiable when $\alpha < \alpha_s(k)$ and unsatisfiable when $\alpha > \alpha_s(k)$.

In the present paper we study the algorithmic threshold, which is much less understood.
The best polynomial-time algorithm known, the \verb|Fix| algorithm of Coja-Oghlan \cite{Coj10}, finds a satisfying assignment with high probability at clause density $(1-o_k(1)) 2^k \log k / k$, nearly a factor of $k$ below the satisfiability threshold.
A body of evidence has emerged to suggest that this is the correct threshold, but rigorous results that efficient algorithms fail beyond this threshold have been scarce.

In the early 2000s, statistical physicists developed a rich but non-rigorous theory describing the solution geometry of random $k$-SAT, among other random constraint satisfaction problems \cite{KMRSZ07}.
This theory predicts several phase transitions in random $k$-SAT's solution geometry, which we now summarize; see \cite[Figure 2]{KMRSZ07} for an illustration.
At low clause density, the space of satisfying assignments is one large cluster.
When the clause density reaches the \vocab{uniqueness threshold}, disconnected solution clusters appear but the main cluster contains all but an exponentially small fraction of solutions.
At the \vocab{clustering threshold}, the solution space shatters into an exponentially large number of clusters, each with an exponentially small fraction of solutions.
Additional clauses cause these clusters to shrink until at the \vocab{condensation threshold}, the solution space is dominated by a few clusters of strongly varying sizes.
Finally, beyond the \vocab{satisfiability threshold} there are no satisfying assignments.
Many of these predictions have since been proven rigorously: the prediction of the satisfiability threshold was confirmed in \cite{DSS15}, and the physics prediction of the random regular NAE-$k$-SAT condensation threshold was recently confirmed in \cite{NSS20}.

\cite{KMRSZ07} predicted that Markov Chain Monte Carlo (MCMC) algorithms succeed up to the clustering threshold and no more.
Since then, this threshold has emerged as the predicted limit of \emph{all} efficient algorithms, and structural phenomena in the clustered regime have been rigorously established that (still non-rigorously) suggest algorithmic hardness.
\cite{AC08} showed that clustering occurs at clause density $(1+o_k(1)) 2^k \log k / k$, confirming the prediction of \cite{KMRSZ07}.
They showed that at this clause density, long-range correlations appear in random $k$-SAT's solution space, in the following sense.
Say variable $x_i$ of a satisfying assignment $x\in \{\T,\F\}^n$ is \vocab{frozen} if any satisfying assignment $y$ with $x_i\neq y_i$ is at Hamming distance $\Omega(n)$ from $x$.
Then, in all but an $o(1)$ fraction of satisfying assignments, all but an $o_k(1)$ fraction of bits are frozen with high probability.
This suggests that above this clause density, local search is unlikely to succeed, and any algorithmic solution to random $k$-SAT must use a qualitatively different approach.

The rigorous evidence for the algorithmic threshold consists of exhibiting algorithms on one side and producing bounds against specific algorithms or restricted computational models on the other side.
There is a long history of work on heuristic algorithms for $k$-SAT.
The oldest heuristic is the \vocab{Davis-Putnam-Logemann-Loveland} (DPLL) algorithm \cite{DP60, DLL61}, a backtracking based search algorithm which still forms the basis for many modern SAT solvers.
Other heuristics that have emerged include the \vocab{pure literal} rule \cite{GPB82}; \vocab{unit clause propagation} \cite{CF90}; \vocab{shortest clause} \cite{CR92, FS96}; \vocab{walksat} \cite{Pap91, CFFKV09}; and \vocab{Belief} and \vocab{Survey Propagation guided decimation} \cite{MRS07, BMZ05}.
However, there is no evidence, rigorous or non-rigorous, that any of these algorithms succeed beyond clause density $O_k(2^k/k)$.
(See \cite[Table 1]{Coj10} for a review of these algorithms' performances.)
The breakthrough result \cite{Coj10} produced the algorithm \verb|Fix|, which provably finds a satisfying assignment with high probability up to clause density $(1-o_k(1)) 2^k \log k / k$.
This is the best algorithm to date, and the above physics evidence suggests that this clause density is optimal up to lower order terms.

The earliest rigorous hardness result is \cite{LMS98}, which proved that the pure literal rule does not solve random $3$-SAT above clause density approximately $1.63$.
\cite{AS00} generalized this result, showing that so-called \emph{myopic algorithms} cannot solve random $3$-SAT above clause density approximately $3.26$.
(The random $3$-SAT satisfiability threshold is conjectured to be about $4.26$ \cite{MPZ02}.)

For large $k$, the earliest hardness result is \cite{ABM04}, which showed that DPLL type algorithms require exponential running time beyond clause density $O_k(2^k/k)$.
Note that this threshold is \emph{smaller} than the clause density $(1-o_k(1))2^k \log k / k$ where \verb|Fix| succeeds; thus DPLL algorithms are provably suboptimal.
Gamarnik and Sudan \cite{GS17} showed that \emph{balanced sequential local algorithms}, which include Belief and Survey Propagation guided decimation (with constant or mildly growing number of message passing rounds) cannot solve random NAE-$k$-SAT at clause density $(1+o_k(1)) 2^{k-1} \log^2 k / k$.
The quantity $2^{k-1}$ is the NAE-$k$-SAT analogue of $2^k$ for $k$-SAT.
Remaining negative results are bounds against specific algorithms.
\cite{Het16} proved that Survey Propagation guided decimation (without restriction on the number of rounds) fails at clause density $(1 + o_k(1)) 2^k \log k / k$, and \cite{CHH17} proved that walksat fails at clause density $O_k(2^k \log^2 k / k)$.
Table~\ref{tab:neg-results} summarizes these results.
To date, all negative results either differ from the conjectured threshold $(1+o_k(1)) 2^{k} \log k / k$ by a factor growing in $k$ or are tailored to a specific algorithm.

\begin{table}
    \centering
    \def\arraystretch{1.5}
    \begin{tabular}{|c|c|c|}
        \hline
        Reference & Algorithm or algorithm class & Clause density \\
        \hline
        \hline
        \cite{ABM04} & DPLL algorithms & $O_k(2^k/k)$ \\
        \hline
        \cite{GS17} & Balanced sequential local algorithms (NAE-$k$-SAT) & $(1+o_k(1)) 2^{k-1} \log^2 k / k$ \\
        \hline
        \cite{Het16} & Survey Propagation guided decimation & $(1+o_k(1)) 2^k \log k / k$ \\
        \hline
        \cite{CHH17} & Walksat & $O_k(2^k \log^2 k / k)$ \\
        \hline
        \hline
        This work & Low degree polynomials & $(1+o_k(1))\kappast 2^k \log k / k$ \\
        \hline
    \end{tabular}
    \caption{Algorithmic hardness results for random $k$-SAT with large $k$. The conjectured algorithmic threshold is $(1+o_k(1)) 2^k \log k / k$.}
    \label{tab:neg-results}
\end{table}

In this paper, we show that \emph{low degree polynomial} algorithms do not solve random $k$-SAT above clause density $(1+o_k(1)) \kappast 2^k \log k / k$ for a universal constant $\kappast \approx 4.911$.
Low degree polynomials encompass many of the above algorithms, including \verb|Fix|, Belief and Survey Propagation guided decimation, and local and sequential local algorithms on the factor graph.
This is the first hardness result for any class of algorithms within a constant factor of the conjectured algorithmic threshold.

Our result gives strong evidence that the algorithmic threshold is within a constant factor of $2^k \log k / k$.
Because our techniques link clustering to hardness, we believe the true algorithmic threshold is $(1+o_k(1)) 2^k \log k / k$, matching \verb|Fix| and the onset of clustering; we leave the question of closing this constant factor gap as an important open problem.

\subsection{Algorithmic Hardness from the Overlap Gap Property}

The proof of our result is based on making rigorous an appropriate understanding of random $k$-SAT's solution geometry.
This proof extends a line of work on the \emph{overlap gap property} (OGP) and develops techniques to overcome obstacles limiting the reach of prior OGP methodology.
We now summarize the OGP program and our contribution to it; a more detailed discussion can be found in Section~\ref{sec:ogp-discussion}.

The recent line of work on the OGP \cite{GS14, RV17, GS17, GL18, CGPR19, GJ21, GJW20, Wei20, GK21, GJW21, HS21}, see \cite{Gam21} for a survey, is the first to link physics intuitions about solution geometry to rigorous algorithmic hardness results.
Initiated by Gamarnik and Sudan in \cite{GS14}, the OGP program links algorithmic hardness to an ``overlap gap," a formalization of clustering defined as the absence of a pair of solutions a medium distance apart.
In its original form, an OGP argument shows that in (part of) the clustered regime, the problem's solution space exhibits an overlap gap with high probability.
It then shows that any stable algorithm solving the problem can be used to construct a forbidden pair of solutions, and thus such an algorithm cannot exist.

In many problems, the classic OGP shows stable algorithms fail well below the point where solutions exist, but not to the believed algorithmic threshold.
This is because the overlap gap is a ``worst case" notion of clustering requiring there to be zero solution pairs at medium distance, while the notion of clustering that appears to coincide with hardness is ``average case," allowing a small minority of medium distance solution pairs.
To improve the threshold where hardness for stable algorithms is shown, later work has considered forbidden structures consisting of \emph{several} solutions, which we term multi-OGPs.
Building on the line of work \cite{GS14, RV17, GJW20}, the paper of Wein \cite{Wei20} showed using a multi-OGP that low degree polynomials cannot solve maximum independent set at any objective larger than the believed algorithmic threshold.
Multi-OGPs have also been used to rule out stable algorithms for random NAE-$k$-SAT \cite{GS17} and the Number Partitioning Problem \cite{GK21} well below the existential threshold (and the reach of the classic OGP), and for spin glass optimization \cite{HS21} at the algorithmic threshold.

For random $k$-SAT, early work \cite{DMMZ08} showed that the classic OGP occurs at clause density $(1+o_k(1)) \f12 2^k \log 2$.
This clause density is below the satisfiability threshold, confirming the clustering picture at this clause density.
However, it remains well above the conjectured algorithmic threshold.
Establishing a multi-OGP for random $k$-SAT within a constant factor of the algorithmic limit presents unique challenges.
Our approach is most similar to that of \cite{Wei20}: we express the log first moment of an overlap structure as a free entropy, and our goal is to find an overlap structure making this quantity negative.
However, in contrast to the maximum independent set problem considered in \cite{Wei20}, where the independence of the Erd\H{o}s-R\'enyi graph's edges makes the free entropy analysis tractable by the principle of deferred decisions, the free entropy for random $k$-SAT has complex dependencies which make a tight analysis difficult.
It is a priori unclear how to even define the forbidden structure in the multi-OGP.

We identify the correct forbidden structure and prove that it does not occur with high probability.
To achieve this, we make three conceptual contributions.
First, we define notions of overlap profile and overlap entropy.
Second, we define the multi-OGP in terms of this formalism; this is itself a key innovation, as all (multi-)OGPs in the literature have not required the overlap profile's full power.
Third, we perform a novel free entropy analysis to show the multi-OGP occurs.
We are optimistic that many problems, including those with similarly complex energy landscapes, may be amenable to the techniques developed in this paper.

\subsection{Hardness for Restricted Classes of Computation}

Reasoning about the power of restricted classes of algorithms is at the heart of theoretical computer science.
As discussed above, there is a long line of work showing hardness of random $k$-SAT for various computational models and algorithms \cite{LMS98, AS00, ABM04, GS17, Het16, CHH17}.
More generally, for other problems, the limits of various computational models have been studied, including circuits \cite{Ajt83, FSS84, Has86, CSS18}, the convex hierarchies of Sherali-Adams and L\'ov\'asz-Schrijver (see \cite{CMM09} and references therein), the sum of squares hierarchy \cite{Gri01, KMOW17, BHKKMP19}, and local algorithms on graphs \cite{GS14, RV17}.

Recently, low degree polynomial algorithms have emerged as a prominent class in average case complexity and statistical inference.
As outlined in \cite[Appendix A]{GJW20}, this class contains many popular and powerful frameworks, including spectral methods, local algorithms on graphs, and (approximate) message passing \cite{DMM09, BM11, JM13, Mon19, AMS20, Sel21}.
In addition, a recent flurry of work has shown that for many problems in high-dimensional statistics, including planted clique, sparse PCA, community detection, and tensor PCA, low degree polynomials are as powerful as the best polynomial-time algorithms known \cite{HS17, HKPRSS17, Hop18, BKW20, KWB19, DKWB19, CHKRT20, BB20, LZ20, SW20, BBKMW21, BBHLS21}.
Thus, showing that low degree polynomial algorithms fail at some threshold provides evidence that all polynomial-time algorithms fail at that threshold.

\subsection{Notation}

For all positive integers $n$, $[n]$ denotes the set $\{1,\ldots,n\}$.
For two assignments $x,y\in \{\T,\F\}^n$, let $\Delta(x,y) = \f{1}{n} |\{i\in [n]: x_i\neq y_i\}|$ denote the normalized Hamming distance.
We occasionally consider assignments $x,y\in \{\T,\F,\Q\}^n$ which allow an error symbol; for such assignments the definition of $\Delta$ extends verbatim.

Throughout, $\log$ denotes the natural logarithm. The binary entropy function $H:[0,1]\to [0,\log 2]$ is $H(x) = -x\log x - (1-x)\log (1-x)$.
We often use the basic inequality $H(x)\le x\log \f{e}{x}$.
We also overload notation and denote by $H(\cdot)$ the entropy of certain distributions.
These will be defined where first used.

All our results are in the double limit as $n\to\infty$, and then $k\to\infty$.
Thus, the notations $O(\cdot), \Omega(\cdot), o(\cdot), \omega(\cdot)$ indicate asymptotic behavior in $n$, suppressing any dependence on $k$.
With a tilde (e.g. $\tO(\cdot)$) these notations also suppress $\polylog(n)$ factors.
When subscripted with $k$, these notations indicate asymptotic behavior in $k$ of a quantity independent of $n$.

\paragraph{Organization.}
The rest of this paper is structured as follows.
In Section~\ref{sec:results}, we state our main results.
Section~\ref{sec:ogp-discussion} summarizes the progress of the OGP program and places our contributions in this context.
Sections~\ref{sec:impossibility-ogp} through \ref{sec:impossibility-ldp-stability} are devoted to the proof of Theorem~\ref{thm:impossibility}, our hardness result for low degree polynomials.
Section~\ref{sec:impossibility-ogp} develops the formalism needed to define our central multi-OGP.
This section proves Theorem~\ref{thm:impossibility} assuming Proposition~\ref{prop:impossibility-prob-bounds}(\ref{itm:impossibility-prob-bound-consecutive}), that outputs of the low degree polynomial are stable with nontrivial probability, and Proposition~\ref{prop:impossibility-prob-bounds}(\ref{itm:impossibility-prob-bound-ogp}), that the main multi-OGP occurs with high probability.
Sections~\ref{sec:impossibility-multi-ogp} and \ref{sec:impossibility-ldp-stability} prove these propositions.
Section~\ref{sec:local-algs} proves Theorem~\ref{thm:impossibility-local-algs}, which shows that at clause density $(1+o_k(1)) \kappast 2^k \log k / k$, local algorithms cannot solve random $k$-SAT with even very small probability.
Section~\ref{sec:simulation} shows that a class of algorithms we call local memory algorithms, which include \verb|Fix| and sequential local algorithms, can be simulated by local algorithms and low degree polynomials.
Using these simulation results, Section~\ref{sec:achievability} proves Theorem~\ref{thm:achievability}, our converse achievability result that local algorithms and low degree polynomials both solve random $k$-SAT at clause density $(1-o_k(1)) 2^k \log k / k$.
Section~\ref{sec:discussion} gives concluding remarks.

\paragraph{Acknowledgements.}
We are grateful to the anonymous reviewers for their comments and suggestions, which have improved this paper.
We thank David Gamarnik, Mehtaab Sawhney, Mark Sellke, and Alex Wein for helpful conversations.
This work was done in part while the authors were participating in the Simons Institute programs in Probability, Geometry, and Computation in High Dimensions (Fall 2020) and Computational Complexity of Statistical Inference (Fall 2021).

%% file: tex/2-results.tex
\section{Results}
\label{sec:results}

Throughout this paper, $\cV = \{x_1,\ldots,x_n\}$ denotes a set of propositional variables.
The set of corresponding literals, consisting of the variables in $\cV$ and their negations, is $\cL = \{x_1,\ldots,x_n,\bx_1,\ldots,\bx_n\}$.
Let $\Omega_k(n,m)$ denote the set of all $k$-CNF formulas over $\cV$ with $m$ clauses.
We allow literals to appear multiple times in a clause and clauses to appear multiple times in a formula.
We treat each $\Phi \in \Omega_k(n,m)$ as an ordered $m$-tuple of clauses, each of which is an ordered $k$-tuple of literals.
Let $\Phi_i$ ($i\in [m]$) denote the $i$th clause of $\Phi$ and $\Phi_{i,j}$ ($j\in [k]$) denote the $j$th literal of $\Phi_i$.
The central object of this paper is the following distribution.

\begin{definition}[Random $k$-SAT]
    \label{defn:random-ksat}
    The \vocab{random $k$-SAT} distribution $\Phi_k(n,m)$ is the law of a uniformly random sample from $\Omega_k(n,m)$.
    Equivalently, we can sample $\Phi\sim \Phi_k(n,m)$ by sampling the literals $\Phi_{i,j}$ i.i.d. from $\unif(\cL)$.
\end{definition}

We now define the constant $\kappast$ in our hardness results.
Define the function $\iota : (1, +\infty) \to \bR$ by
\[
    \iota(\beta) = \f{\beta}{1 - \beta e^{-(\beta - 1)}}.
\]
One easily checks that $\iota$ is strictly convex, with $\iota(\beta) \to +\infty$ when $\beta \to 1^+$ or $\beta \to +\infty$.
Let $\kappast = \min \iota(\beta) \approx 4.911$.
The minimum is attained at $\betast\approx 3.513$, the unique solution to $\beta^2 e^{-(\beta-1)} = 1$ in $(1, +\infty)$.

\subsection{Computational Hardness for Low Degree Polynomials}

We study the class of low degree polynomial algorithms, defined as follows.
This is the same computational model considered in \cite{GJW20, Wei20}.

\begin{definition}[Low degree polynomial]
    \label{defn:ldp}
    A \vocab{degree-$D$ polynomial} is a function $f: \bR^N \to \bR^n$ of the form
    \[
        f(x) = \lt(f_1(x), \ldots, f_n(x)\rt),
    \]
    where each $f_i : \bR^N \to \bR$ is a multivariate polynomial (in the ordinary sense) with real coefficients of degree at most $D$.
    A \vocab{random degree-$D$ polynomial} is defined similarly, except the coefficients are random (but independent of the input $x$).
    Formally, for an arbitrary probability space $(\Omega, \P_\omega)$, a random degree-$D$ polynomial is a function $f: \bR^N \times \Omega \to \bR^n$ such that for each $\omega \in \Omega$, $f(\cdot, \omega)$ is a degree-$D$ polynomial.
\end{definition}

\begin{remark}
    We will see in Lemma~\ref{lem:impossibility-assume-deterministic} that randomness does not increase the power of the class of low degree polynomials.
\end{remark}

We now define how to encode a $k$-SAT formula as an input to a low degree polynomial.
Define an arbitrary total order on $\cL$.
We encode each $\Phi\in \Omega_k(n,m)$ as a ``one-hot" vector of indicators $\Phi_{i,j,s}$ ($i\in [m], j\in [k], s\in [2n]$) that $\Phi_{i,j}$ is the $s$th element of $\cL$.
This encoding is an element of $\{0,1\}^N$, where $N = m\cdot k \cdot 2n$.
Slightly abusing notation, we identify $\Phi$ with this encoding.

Next, we define how to interpret the output of a low degree polynomial as a Boolean assignment.
We introduce the symbol $\Q$ and define the function $\round : \bR \to \{\T,\F,\Q\}$ by
\[
    \round(x) =
    \begin{cases}
        \T & x \ge 1, \\
        \F & x \le -1, \\
        \Q & \text{otherwise}.
    \end{cases}
\]
When applied to a real-valued vector, $\round$ is applied coordinate-wise.
Thus, outputs of the polynomial that are at least $1$ represent true, outputs that are at most $-1$ represent false, and
outputs in the interval $(-1, 1)$ are errors.
It is important to exclude $(-1, 1)$ so that a small change in the polynomial output cannot induce a large change in (the valid outputs of) the assignment.

In the following definition, we relax the notion of satisfying assignment in two ways: we allow the algorithm to make mistakes in a small fraction $\eta$ of positions (including all $\Q$ outputs and possibly others), and after repairing these mistakes we allow a small fraction $\nu$ of clauses to not be satisfied.

\begin{definition}
    [$(\eta, \nu)$-satisfy]
    Let $\eta, \nu \in [0,1]$.
    An assignment $x\in \{\T,\F\}^n$ \emph{$\nu$-satisfies} $\Phi \in \Omega_k(n,m)$ if it satisfies at least $(1-\nu)m$ clauses of $\Phi$.
    Moreover, $x \in \{\T,\F,\Q\}^n$ \emph{$(\eta, \nu)$-satisfies} $\Phi$ if there exists $y\in \{\T,\F\}^n$ such that $\Delta(x,y) \le \eta$ and $y$ $\nu$-satisfies $\Phi$.
\end{definition}
We remark that any $x$ with more than $\eta n$ entries equal to $\Q$ does not $(\eta, \nu)$-satisfy $\Phi$.
We will show that for small $\eta, \nu$ independent of $n$, a low degree polynomial cannot produce a satisfying assignment for random $k$-SAT even in this relaxed sense.
Formally, we will show hardness for the following notion of solve.

\begin{definition}
    [$(\delta, \gamma, \eta, \nu)$-solve]
    \label{defn:ldp-solve-ksat}
    Let $\delta, \eta, \nu \in [0,1]$ and $\gamma \ge 1$.
    A random polynomial $f: \bR^N \times \Omega \to \bR^n$ \vocab{$(\delta, \gamma, \eta, \nu)$-solves} $\Phi_k(n,m)$ if the following conditions hold.
    \begin{enumerate}[label=(\alph*), ref=\alph*]
        \item $\P_{\Phi, \omega} \lt[\text{$(\round \circ f)(\Phi, \omega)$ $(\eta, \nu)$-satisfies $\Phi$}\rt] \ge 1-\delta$.
        \item \label{itm:solve-normalization} $\E_{\Phi, \omega} \lt[\norm{f(\Phi, \omega)}_2^2\rt] \le \gamma n$.
    \end{enumerate}
\end{definition}

Here, $\delta$ is the algorithm's failure probability and
$\gamma$ is a normalization parameter.
We think of $\gamma$ as a large constant; condition (\ref{itm:solve-normalization}) is necessary because otherwise we can scale $f$ to make the condition that valid outputs of $f$ are outside the interval $(-1, 1)$ meaningless.

The following theorem is our main result, that no low degree polynomial can solve random $k$-SAT at clause density $\kappa 2^k\log k / k$ for any $\kappa > \kappast$.

\begin{theorem}[Hardness for low degree polynomials]
    \label{thm:impossibility}
    Fix $\kappa > \kappast$.
    Let $\alpha = \kappa 2^k\log k / k$ and $m = \lfloor \alpha n \rfloor$.
    There exists $\kst = \kst(\kappa) > 0$ such that for any $k\ge \kst$, there exists $\nst > 0$, $\eta = \Omega_k(k^{-1})$, $\nu = \f{1}{k^2 2^k}$, and $C_1, C_2 > 0$ (depending on $\kappa, k$) such that the following holds.
    If $n\ge \nst$, $\gamma \ge 1$, $1\le D \le \f{C_1n}{\gamma \log n}$ and
    \[
        \delta \le \exp\lt(-C_2 \gamma D \log n\rt),
    \]
    then there is no random degree-$D$ polynomial that $(\delta, \gamma, \eta, \nu)$-solves $\Phi_k(n,m)$.
\end{theorem}

The only property of low degree polynomials we use is their smoothness, in the sense of Proposition~\ref{prop:impossibility-no-cbad-edge}.
Thus Theorem~\ref{thm:impossibility} applies to any algorithm satisfying the conclusion of this proposition.

Note that Theorem~\ref{thm:impossibility} only rules out algorithms succeeding with quite large probability.
This is a limitation of our methods, shared by all results leveraging OGP to show hardness for low degree polynomials \cite{GJW20, Wei20}.
Our converse achievability result, Theorem~\ref{thm:achievability}, will show that at clause densities where efficient algorithms solving random $k$-SAT exist, they can be simulated by low degree polynomials and succeed with probability larger than that forbidden by Theorem~\ref{thm:impossibility}.
We will also see in Theorem~\ref{thm:impossibility-local-algs} that local algorithms, a more restricted computation class that nonetheless simulates \verb|Fix|, as well as Belief and Survey Propagation Guided Decimation, cannot solve random $k$-SAT with even very small probability.

The constant $\kappast$ can likely be optimized further.
However, without further conceptual insights our methods stall at a value of $\kappast$ strictly larger than $1$, lower bounded by approximately $1.716$.
Thus further ideas are needed to close the constant factor gap between our hardness results and the best algorithms.
See Appendix~\ref{appsec:impossibility-kappast-discussion} for a discussion of these points.
Despite this barrier, we believe the algorithmic phase transition for low degree polynomials does occur at clause density $(1+o_k(1)) 2^k \log k / k$, matching the physics prediction and positive results.
This is formalized in the following conjecture, which we leave as an open problem.
\begin{conjecture}
    \label{conj:alg-threshold-one}
    Theorem~\ref{thm:impossibility} (and Theorem~\ref{thm:impossibility-local-algs}) holds for all $\kappa > 1$.
\end{conjecture}

\subsection{Computational Hardness for Local Algorithms}

We now consider local algorithms on the factor graph.
We first define the factor graph of a $k$-SAT instance.

\begin{definition}[Factor graph]
    \label{defn:factor-graph}
    The factor graph of $\Phi\in \Omega_k(n,m)$ is a signed bipartite graph $(G, \rho)$, where $G = (\Va_G, \Cl_G, E_G)$ is a bipartite graph with left-vertices $\Va_G$, right-vertices $\Cl_G$, and edges $E_G$, and $\rho : E_G \to \{\T,\F\}$ associates each edge with a polarity.
    Here, $\Va_G = \{v_1,\ldots,v_n\}$ and $\Cl_G = \{c_1,\ldots,c_m\}$ represent the variables and clauses of $\Phi$.
    Every literal $x_j$ or $\bar x_j$ in clause $\Phi_i$ corresponds to an edge $e = (v_j, c_i) \in E_G$.
    Edge $e$ has polarity $\rho_G(e) = \T$ if the literal is $x_j$ and $\rho_G(e) = \F$ if the literal is $\bar x_j$.
\end{definition}

To define local algorithms, we first introduce formalism for rooted graphs.
Let $(\Omega, \P_\omega)$ be an arbitrary probability space.
\begin{definition}[Rooted decorated bipartite graph]
    \label{defn:rooted-decorated-bipartite-graph}
    A \emph{decorated bipartite graph} is a tuple $(G, \rho, \varphi)$.
    Here $G = (\Va_G, \Cl_G, E_G)$ is a bipartite graph and $V_G = \Va_G \cup \Cl_G$.
    Moreover, $\rho, \varphi$ are maps $\rho : E_G \to \{\T,\F\}$ and $\varphi : V_G \cup E_G \to \Omega$.
    A \emph{rooted decorated bipartite graph} is a tuple $(v, G, \rho, \varphi)$, where $(G, \rho, \varphi)$ is a decorated bipartite graph and $v\in V_G$.
\end{definition}

Let $\Lambda$ denote the set of rooted decorated bipartite graphs.
Two such graphs are isomorphic of there exists a bijection between them preserving $v, \Va_G, \Cl_G, E_G, \rho, \varphi$.

\begin{definition}[$r$-neighborhood]
    \label{defn:r-nbd}
    Let $(v, G, \rho, \varphi) \in \Lambda$ and $r\in \bN$.
    Define the $r$-neighborhood $N_r(v, G) = (v, G')$, where $\Va_{G'}\subseteq \Va_G$, $\Cl_{G'}\subseteq \Cl_G$ are the sets of vertices reachable from $v$ by a path of length at most $r$ and $E_{G'}$ is the set of edges on those paths.
    Further, define $N_r(v, G, \rho, \varphi) = (v, G', \rho', \varphi') \in \Lambda$, where $(v, G') = N_r(v, G)$ and $\rho' = \rho \big|_{G'}$, $\varphi' = \varphi \big|_{G'}$ are the restrictions of $\rho, \varphi$ to $G'$.
\end{definition}

\begin{definition}[$r$-local function]
    \label{defn:local-fn}
    A function $f$ with domain $\Lambda$ is \emph{$r$-local} if the value of $f(v, G, \rho, \varphi)$ depends only on the isomorphism class of $N_r(v, G, \rho, \varphi)$.
\end{definition}

In other words, a local function has access to the topology of the $r$-neighborhood, the vertex and edge decorations, and the location of the root, but not the identities of the vertices and edges.

\begin{definition}[$r$-local algorithm]
    \label{defn:local-alg}
    Let $f$ be an $r$-local function with codomain $\{\T,\F\}$.
    The $r$-local algorithm based on $f$, denoted $\cA_f$, runs as follows on input $\Phi \in \Omega_k(n,m)$ with factor graph $(G, \rho)$.
    \begin{enumerate}[label=(\arabic*), ref=\arabic*]
        \item Sample $\varphi \sim (\Omega, \P_\omega)^{\otimes (V_G \cup E_G)}$ (i.e. each output of $\varphi : V_G \cup E_G \to \Omega$ is sampled i.i.d. from $(\Omega, \P_\omega)$) independently of $\Phi$.
        \item For each $v = v_i \in \Va_G$, set $x_i = f(v, G, \rho, \varphi)$.
        \item Output $(x_1,\ldots,x_n)\in \{\T,\F\}^n$.
    \end{enumerate}
\end{definition}

We now state our hardness result for local algorithms.

\begin{theorem}[Hardness for local algorithms]
    \label{thm:impossibility-local-algs}
    Fix $\kappa > \kappast$.
    Let $\alpha = \kappa 2^k\log k / k$ and $m = \lfloor \alpha n \rfloor$.
    There exists $\kst = \kst(\kappa) > 0$ such that for any $k\ge \kst$, there exists $\eta = \Omega_k(k^{-2})$ (depending on $\kappa, k$) and $\nu = \f{1}{k^2 2^k}$ such that the following holds.
    For all $r\in \bN$, there exists $\nst > 0$ (depending on $\kappa, k, r$) such that if $n\ge \nst$, then for any $r$-local algorithm $\cA$ with output in $\{\T,\F\}^n$,
    \[
        \P \lt[
            \text{$\cA(\Phi)$ $(\eta, \nu)$-satisfies $\Phi$}
        \rt]
        \le
        \exp(-\tOmega(n^{1/3})).
    \]
    The probability is over the randomness of $\Phi \sim \Phi_k(n,m)$ and the (independent) internal randomness of $\cA$.
\end{theorem}
This theorem rules out a much smaller success probability than Theorem~\ref{thm:impossibility} because our OGP argument in this setting can leverage concentration properties of local algorithms, which are considerably stronger than stability properties of low degree polynomials.

\subsection{Achievability Results}

The following result shows that local algorithms and \emph{constant} degree polynomials solve random $k$-SAT at clause density $(1-\eps)2^k \log k / k$ for any $\eps > 0$.
This gives a lower bound on the algorithmic phase transition within a constant factor and provides a converse to Theorems~\ref{thm:impossibility} and \ref{thm:impossibility-local-algs}.

\begin{theorem}
    \label{thm:achievability}
    Fix $\eps > 0$.
    Let $\alpha = (1-\eps) 2^k \log k / k$ and $m = \lfloor \alpha n \rfloor$.
    There exists $\kst = \kst(\eps) > 0$ such that for any $k\ge \kst$ and $\eta > k^{-12}$, there exist $\nst, r, D, \gamma > 0$ and a sequence $\delta(n)=o(1)$ (dependent on $\eps , k, \eta$) such that the following holds for all $n\ge \nst$.
    \begin{enumerate}[label=(\alph*), ref=\alph*]
        \item \label{itm:achievability-local}
        There exists an $r$-local algorithm $\cA$ such that
        \[
            \P \lt[
                \text{$\cA(\Phi)$ $(\eta, 0)$-satisfies $\Phi$}
            \rt]
            \ge
            1-\delta(n).
        \]
        \item \label{itm:achievability-ldp}
        There exists a (deterministic) degree-$D$ polynomial that $(\delta(n), \gamma, \eta, 0)$-solves $\Phi_k(n,m)$.
    \end{enumerate}
    There also exists a sequence $\nu(n) = o(1)$ (dependent on $\eps, k, \eta$) such that the following holds for all $n\ge \nst$.
    \begin{enumerate}[resume, label=(\alph*), ref=\alph*]
        \item \label{itm:achievability-local-conc}
        There exists an $r$-local algorithm $\cA$ such that
        \[
            \P \lt[
                \text{$\cA(\Phi)$ $(\eta, \nu(n))$-satisfies $\Phi$}
            \rt]
            \ge
            1-\exp(-\tOmega(n^{1/3})).
        \]
        \item \label{itm:achievability-ldp-conc}
        There exists a (deterministic) degree-$D$ polynomial that $(\exp(-\tOmega(n^{1/5})), \gamma, \eta, \nu(n))$-solves $\Phi_k(n,m)$.
    \end{enumerate}
\end{theorem}

We prove this theorem by simulating the first phase of \verb|Fix| by a local algorithm and any local algorithm by a constant degree polynomial.
We can arrange both simulations to be accurate within an arbitrarily small constant (i.e. independent of $n$, arbitrarily small in $k$) normalized Hamming distance, with failure probability $\exp(-\Omega(n^{1/3}))$.
The requirement $\eta > k^{-12}$ arises because the first phase of \verb|Fix| produces an assignment within normalized Hamming distance $k^{-12}$ of a satisfying assignment, which is repaired by the rest of \verb|Fix|.
We believe that it is possible to simulate the rest of \verb|Fix| by a local algorithm, which would show Theorem~\ref{thm:achievability} for any $\eta > 0$; we do not attempt this improvement.
Note that $k^{-12}$ is well within the range of $\eta$ ruled out by our hardness results.

In fact, we will show that local algorithms simulate any \emph{local memory algorithm}.
In this generalization of local algorithms, the algorithm makes its local decisions in series (in a random vertex order), and each decision can leave information on the vertices it accesses, which future decisions can see.
This class includes the first phase of \verb|Fix| and the sequential local algorithms considered in \cite{GS17}.
Recall that the latter class includes Belief and Survey Propagation Guided Decimation.

In parts (\ref{itm:achievability-local-conc},\ref{itm:achievability-ldp-conc}), where the goal is to satisfy all but an $o(1)$ fraction of clauses, Theorem~\ref{thm:achievability} gives algorithms with success probability $1-\exp(-\tOmega(n^{1/3}))$ and $1-\exp(-\tOmega(n^{1/5}))$.
This is within the range ruled out by even Theorem~\ref{thm:impossibility}.
Of course, if the goal is to satisfy \emph{all} clauses as in parts (\ref{itm:achievability-local},\ref{itm:achievability-ldp}), we cannot ensure such a high success probability because $\Phi\sim \Phi_k(n,m)$ is unsatisfiable with probability $1/\poly(n)$ -- for example, if the first $2^k$ clauses each contain variables $x_1,\ldots,x_k$ with all $2^k$ possible polarities.

%% file: tex/3-ogp.tex
\section{The Overlap Gap Program and Sketch of Main Ideas}
\label{sec:ogp-discussion}

In this section, we outline our methods in the context of the OGP literature.
We review the OGP work on maximum independent set and introduce the negative free entropy chaining approach to multi-OGP from \cite{Wei20}.
We discuss the challenges to extending this approach beyond maximum independent set and how we overcome these challenges for random $k$-SAT.

We remark that OGP and (several forms of) multi-OGP have also been harnessed to show the failure of stable algorithms in problems such as largest submatrix \cite{GL18}, maxcut \cite{CGPR19}, number partitioning \cite{GK21}, and spin glass optimization \cite{GJ21, GJW20, GJW21, HS21}.
A different variant of OGP has been linked to hardness in regression and planted problems \cite{GZ17, GZ19, GJS19, BWZ20}.

\subsection{OGP and Multi-OGP for Maximum Independent Set}

Maximum independent set was the first problem where OGP methods derived a sharp algorithmic phase transition.
In this problem, we are given a sample $G\sim G(n,d/n)$ of a sparse Erd\H{o}s-R\'enyi graph and our task is to find a large independent set; the desired size of the set controls the problem difficulty.
We work in the double limit where $n\to \infty$, and then $d \to \infty$.
It is known \cite{Fri90, BGT10} that the largest independent set of this graph has asymptotic size $\f{2\log d}{d} n$.
More precisely, if $S_{\max}$ is the largest independent set, then as $n\to\infty$ for fixed $d$ we have $\f{1}{n}|S_{\max}| \to \alpha_d$, for some $\alpha_d = (1 + o_d(1)) \f{2\log d}{d}$.
However, the best polynomial-time algorithm to date \cite{Kar76} only finds an independent set of asymptotic size $\f{\log d}{d} n$, half the optimum.
It is believed that no polynomial-time algorithm can find an asymptotically larger independent set.

Rigorous results about this problem's solution geometry support this conjecture: \cite{CE15} showed that for any fixed $\eps > 0$, independent sets of size $(1 + \eps) \f{\log d}{d} n$ are clustered in a way that implies that any local Markov chain that samples these sets mixes slowly (but not necessarily that a local Markov chain cannot efficiently find a single such set).

\paragraph{Hardness against local algorithms via OGP and multi-OGP.}
In \cite{GS14}, Gamarnik and Sudan proved that \vocab{local algorithms} (also called \vocab{factors of i.i.d.} algorithms) cannot find independent sets of size $(1 + 1/\sqrt2 + \eps) \f{\log d}{d}n$ for any $\eps > 0$.
Their argument consists of two parts.
First, they show that with high probability, $G\sim G(n, d/n)$ does not have two independent sets of this size with intersection size in $[(1-\delta)\f{\log d}{d}n, (1+\delta)\f{\log d}{d}n]$, for $\delta > 0$ depending on $\eps$.
Then, they construct an interpolation of correlated runs of a putative local algorithm that finds an independent set of the desired size.
From this interpolation, they extract two runs that find two large independent sets with the forbidden intersection, yielding a contradiction.

Rahman and Vir\'ag \cite{RV17} generalized this argument, showing that local algorithms cannot find an independent set of size $(1+\eps) \f{\log d}{d} n$ for any $\eps > 0$.
Their key insight is to consider a forbidden overlap structure involving \emph{several} large independent sets, generated from several correlated runs of a local algorithm.
They also showed local algorithms can find an independent set of size $(1-\eps) \f{\log d}{d}n$, giving the first instance of a multi-OGP identifying a sharp algorithmic phase transition.


\paragraph{Hardness against low degree polynomials by the ensemble innovation.}
Later work extended this impossibility result to low degree polynomials, a significantly more powerful class of algorithms.
Gamarnik, Jagannath, and Wein \cite{GJW20} showed that low degree polynomials cannot find independent sets of size $(1 + 1/\sqrt2 + \eps) \f{\log d}{d}n$.
Their argument leverages an \vocab{ensemble OGP}, an idea introduced in \cite{CGPR19}.
They construct an interpolation, this time over a sequence of \emph{correlated problem instances}.
They show that with high probability, there do not exist two independent sets, \emph{possibly of different problem instances}, of the desired size with intersection size in $[(1-\delta)\f{\log d}{d}n, (1+\delta)\f{\log d}{d}n]$.
Due to the stability of low degree polynomials, the outputs of a low degree polynomial on consecutive problems in the interpolation are close with nontrivial probability.
So, a low degree polynomial finding independent sets of the desired size can be used to construct the forbidden structure.

Wein \cite{Wei20} tightened this result using an \vocab{ensemble multi-OGP}, combining the multi-OGP and ensemble OGP ideas.
In this approach, the interpolation is over a sequence of correlated problem instances and the forbidden structure consists of several independent sets, possibly of different problems, with prescribed many-way overlaps.
Wein showed that a low degree polynomial that finds independent sets of size $(1+\eps) \f{\log d}{d}n$ can be used to construct the forbidden structure, and thus low degree polynomials cannot find independent sets of this size.
Conversely, Wein showed that low degree polynomials can simulate the local algorithms that find independent sets of size $(1-\eps) \f{\log d}{d} n$.
This gives a stronger algorithmic phase transition: low degree polynomials find independent sets of asymptotic size $\f{\log d}{d}n$ and no more.

\paragraph{Negative free entropy chaining in ensemble multi-OGP.}
At a high level, the ensemble multi-OGP in \cite{Wei20} chains together many small negative free entropy contributions to force a free entropy to be negative.
To simplify the discussion, we consider an overlap structure consisting of several independent sets in the \emph{same} problem instance.
We will see that the following argument shows this structure does not occur with high probability exactly when it also shows this structure, where the $\Sp{\ell}$ can be from different problem instances, does not occur with high probability (see Remark~\ref{rem:ensemble-commentary}).
Consider the normalized log first moment
\begin{equation}
    \label{eq:ogp-free-entropy}
    \f{1}{n}
    \log
    \E_{G\sim G(n, d/n)}
    \# \lt(
    \begin{array}{c}
        (\Sp{1},\ldots,\Sp{L}) :
        \text{$\Sp{1},\ldots,\Sp{L}$ are independent} \\
        \text{sets of $G$ of size $(1+\eps) \f{\log d}{d}n$ satisfying $P$}
    \end{array}
    \rt),
\end{equation}
where $P$ is a set of conditions on how $\Sp{1},\ldots,\Sp{L}$ overlap.
The structure inside the expectation in \eqref{eq:ogp-free-entropy} is the forbidden structure we wish to rule out.
The log first moment \eqref{eq:ogp-free-entropy} can be thought of as a free entropy density of the uniform model over copies of this structure; we henceforth refer to \eqref{eq:ogp-free-entropy} as a free entropy.
If \eqref{eq:ogp-free-entropy} is negative, then this structure does not occur with high probability and the multi-OGP occurs.

The key idea in \cite{Wei20} is to set $P = P_2 \cap P_3 \cap \cdots \cap P_L$, where $P_\ell$ is a condition on how $\Sp{\ell}$ overlaps with $\Sp{1},\ldots,\Sp{\ell-1}$, such that the following occurs for all $2\le \ell \le L$.
\begin{enumerate}[label=(\arabic*), ref=\arabic*]
    \item \label{itm:multi-ogp-intuition-small-neg} Let $\cE_\ell$ denote \eqref{eq:ogp-free-entropy} with $(\Sp{1},\ldots,\Sp{\ell})$ in place of $(\Sp{1},\ldots,\Sp{L})$ and $P_2 \cap \cdots \cap P_\ell$ in place of $P$.
    Then, $\cE_\ell$ is smaller than $\cE_{\ell-1}$ by an amount bounded away from $0$.
    Informally, $P_\ell$ requires $\Sp{\ell}$ to overlap with its predecessors in a way that contributes a small negative free entropy to \eqref{eq:ogp-free-entropy}.
    \item \label{itm:multi-ogp-intuition-moat} For any fixed $\Sp{1},\ldots, \Sp{\ell-1}$, if $\Sp{\ell}$ starts at $\Sp{\ell-1}$, evolves by small steps, and eventually evolves far away from all of $\Sp{1},\ldots, \Sp{\ell-1}$, then at some point along this evolution the condition $P_\ell$ occurs.
    Informally, $P_\ell$ defines a moat that a stably evolving $\Sp{\ell}$ must cross.
\end{enumerate}
Due to condition (\ref{itm:multi-ogp-intuition-small-neg}), if we set $L$ large enough, \eqref{eq:ogp-free-entropy} becomes negative, and the structure in \eqref{eq:ogp-free-entropy} is forbidden with high probability.
Suppose a low degree polynomial can find a size $(1+\eps) \f{\log d}{d}n$ independent set with large enough probability.
Because the outputs of a low degree polynomial on a sequence of correlated problem instances is (with nontrivial probability) a stable sequence, condition (\ref{itm:multi-ogp-intuition-moat}) allows us to find a subsequence of $L$ outputs forming the forbidden structure.
Namely, we take $\Sp{1}$ to be the first output in the sequence, and for $\ell \ge 2$ we take $\Sp{\ell}$ to be the first output after $\Sp{\ell-1}$ such that $P_\ell$ holds.
This derives the desired contradiction.

The main technical challenge is to design the $P_\ell$ such that both (\ref{itm:multi-ogp-intuition-small-neg}) and (\ref{itm:multi-ogp-intuition-moat}) hold.
To do this, one must construct a moat topologically disconnecting a high-dimensional space such that, for all values of $\Sp{\ell}$ in the moat, the free entropy decrease in condition (\ref{itm:multi-ogp-intuition-small-neg}) occurs.
The requirement that the moat topologically disconnects the space gives us little control, and therein lies the difficulty.

\cite[Proposition 2.3]{Wei20} carries out this approach by defining $P_\ell$ as the condition that
\[
    \lt|
        \Sp{\ell} \setminus
        \lt(\Sp{1} \cup \cdots \cup \Sp{\ell-1}\rt)
    \rt|
    \in
    \lt[
        \f{\eps \log d}{4d} n,
        \f{\eps \log d}{2d}n
    \rt]
\]
and proving that the free entropy decrease in condition (\ref{itm:multi-ogp-intuition-small-neg}) occurs.

Let us remark on the challenges of extending this technique beyond maximum independent set.
In maximum independent set, due to the independence of the edges of $G\sim G(n,d/n)$, the expectation in \eqref{eq:ogp-free-entropy} is essentially controlled by the total number of non-edges in the union $\Sp{1}\cup \cdots \cup \Sp{L}$, for $\Sp{1},\ldots,\Sp{L}$ with overlap structure satisfying $P$.
This fact makes the analysis of \eqref{eq:ogp-free-entropy} tractable and shows in the relative simplicity of the moats $P_\ell$, which only consider $\Sp{\ell}$'s non-intersection with the union of its predecessors.

In random $k$-SAT and other problems, the corresponding free entropy is more dependent and more tools are needed to carry out this technique.
We develop these tools for random $k$-SAT.
The forbidden structure we devise will take into account more fine-grained overlap information than previous work.

\subsection{Multi-OGP for Random $k$-SAT and Our Contributions}

This paper extends the negative free entropy chaining technique to show an ensemble multi-OGP for random $k$-SAT at clause density $(1+o_k(1)) \kappast 2^k \log k / k$.
We leverage this ensemble multi-OGP to show our hardness results.

Prior to this work, Gamarnik and Sudan \cite{GS17} used a (non-ensemble) multi-OGP to prove that balanced sequential local algorithms do not solve random NAE-$k$-SAT beyond clause density $(1+o_k(1)) 2^{k-1} \log^2 k / k$.
They required the algorithm to be \vocab{balanced}: on any input, each of the algorithm's output bits must be unbiased over the algorithm's internal randomness.
Their interpolation is over correlated runs of the algorithm on a single input, and their proof requires balance to ensure that two fully independent runs give outputs that are far apart.
Due to this requirement, their result required the symmetry provided by the NAE variant of random $k$-SAT.
We improve on this result in three ways:
\begin{enumerate}[label=(\arabic*), ref=\arabic*]
    \item \label{itm:improvement-threshold} We improve the threshold clause density by a logarithmic factor, to $(1+o_k(1)) \kappast 2^k \log k / k$.
    \item \label{itm:improvement-ldp} We generalize the algorithm class from balanced sequential local algorithms to local and low degree algorithms.
    Recall that both of these computation classes simulate sequential local algorithms, even without the balance requirement.
    \item \label{itm:improvement-no-nae} We show hardness for random $k$-SAT instead of NAE-$k$-SAT.
    A simple adaptation of our argument shows hardness of random NAE-$k$-SAT at clause density $(1+o_k(1)) \kappast 2^{k-1} \log k / k$.
\end{enumerate}

\paragraph{Improvements due to ensemble OGP.}
We consider an ensemble multi-OGP, where the random variable resampled in the interpolation is the $k$-SAT instance instead of the algorithm's internal randomness.
The ensemble interpolation allows us to show hardness for local and low degree algorithms.
It also obviates the requirement of balance, so we no longer require the additional symmetry provided by NAE-$k$-SAT.
This achieves improvements (\ref{itm:improvement-ldp}) and (\ref{itm:improvement-no-nae}).

\paragraph{A tighter free entropy analysis.}
Crucially, we conduct a tighter free entropy analysis to achieve improvement (\ref{itm:improvement-threshold}).
In contrast to previous work, our forbidden structure considers all $2^k$ ways $k+1$ satisfying assignments $\yp{0},\dots,\yp{k}$ can agree or disagree.
We formalize such an agreement pattern as an \vocab{overlap profile} $\pi$.
We will introduce this formally in Subsection~\ref{subsec:impossibility-overlap-profile}.
We will see in Lemma~\ref{lem:impossibility-sogp-exp-rate} that the analogue of the free entropy \eqref{eq:ogp-free-entropy} for random $k$-SAT at clause density $\alpha$ is
\begin{equation}
    \label{eq:ogp-ksat-free-entropy}
    \log 2 +
    \max_{\pi \in P}
    \lt[
        H(\pi) -
     \f{\alpha}{2^k}
        \E_{I\sim \unif\lt([n]^k\rt)}
        \lt|\lt\{\yp{\ell}[I]: 0\le \ell\le k\rt\}\rt|
    \rt].
\end{equation}
Here $\yp{0},\dots,\yp{k}$ have overlap profile $\pi$, $P$ is a collection of overlap constraints, and $\pi \in P$ denotes the set of overlap profiles $\pi$ consistent with $P$.
Moreover, $\yp{\ell}[I]$ is the bit string obtained by indexing $\yp{\ell}$ in positions $I$, namely $(\yp{\ell}_{I_1}, \ldots, \yp{\ell}_{I_k})$.
The positive term $H(\pi)$ is the \emph{overlap entropy} of $\pi$, which arises because $\log 2 + H(\pi)$ is the exponential rate of the number of assignment sequences $\yp{0},\dots,\yp{k}$ with overlap profile $\pi$.
The negative term captures the log likelihood that a random formula is satisfied by all of $\yp{0},\dots,\yp{k}$.
We think of these two terms as the entropy and energy terms, respectively.
We will choose $P$ such that the magnitude of the energy term exceeds the entropy term by more than $\log 2$, which causes \eqref{eq:ogp-ksat-free-entropy} to be negative.
This implies the absence (except with exponentially small probability) of a constellation of satisfying assignments with overlap profile $\pi \in P$.

Similarly to \cite{Wei20}, we chain together many small negative free entropies to make \eqref{eq:ogp-ksat-free-entropy} negative.
Because the random $k$-SAT free entropy is dependent and harder to analyze, it is significantly more difficult to identify the correct high-dimensional moats.
In the multi-OGP of \cite{GS17}, the condition $P$ stipulates that the normalized Hamming distances $\Delta(\yp{i}, \yp{j})$ of $k$ satisfying assignments are pairwise approximately $\f{\log k}{k}$.
Using this, the energy term in \eqref{eq:ogp-ksat-free-entropy} can be lower bounded by an inclusion-exclusion truncated at level $2$.
The inclusion-exclusion truncation is not sharp, and consequently this analysis requires the larger clause density $\alpha = (1+o_k(1)) 2^k \log^2 k / k$ (for NAE-$k$-SAT, $(1+o_k(1)) 2^{k-1} \log^2 k / k$) to show that the contribution of each $\yp{\ell}$ to \eqref{eq:ogp-ksat-free-entropy} is a small negative number.
The fact that this natural estimate of \eqref{eq:ogp-ksat-free-entropy} gives a threshold too large by a $\log k$ factor highlights the difficulty of accurately controlling the $k$-SAT free entropy and the necessity of finding good moats.

We find the correct moats.
We set $P = P_1 \cap \cdots \cap P_k$, where $P_\ell$ governs how $\yp{\ell}$ overlaps with its predecessors $\yp{0},\ldots,\yp{\ell-1}$.
Each $P_\ell$ defines a moat that a smooth evolution of $\yp{\ell}$ starting from $\yp{\ell-1}$ must cross.
In order to obtain a fine control over the tradeoff between entropy and energy in \eqref{eq:ogp-ksat-free-entropy}, we develop a notion of \vocab{conditional overlap entropy} $H(\pi(\yp{\ell} | \yp{0}, \ldots, \yp{\ell-1}))$, which is the contribution of $\yp{\ell}$ to the entropy term $H(\pi)$.
Informally, this is a measure of the additional diversity that $\yp{\ell}$ adds to the assignments $\yp{0}, \ldots, \yp{\ell-1}$.
For each $\ell \ge 1$, our condition $P_\ell$ stipulates that
\[
    H\lt(\pi(\yp{\ell} | \yp{0}, \ldots, \yp{\ell-1})\rt)
    \in
    \lt[\bm \f{\log k}{k}, \bp \f{\log k}{k}\rt].
\]
This choice of forbidden structure \emph{in terms of the conditional overlap entropy} is an important contribution of our work.
The choice is motivated by the subsequent energy analysis, which shows a lower bound on the energy contribution of $\yp{\ell}$ that counterbalances the entropy increase.
We next summarize this analysis.

\paragraph{Energy increment bound via decoupling.}
We can express the energy term (without the prefactor) as
\begin{align}
    \label{eq:ogp-energy-term}
    \E_{I\sim \unif\lt([n]^k\rt)}
    \lt|\lt\{\yp{\ell}[I]: 0\le \ell\le k\rt\}\rt|
    &=
    \sum_{\sigma \in \{\T,\F\}^k}
    p(\sigma),
    \quad \text{where} \\
    \notag
    p(\sigma)
    &=
    \P_{I\sim \unif\lt([n]^k\rt)}
    \lt[\text{$\sigma = \yp{\ell}[I]$ for some $0\le \ell\le k$}\rt].
\end{align}
For each $\sigma$, $1-p(\sigma)$ is the probability that $\sigma \neq \yp{\ell}[I]$ for all $0\le \ell \le k$.
This can be conditionally expanded as a product of $k$ factors, where the $\ell$th factor is the probability that $\sigma \neq \yp{\ell}[I]$ given the values of $\yp{0}[I], \ldots, \yp{\ell-1}[I]$.
We think of ($1$ minus) this factor as the contribution of $\yp{\ell}$ to $p(\sigma)$.

We apply the following estimate to decouple these products into sums.
We round any factors in the conditional expansion that are less than $1 - \f{1}{k\log k}$ up to $1$.
Then, we note that for $0\le \eps_1,\ldots,\eps_k \le \f{1}{k\log k}$,
\[
    1 - (1 - \eps_1)(1-\eps_2)\cdots (1-\eps_k) \approx \eps_1 + \eps_2 + \cdots + \eps_k,
\]
up to a $1+o_k(1)$ multiplicative factor.
This decouples the contributions of the $\yp{\ell}$ to the $p(\sigma)$.
We can bound the total contribution of $\yp{\ell}$ to the energy term \eqref{eq:ogp-energy-term} by summing the now-decoupled contributions over $\sigma\in \{\T,\F\}^k$.

\paragraph{Probabilistic reinterpretation.}
Miraculously, this sum can be reinterpreted as the success probability of an experiment involving a sum of $k$ i.i.d. random variables, which can be controlled by concentration inequalities.
We find that if the contribution of $\yp{\ell}$ to the entropy term $H(\pi)$ is $\beta \f{\log k}{k}$, then its contribution to the rescaled energy term \eqref{eq:ogp-energy-term} is at least $1 - \beta e^{-(\beta-1)}$.
This motivates the choice of $\iota(\beta)$ as the (rescaled) ratio of these contributions, and $\kappast$ as the best possible ratio.
When $\alpha = \kappa 2^k \log k / k$ for $\kappa > \kappast$, the condition $P_\ell$ requires $\beta$ to be in a range where the contribution of $\yp{\ell}$ to the energy term of \eqref{eq:ogp-ksat-free-entropy} exceeds its contribution to the entropy term by at least $\eps \f{\log k}{k}$, for constant $\eps$ depending on $\kappa$.
Thus the overall contribution of $\yp{\ell}$ to \eqref{eq:ogp-ksat-free-entropy} is upper bounded by $-\eps \f{\log k}{k}$.
Summed over the $\yp{\ell}$, this shows that \eqref{eq:ogp-ksat-free-entropy} is negative, establishing the multi-OGP.

This energy analysis via decoupling and probabilistic reinterpretation is original and is another key contribution of our work.

\paragraph{Future directions.}
Because we establish a multi-OGP for random $k$-SAT within a constant factor of the conjectured algorithmic threshold, we believe it is possible to leverage multi-OGPs to show algorithmic hardness at or near the limits of efficient algorithms for many other problems.
Closing the remaining constant factor gap and extending the results of this paper to other random constraint satisfaction problems are important open problems.


%% file: tex/4-ldp-hardness.tex
\section{Proof of Impossibility for Low Degree Polynomials}
\label{sec:impossibility-ogp}

This section and the next two sections are devoted to proving our main impossibility result, Theorem~\ref{thm:impossibility}.
Throughout, we fix $\kappa > \kappast$.
We set $\alpha = \kappa 2^k \log k / k$ and $m = \lfloor \alpha n \rfloor$.

\subsection{Reduction to Deterministic Low Degree Polynomial}
\label{subsec:impossibility-assume-deterministic}

The following lemma shows that randomness does not significantly improve the power of low degree polynomial algorithms.

\begin{lemma}
    \label{lem:impossibility-assume-deterministic}
    Suppose there exists a random degree-$D$ polynomial that $(\delta,\gamma,\eta,\nu)$-solves $\Phi_k(n,m)$.
    Then, there exists a deterministic degree-$D$ polynomial that $(3\delta,3\gamma,\eta,\nu)$-solves $\Phi_k(n,m)$.
\end{lemma}
\begin{proof}
    Let $f: \bR^N \times \Omega \to \bR^n$ be a random degree-$D$ polynomial that $(\delta,\gamma,\eta,\nu)$-solves $\Phi_k(n,m)$.
    Then,
    \[
        \E_{\omega}\lt[
            \P_{\Phi} \lt[
            \text{$(\round \circ f)(\Phi, \omega)$ does not $(\eta, \nu)$-satisfy $\Phi$}
            \rt]
        \rt]
        \le \delta
        \qquad
        \text{and}
        \qquad
        \E_{\omega} \lt[
            \E_{\Phi} \lt[
            \norm{f(\Phi, \omega)}_2^2
            \rt]
        \rt]
        \le \gamma n.
    \]
    By Markov's inequality,
    \[
        \P_{\omega} \lt[
            \P_{\Phi} \lt[
                \text{$(\round \circ f)(\Phi, \omega)$ does not $(\eta, \nu)$-satisfy $\Phi$}
            \rt]
            \ge
            3\delta
        \rt]
        \le \f13
        \qquad
        \text{and}
        \qquad
        \P_{\omega} \lt[
            \E_{\Phi} \lt[
                \norm{f(\Phi, \omega)}_2^2
            \rt]
            \ge
            3\gamma n
        \rt]
        \le \f13.
    \]
    So, there exists $\omega \in \Omega$ such that the deterministic polynomial $g(\Phi) = f(\Phi, \omega)$ satisfies
    \[
        \P_{\Phi} \lt[
            \text{$(\round \circ g)(\Phi)$ $(\eta, \nu)$-satisfies $\Phi$}
        \rt]
        \ge 1-3\delta
        \qquad
        \text{and}
        \qquad
        \E_{\Phi} \lt[
            \norm{g(\Phi)}_2^2
        \rt]
        \le 3\gamma n.
    \]
\end{proof}

By Lemma~\ref{lem:impossibility-assume-deterministic}, it suffices to show hardness for deterministic polynomials.
For the rest of this section and Section~\ref{sec:impossibility-ldp-stability}, except where stated, $f : \bR^N \to \bR^n$ is a deterministic degree-$D$ polynomial.

We let $\cA(\Phi) = \cB((\round \circ f)(\Phi), \Phi)$, where $\cB(x,\Phi)$ is a deterministic, computationally unbounded subroutine outputting $y\in \{\T,\F\}^n$ with $\Delta(x,y)\le \eta$.
(If $x$ has more than $\eta n$ entries equal to $\Q$, $\cB$ outputs ``fail.")
Informally, $\cB$ is a computationally unbounded assistant that repairs an $\eta n$ fraction of entries of $(\round \circ f)(\Phi)$.

Because $f$ is deterministic, $\cA$ is also deterministic.
Note that $\round \circ f$ outputting a $(\eta, \nu)$-satisfying assignment of $\Phi$ is equivalent to $\cA$ outputting a $\nu$-satisfying assignment of $\Phi$.
Showing that this does not occur with the required probability will be our task from here on.

\subsection{The Interpolation Path}
\label{subsec:impossibility-interpolation}

We can enumerate the $km$ literals of a formula $\Phi\in \Omega_k(n,m)$ in lexicographic order:
\[
    \Phi_{1,1}, \Phi_{1,2}, \ldots, \Phi_{1,k}, \Phi_{2,1}, \ldots, \Phi_{m,k}.
\]
For $j\in [km]$, let $L(j)$ denote the pair $(a,b)$ such that $\Phi_{a,b}$ is the $j$th literal in this order.
That is, $L(j) = (a,b)$ is the unique pair of integers $(a,b)\in [m]\times [k]$ satisfying $k(a-1)+b = j$.
We now define a sequence of correlated random $k$-SAT formulas.

\begin{definition}
    [Interpolation path]
    \label{defn:interpolation}
    Let $T = k^2m$.
    Let $\phip{0}, \ldots, \phip{T} \in \Omega_k(n,m)$ be the sequence of $k$-SAT instances sampled as follows.
    First, sample $\phip{0} \sim \Phi_k(n,m)$.
    For each $1\le t\le T$, let $\sigma(t) \in [km]$ be the unique integer such that $t \equiv \sigma(t) \pmod{km}$.
    Then, $\phip{t}$ is obtained from $\phip{t-1}$ by resampling $\phip{t}_{L(\sigma(t))}$ from $\unif(\cL)$.
    Moreover, for $0\le t\le T$, let $\xp{t} = \cA(\phip{t})$.
\end{definition}
In other words, we start from a random $k$-SAT instance and resample the literals one by one in lexicographic order.
After we have resampled all the literals we start over, repeating the procedure until each literal has been resampled $k$ times.
Note that each $\phip{t}$ is marginally a sample from $\Phi_k(n,m)$ and that if $|t-t'| \ge km$, then $\phip{t} \indep \phip{t'}$.
We run our assisted low degree algorithm $\cA$ on all these $k$-SAT instances and collect the outputs as the sequence $\xp{0}, \xp{1}, \ldots, \xp{T} \in \{\T,\F\}^n$.

\subsection{Overlap Profiles}
\label{subsec:impossibility-overlap-profile}

We now introduce the overlap profile of an ordered list of assignments.
The overlap profile summarizes the bitwise agreement and disagreement pattern of a list of assignments.

Let $\cPt(\ell)$ denote the set of unordered partitions of $\{0,\ldots,\ell-1\}$ into two (possibly empty) sets.
For example,
\[
    \cPt(3)
    =
    \big\{
      \lt\{ \{0,1,2\}, \emptyset \rt\},
      \lt\{ \{0,1\}, \{2\} \rt\},
      \lt\{ \{0,2\}, \{1\} \rt\},
      \lt\{ \{1,2\}, \{0\} \rt\}
    \big\}.
\]
Note that $|\cPt(\ell)| = 2^{\ell-1}$.

\begin{definition}
    [Overlap profile]
    \label{defn:impossibility-overlap-profile}
    Let $\yp{0}, \ldots, \yp{\ell-1} \in \{\T,\F\}^n$ be a sequence of assignments.
    Their \vocab{overlap profile} $\pi = \pi(\yp{0}, \ldots, \yp{\ell-1})$, is a vector $\pi \in \bR^{2^{\ell - 1}}$ indexed by unordered pairs $\{S,T\}\in \cPt(\ell)$, where
    \[
        \pi_{S,T} =
        \f{1}{n} \lt|
            i\in [n]:
            \text{all $\{\yp{t}_i : t\in S\}$ equal one value and all $\{\yp{t}_i : t\in T\}$ equal the other value}
        \rt|.
    \]
\end{definition}

\begin{example}
    Let $\ell = 3$.
    The overlap profile $\pi = \pi(\yp{0}, \yp{1}, \yp{2})$ consists of four entries $\pi_{012,\emptyset}$, $\pi_{01,2}$, $\pi_{02,1}$, and $\pi_{12,0}$, where
    \[
        \pi_{012,\emptyset}
        =
        \f{1}{n} \lt|
            i\in [n]:
            \yp{0}_i = \yp{1}_i = \yp{2}_i
        \rt|
        \qquad
        \text{and}
        \qquad
        \pi_{01,2}
        =
        \f{1}{n} \lt|
            i\in [n]:
            \yp{0}_i = \yp{1}_i \neq \yp{2}_i
        \rt|,
    \]
    and $\pi_{02,1}, \pi_{12,0}$ are analogous to $\pi_{01,2}$.
\end{example}

We can interpret an overlap profile as a probability distribution: $\pi_{S,T}$ is the probability that in a random position $i\sim \unif([n])$, all $\{\yp{t}_i: t\in S\}$ equal one value and all $\{\yp{t}_i: t\in T\}$ equal the other.
We naturally define the \emph{overlap entropy} of $\yp{0},\ldots,\yp{\ell-1}$ by
\[
    H\lt(\pi(\yp{0}, \ldots, \yp{\ell-1})\rt) =
    -\sum_{\{S,T\} \in \cPt(\ell)}
    \pi_{S,T} \log \pi_{S,T}.
\]
This is the entropy of the unordered pair of sets $\{S,T\}$ obtained by sampling $i\sim \unif([n])$ and partitioning $\{0,\ldots,\ell-1\}$ based on the value of $\yp{t}_i$.

We also define conditional overlap profiles.
Let $\pi = \pi(\yp{0},\ldots,\yp{\ell-1})$.
For each $\{S,T\} \in \cPt(\ell-1)$ with $\pi_{S,T} > 0$, $\pi_{\cdot | S,T}$ is a probability distribution on the two partitions $\{S \cup \{\ell-1\}, T\}$ and $\{S, T \cup \{\ell-1\}\}$ with
\[
    \pi_{S \cup \{\ell-1\}, T | S,T} = \f{\pi_{S \cup \{\ell-1\}, T}}{\pi_{S \cup \{\ell-1\}, T} + \pi_{S, T\cup \{\ell-1\}}}
    \qquad
    \text{and}
    \qquad
    \pi_{S, T \cup \{\ell-1\} | S,T} = \f{\pi_{S, T \cup \{\ell-1\}}}{\pi_{S \cup \{\ell-1\}, T} + \pi_{S, T\cup \{\ell-1\}}}.
\]
(If $\pi_{S \cup \{\ell-1\}, T} = \pi_{S, T\cup \{\ell-1\}} = 0$, we define this distribution arbitrarily.)
This is the distribution of the agreement pattern of $\yp{0},\ldots,\yp{\ell-1}$ on a uniformly random position, conditioned on the agreement pattern of $\yp{0},\ldots,\yp{\ell-2}$ in that position being $\{S,T\}$.
We denote the resulting collection of distributions, one for each $\{S,T\}\in \cPt(\ell-1)$, by $\pi_{\cdot | \cdot} = \pi( \yp{\ell-1} | \yp{0}, \ldots, \yp{\ell-2} )$.
We analogously define the \emph{conditional overlap entropy}
\[
    H\lt( \pi( \yp{\ell-1} | \yp{0}, \ldots, \yp{\ell-2} ) \rt) =
    \sum_{\{S,T\} \in \cPt(\ell-1)}
    \pi_{S,T}
    H(\pi_{\cdot | S,T}).
\]

Before proceeding, we collect some properties of overlap profiles which will be useful in the rest of the section.
The proofs of these assertions follow readily from the above definitions.

\begin{fact}
    \label{fac:overlap-profile-properties}
    Overlap profiles have the following properties.
    \begin{enumerate}[label=(\alph*), ref=\alph*]
        \item \label{itm:overlap-profile-property-count}
        There are at most $n^{2^{\ell-1}}$ distinct overlap profiles of $\ell$ assignments $\yp{0},\ldots,\yp{\ell-1} \in \{\T,\F\}^n$.
        \item \label{itm:overlap-profile-property-chainrule}
        Overlap entropies satisfy the chain rule
        \[
            H\lt( \pi(\yp{0},\ldots,\yp{\ell-1}) \rt)
            =
            H\lt( \pi(\yp{0},\ldots,\yp{\ell-2}) \rt) +
            H\lt( \pi(\yp{\ell-1} | \yp{0},\ldots,\yp{\ell-2}) \rt).
        \]
        \item \label{itm:overlap-profile-property-duplication}
        Repeated assignments do not affect overlap entropies.
        That is, if $\zp{0},\ldots,\zp{r-1}$ are the distinct elements of $\yp{0},\ldots,\yp{\ell-1}$, then
        \[
            H\lt( \pi(\yp{0},\ldots,\yp{\ell-1}) \rt) =
            H\lt( \pi(\zp{0},\ldots,\zp{r-1}) \rt).
        \]
        If $\zp{0},\ldots,\zp{r-2}$ are the distinct elements of $\yp{0},\ldots,\yp{\ell-2}$, then
        \[
            H\lt( \pi(\yp{\ell-1} | \yp{0},\ldots,\yp{\ell-2}) \rt) =
            H\lt( \pi(\yp{\ell-1} | \zp{0},\ldots,\zp{r-2}) \rt).
        \]
        Furthermore, if $\yp{\ell-1} \in \{\yp{0},\ldots,\yp{\ell-2}\}$, then $H\lt( \pi(\yp{\ell-1} | \yp{0},\ldots,\yp{\ell-2}) \rt) = 0$.
    \end{enumerate}
\end{fact}

\subsection{Outline of Proof of Impossibility}
\label{subsec:impossibility-outline}

Recall that $\iota(\beta) = \f{\beta}{1 - \beta e^{-(\beta-1)}}$ is strictly convex with with $\iota(\beta) \to +\infty$ when $\beta \to 1^+$ or $\beta \to +\infty$, and has minimum $\kappast$ attained at $\betast$.
Because $\kappa > \kappast$, there exist two solutions $\betamin, \betamax$ to $\iota(\beta) = \kappa$, with $\betamin \in (1, \betast)$ and $\betamax \in (\betast, +\infty)$.
Set $\bm = \f{\betamin + \betast}{2}$ and $\bp = \f{\betamax + \betast}{2}$.
(This choice is arbitrary; any deterministic $\betamin < \bm < \bp < \betamax$ will do.)
Set $\eps > 0$ such that $\f{\beta + \eps}{1 - \beta e^{-(\beta - 1)}} \le \kappa$ for all $\beta \in [\bm, \bp]$.
We emphasize that $\bm, \bp, \eps$ depend on $\kappa$ only.

For the rest of this proof, take $\eta = \f{\bp - \bm}{8k}$ and $\nu = \f{1}{k^2 2^k}$.
We next define the events $\Svalid, \Sconsec, \Sindep, \Sogp$, which are measurable in the interpolation path $\phip{0},\ldots,\phip{T}$ defined in Definition~\ref{defn:interpolation}.
Define
\[
    \Svalid = \lt\{
        \text{$\xp{t}$ $\nu$-satisfies $\phip{t}$ for all $0\le t\le T$}
    \rt\}.
\]
This is the event that $\cA$ succeeds on all $\phip{t}$.
Define
\[
    \Sconsec = \lt\{
        \text{
            $\Delta(\xp{t}, \xp{t-1}) \le \f{\bp - \bm}{2k}$
            for all $1\le t\le T$
        }
    \rt\}.
\]
This is the event that outputs of $\cA$ on consecutive $\phip{t}$ are close in Hamming distance.
Define $\Sindep$ as the event that there do not exist indices $0\le t_0 \le t_1 \le \cdots \le t_k\le T$ with $t_k \ge t_{k-1}+km$ and an assignment $y\in \{\T,\F\}^n$ such that
\begin{enumerate}[label=(IND-\Alph*), ref=IND-\Alph*, leftmargin=50pt]
    \item \label{itm:def-sindep-satisfy} $y$ $\nu$-satisfies $\phip{t_k}$;
    \item \label{itm:def-sindep-entropy} $H\lt(\pi(y | \xp{t_0}, \ldots,\xp{t_{k-1}}) \rt) \le \bp \f{\log k}{k}$.
\end{enumerate}
This is the event that if $t_k$ is large enough that $\phip{t_k}$ is independent of $\phip{t_0},\ldots,\phip{t_{k-1}}$, then all $\nu$-satisfying assignments to $\phip{t_k}$ have high conditional overlap entropy relative to the outputs of $\cA$ on $\phip{t_0},\ldots,\phip{t_{k-1}}$.
Finally, define $\Sogp$ as the event that there do not exist indices $0\le t_0 \le t_1 \le \cdots \le t_k\le T$ and assignments $\yp{0}, \ldots,\yp{k} \in \{\T, \F\}^n$ such that
\begin{enumerate}[label=(OGP-\Alph*), ref=OGP-\Alph*, leftmargin=50pt]
    \item \label{itm:def-sogp-satisfy} For all $0\le \ell\le k$, $\yp{\ell}$ $\nu$-satisfies $\phip{t_\ell}$;
    \item \label{itm:def-sogp-entropy} For all $1\le \ell\le k$, $H\lt(\pi(\yp{\ell} | \yp{0}, \ldots,\yp{\ell-1}) \rt) \in \lt[\bm \f{\log k}{k}, \bp \f{\log k}{k}\rt]$.
\end{enumerate}
$\Sogp$ defines the main forbidden structure of our argument.
Informally, this forbidden structure consists of $k+1$ assignments, each $\nu$-satisfying possibly different $\phip{t}$ in the interpolation, such that each assignment has medium conditional overlap entropy relative to its predecessors.

The key ingredients in our proof of Theorem~\ref{thm:impossibility} are the following two propositions.
Proposition~\ref{prop:impossibility-events-relation} shows that these four events do not simultaneously occur, and Proposition~\ref{prop:impossibility-prob-bounds} controls their probabilities.
These two propositions derive the main contradiction: if a low degree algorithm $(\delta,\gamma,\eta,\nu)$-solves $\Phi_k(n,m)$ for the requisite $(\delta,\gamma,\eta,\nu)$, then Proposition~\ref{prop:impossibility-events-relation} implies $\Svalid \cap \Sconsec \cap \Sindep \cap \Sogp = \emptyset$, while Proposition~\ref{prop:impossibility-prob-bounds} and a union bound imply $\Svalid \cap \Sconsec \cap \Sindep \cap \Sogp \neq \emptyset$.

\begin{proposition}
    \label{prop:impossibility-events-relation}
    For all sufficiently large $k$, $\Svalid \cap \Sconsec \cap \Sindep \cap \Sogp = \emptyset$.
\end{proposition}

\begin{proposition}
    \label{prop:impossibility-prob-bounds}
    Suppose $f$ is a deterministic degree-$D$ polynomial that $(\delta,\gamma,\eta,\nu)$-solves $\Phi_k(n,m)$.
    For all sufficiently large $k$, the following inequalities hold.
    \begin{enumerate}[label=(\alph*), ref=\alph*]
        \item \label{itm:impossibility-prob-bound-consecutive} $\P(\Svalid \cap \Sconsec) \ge (2n)^{- 4\gamma Dk^2 / (\bp - \bm)} - (T+1) \delta$.
        \item \label{itm:impossibility-prob-bound-independent} $\P(\Sindep^c) \le \exp(-\Omega(n))$.
        \item \label{itm:impossibility-prob-bound-ogp} $\P(\Sogp^c) \le \exp(-\Omega(n))$.
    \end{enumerate}
\end{proposition}

The remainder of this section and Sections~\ref{sec:impossibility-multi-ogp} and \ref{sec:impossibility-ldp-stability} will be devoted to proving these propositions.
We will prove Proposition~\ref{prop:impossibility-events-relation} in Subsection~\ref{subsec:impossibility-construct-ogp} and Proposition~\ref{prop:impossibility-prob-bounds}(\ref{itm:impossibility-prob-bound-independent}) in Subsection~\ref{subsec:impossibility-ogp-independent}.
We will prove Proposition~\ref{prop:impossibility-prob-bounds}(\ref{itm:impossibility-prob-bound-ogp}), which establishes the main multi-OGP, in Section~\ref{sec:impossibility-multi-ogp}.
Finally, we will prove Proposition~\ref{prop:impossibility-prob-bounds}(\ref{itm:impossibility-prob-bound-consecutive}) in Section~\ref{sec:impossibility-ldp-stability}.
Let us first see how these results imply Theorem~\ref{thm:impossibility}.

\begin{proof}[Proof of Theorem~\ref{thm:impossibility}]
    Assume for sake of contradiction that there exists a (random) degree-$D$ polynomial $f:\bR^N \times \Omega \to \bR^n$ that $(\delta,\gamma,\eta,\nu)$-solves $\Phi_k(n,m)$.
    By Lemma~\ref{lem:impossibility-assume-deterministic}, there exists a deterministic degree-$D$ polynomial $g:\bR^N \to \bR^n$ that $(3\delta,3\gamma,\eta,\nu)$-solves $\Phi_k(n,m)$.
    We set $\kst = \kst(\kappa)$ large enough that Propositions~\ref{prop:impossibility-events-relation} and \ref{prop:impossibility-prob-bounds} both hold.
    By Proposition~\ref{prop:impossibility-prob-bounds} and a union bound,
    \[
        \P(\Svalid \cap \Sconsec \cap \Sindep \cap \Sogp)
        \ge
        (2n)^{-12 \gamma Dk^2/(\bp - \bm)}
        - 3(T+1) \delta
        - \exp(-\Omega(n)).
    \]
    We will show this probability is positive for suitable $C_1, C_2$.
    Let $C_2 = 2 + \f{12k^2}{\bp - \bm}$.
    Recall that $T = k^2 m = k^2 \lfloor \alpha n \rfloor$.
    If $\delta \le \exp(-C_2 \gamma D \log n)$, then $3(T+1) \delta \le \f13 (2n)^{-12 \gamma Dk^2/(\bp - \bm)}$ for sufficiently large $n$.
    Note that if $D \le \f{C_1 n}{\gamma \log n}$, then
    \[
        (2n)^{-12 \gamma Dk^2/(\bp - \bm)}
        \ge
        n^{-24 \gamma Dk^2/(\bp - \bm)}
        \ge
        \exp\lt(-\f{24 C_1 k^2}{\bp - \bm}n\rt).
    \]
    Let $C_1$ be small enough that the right-hand side is asymptotically larger than the $\exp(-\Omega(n))$ term.
    Thus for sufficiently large $n$, the $\exp(-\Omega(n))$ term is at most $\f13 (2n)^{-12 \gamma Dk^2/(\bp - \bm)}$.
    Therefore, there exists $\nst$ such that if $n\ge \nst$, then $\P(\Svalid \cap \Sconsec \cap \Sindep \cap \Sogp) > 0$.
    This implies that $\Svalid \cap \Sconsec \cap \Sindep \cap \Sogp \neq \emptyset$, contradicting Proposition~\ref{prop:impossibility-events-relation}.
\end{proof}

\subsection{Constructing the Forbidden Structure from Algorithm Outputs}
\label{subsec:impossibility-construct-ogp}

In this subsection, we will prove Proposition~\ref{prop:impossibility-events-relation}.
We will show that if $\Svalid$, $\Sconsec$, and $\Sindep$ all hold, then we can construct an example of the structure forbidden by $\Sogp$, and therefore all four events cannot hold simultaneously.

We will need the following auxiliary lemma, which shows that a small change of $x\in \{\T,\F\}^n$ in Hamming distance induces only a small change in $H(\pi(x | \yp{0},\ldots,\yp{\ell-1}))$.
This lemma allows us to convert $\Sconsec$ to a guarantee that consecutive conditional overlap entropies are small.
We defer the proof to Subsection~\ref{subsec:hamming-to-entropy}.

\begin{lemma}
    \label{lem:impossibility-hamming-to-entropy}
    Let $\ell \in \bN$ be arbitrary and let $x, x', \yp{0},\ldots,\yp{\ell-1} \in \{\T,\F\}^n$.
    If $\Delta(x,x') \le \f12$, then
    \[
        \lt|
            H\lt( \pi(x  | \yp{0},\ldots,\yp{\ell-1}) \rt) -
            H\lt( \pi(x' | \yp{0},\ldots,\yp{\ell-1}) \rt)
        \rt|
        \le
        H\lt( \Delta(x,x') \rt).
    \]
    The $H(\cdot)$ on the right denotes the binary entropy function.
\end{lemma}

\begin{proof}[Proof of Proposition~\ref{prop:impossibility-events-relation}]
    Set $k$ large enough that $\f{\bp - \bm}{2k} \le \f12$ and $H\lt(\f{\bp - \bm}{2k}\rt) \le (\bp - \bm)\f{\log k}{k}$.
    The second inequality holds for all sufficiently large $k$ due to the inequality $H(x) \le x \log \f{e}{x}$.

    Suppose that $\Svalid$, $\Sconsec$, and $\Sindep$ all hold.
    For $0\le \ell \le k$, let $\yp{\ell} = \xp{t_\ell}$, where $0 \le t_0 \le t_1 \le \cdots \le t_k \le T$ are defined as follows.
    Let $t_0 = 0$.
    For $1\le \ell\le k$, let $t_\ell$ be the smallest $t>t_{\ell-1}$ such that $H(\xp{t} | \yp{0},\ldots,\yp{\ell-1}) \in \lt[\bm \f{\log k}{k}, \bp \f{\log k}{k}\rt]$.
    We will show that such $t_\ell$ exists and satisfies $t_\ell \le t_{\ell-1} + km$.

    Let $t' = t_{\ell-1} + km$, and let $I = \{t_{\ell-1},t_{\ell-1}+1,\ldots,t'\}$.
    For $t\in I$, let $h(t) = H\lt( \pi(\xp{t} | \yp{0},\ldots,\yp{\ell-1}) \rt)$; we will analyze the evolution of $h(t)$ as we increment $t\in I$.
    By Fact~\ref{fac:overlap-profile-properties}(\ref{itm:overlap-profile-property-duplication}), $h(t_{\ell-1}) = 0$.

    In the definition of $\Sindep$, set $t_k = t'$ and $t_\ell, t_{\ell+1}, \ldots, t_{k-1}$ equal to $t_{\ell-1}$.
    By Fact~\ref{fac:overlap-profile-properties}(\ref{itm:overlap-profile-property-duplication}) (which allows us to ignore the duplicated $t_{\ell}, \ldots, t_{k-1}$), $\phip{t'}$ has no $\nu$-satisfying assignment $y$ with $H\lt( \pi(y | \yp{0},\ldots,\yp{\ell-1}) \rt) \le \bp \f{\log k}{k}$.
    But because $\Svalid$ holds, $\xp{t'}$ $\nu$-satisfies $\phip{t'}$.
    It follows that $h(t') > \bp \f{\log k}{k}$.

    Because $\Sconsec$ holds, we have $\Delta(\xp{t}, \xp{t-1}) \le \f{\bp - \bm}{2k}$ for all $t$.
    By Lemma~\ref{lem:impossibility-hamming-to-entropy},
    \begin{equation}
        \label{eq:h-small-move}
        |h(t) - h(t-1)|
        \le
        H\lt( \Delta(\xp{t}, \xp{t-1}) \rt)
        \le
        H\lt( \f{\bp - \bm}{2k} \rt)
        \le
        (\bp - \bm) \f{\log k}{k}.
    \end{equation}
    Since $h(t_{\ell-1}) = 0$ and $h(t') > \bp \f{\log k}{k}$, \eqref{eq:h-small-move} implies the existence of $t\in I$ such that $h(t) \in \lt[\bm \f{\log k}{k}, \bp \f{\log k}{k}\rt]$.
    So, $t_\ell$ is well defined and satisfies $t_\ell \le t_{\ell-1} + km$.

    Because the interpolation path has length $T = k^2 m$, and $t_\ell \le t_{\ell-1} + km$ for all $1\le \ell \le k$, this procedure sets all of $t_1,\ldots,t_k$ before the end of the interpolation.
    Finally, because $\Svalid$ holds, $y_\ell$ $\nu$-satisfies $\phip{t_\ell}$ for all $0\le \ell \le k$.
    We have thus constructed the structure forbidden by $\Sogp$.
\end{proof}

\subsection{Solutions to Independent Instances Contribute Large Overlap Entropy}
\label{subsec:impossibility-ogp-independent}

In this subsection, we will prove Proposition~\ref{prop:impossibility-prob-bounds}(\ref{itm:impossibility-prob-bound-independent}).
The proof is by a first moment argument.

\begin{proof}[Proof of Proposition~\ref{prop:impossibility-prob-bounds}(\ref{itm:impossibility-prob-bound-independent})]
    By Markov's inequality, $\P(\Sindep^c)$ is upper bounded by the expected number of $(t_0,\ldots,t_k,y)$ satisfying $0\le t_0 \le \cdots \le t_k \le T$, $t_k \ge t_{k-1} + km$, and  conditions (\ref{itm:def-sindep-satisfy}) and (\ref{itm:def-sindep-entropy}).
    There are at most $(T+1)^{k+1}$ possible choices of $(t_0,\ldots,t_k)$.
    Because $t_k \ge t_{k-1} + km$, $\phip{t_k}$ is independent of $\xp{t_0},\ldots,\xp{t_{k-1}}$.

    Let $P = P(\xp{t_0},\ldots,\xp{t_{k-1}})$ denote the set of all overlap profiles $\pi = \pi(\xp{t_0},\ldots,\xp{t_{k-1}}, y)$ over $y\in \{\T,\F\}^n$ with $H(\pi(y | \xp{t_0}, \ldots, \xp{t_{k-1}})) \le \bp \f{\log k}{k}$.
    By Fact~\ref{fac:overlap-profile-properties}(\ref{itm:overlap-profile-property-count}), $|P| \le n^{2^k}$.
    Thus,
    \begin{align*}
        \P(\Sindep^c)
        &\le
        (T+1)^{k+1}
        n^{2^k}
        \max_{\substack{0\le t_0\le \cdots \le t_k \le T \\ t_k \ge t_{k-1} + km}}
        \max_{\pi \in P} \\
        &\qquad \E_{\phip{t_k}}
        \# \lt(
            y\in \{\T,\F\}^n :
            \text{$y$ $\nu$-satisfies $\phip{t_k}$ and $\pi(\xp{t_0},\ldots,\xp{t_{k-1}},y) = \pi$}
        \rt)
    \end{align*}
    We can evaluate this inner expectation by linearity of expectation.
    The number of $y$ satisfying that $\pi(\xp{t_0},\ldots,\xp{t_{k-1}},y) = \pi$ is
    \begin{align*}
        \prod_{\{S,T\} \in \cPt(k)}
        \binom{\pi_{S,T} n}{\pi_{S\cup \{k\}, T}n}
        &=
        \exp\lt(
            n \sum_{\{S,T\} \in \cPt(k)}
            \pi_{S,T} H \lt(\f{\pi_{S\cup \{k\}, T}}{\pi_{S,T}}\rt)
            + o(n)
        \rt) \\
        &=
        \exp\lt(
            n H\lt(\pi (y | \xp{t_0},\ldots,\xp{t_{k-1}}) \rt)
            + o(n)
        \rt) \\
        &\le
        \exp\lt(
            n \bp \f{\log k}{k}
            + o(n)
        \rt).
    \end{align*}
    Recall that $\phip{t_k}$ is independent of $\xp{t_0},\ldots,\xp{t_{k-1}}$.
    Because $\sum_{j=0}^{\nu m} \binom{m}{j} \le (m+1)\binom{m}{\nu m}$, the probability that any one of these $y$ $\nu$-satisfies $\phip{t_k}$ is at most
    \begin{align*}
        \sum_{S\subseteq [m], |S| \le \nu m}
        (1 - 2^{-k})^{m-|S|}
        &\le
        (m+1)\binom{m}{\nu m} (1-2^{-k})^{(1-\nu)m} \\
        &\le
        \exp\lt(\nu m \log \f{e}{\nu} - (1-\nu) 2^{-k} m + o(n)\rt) \\
        &=
        \exp\lt(n\lt(-\kappa \f{\log k}{k} + o_k(1)\rt)+ o(n)\rt).
    \end{align*}
    Here we used that $\binom{a}{b} \le \lt(\f{ea}{b}\rt)^b$.
    Thus,
    \[
        \P(\Sindep^c)
        \le
        \exp\lt( - n(\kappa - \bp - o_k(1)) \f{\log k}{k} + o(n)\rt),
    \]
    where the $(T+1)^{k+1}n^{2^k}$ is absorbed in the $o(n)$.
    Finally, as
    \[
      \bp + \eps \le \f{\bp + \eps}{1 - \bp e^{-(\bp - 1)}} \le \kappa,
    \]
    we have $\kappa - \bp \ge \eps$.
    Thus $\P(\Sindep^c) = \exp(-\Omega(n))$ for sufficiently large $k$.
\end{proof}

\subsection{Small Hamming Distance Implies Small Conditional Overlap Entropy Difference}
\label{subsec:hamming-to-entropy}

We now present the deferred proof of Lemma~\ref{lem:impossibility-hamming-to-entropy}, which shows that a small change in $x\in \{\T,\F\}^n$ causes only a small change in $H(\pi(x | \yp{0},\ldots,\yp{\ell-1}))$.

\begin{proof}[Proof of Lemma~\ref{lem:impossibility-hamming-to-entropy}]
    For each partition $\{S,T\}\in \cPt(\ell)$, let
    \[
        I_{S,T} =
        \lt\{
            i\in [n]:
            \text{all $\{\yp{t}_i : t\in S\}$ equal one value and all $\{\yp{t}_i : t\in T\}$ equal the other value}
        \rt\}.
    \]
    Note that $|I_{S,T}| = \pi_{S,T} n$.
    If $\pi_{S,T}\neq 0$, define
    \[
        \lambda_{S,T} = \f{1}{|I_{S,T}|} \# \lt(i \in I_{S,T} : x_i = \T \rt)
        \qquad
        \text{and}
        \qquad
        \lambda'_{S,T} = \f{1}{|I_{S,T}|} \# \lt(i \in I_{S,T} : x'_i = \T \rt).
    \]
    (If $\pi_{S,T} = 0$, we can set these values arbitrarily in $[0,1]$.)
    On each index set $I_{S,T}$, $x$ and $x'$ differ in at least
    \[
        |I_{S,T}| \cdot |\lambda_{S,T} - \lambda'_{S,T}|
        =
        \pi_{S,T} |\lambda_{S,T} - \lambda'_{S,T}| n
    \]
    positions.
    Thus,
    \[
        \f12
        \ge
        \Delta(x,x')
        \ge
        \sum_{\{S,T\} \in \cPt(\ell)}
        \pi_{S,T} |\lambda_{S,T} - \lambda'_{S,T}|.
    \]
    Let $\sum_{\{S,T\} \in \cPt(\ell)} \pi_{S,T} |\lambda_{S,T} - \lambda'_{S,T}| = \mu$.
    Moreover, from the definition of conditional overlap entropy,
    \[
        H\lt( \pi(x  | \yp{0},\ldots,\yp{\ell-1}) \rt)
        =
        \sum_{\{S,T\} \in \cPt(\ell)}
        \pi_{S,T} H(\lambda_{S,T}),
    \]
    and analogously for $x'$.
    Note that $H(\cdot)$ is concave, so $H'(\cdot)$ is decreasing.
    Thus, for all $[a,b] \in [0,1]$ with $a \ge b$,
    \[
        H(a) - H(b)
        =
        \int_a^b H'(x)\diff{x}
        \le
        \int_0^{a-b} H'(x)\diff{x}
        = H(a-b).
    \]
    Similarly $H(1-b) - H(1-a) \le H(a-b)$, whence $|H(a)-H(b)| \le H(a-b)$.
    Thus,
    \begin{align*}
        \lt|
            H\lt( \pi(x  | \yp{0},\ldots,\yp{\ell-1}) \rt) -
            H\lt( \pi(x' | \yp{0},\ldots,\yp{\ell-1}) \rt)
        \rt|
        &\le
        \sum_{\{S,T\} \in \cPt(\ell)}
        \pi_{S,T}
        \lt| H(\lambda_{S,T}) - H(\lambda'_{S,T}) \rt| \\
        &\le
        \sum_{\{S,T\} \in \cPt(\ell)}
        \pi_{S,T}
        H\lt(|\lambda_{S,T} - \lambda'_{S,T}|\rt).
    \end{align*}
    By concavity of $H(\cdot)$, this last quantity has maximum value $H(\mu)$, attained when all the $|\lambda_{S,T} - \lambda'_{S,T}|$ are equal to $\mu$.
    Because $H(\cdot)$ is increasing on $[0,\f12]$ and $\mu \le \Delta(x,x') \le \f12$, we conclude that
    \[
        \lt|
            H\lt( \pi(x  | \yp{0},\ldots,\yp{\ell-1}) \rt) -
            H\lt( \pi(x' | \yp{0},\ldots,\yp{\ell-1}) \rt)
        \rt|
        \le
        H(\mu)
        \le
        H\lt( \Delta(x,x') \rt).
    \]
\end{proof}

%% file: tex/5-multi-ogp.tex
\section{Proof of Presence of Ensemble Multi-OGP}
\label{sec:impossibility-multi-ogp}

In this section, we will prove Proposition~\ref{prop:impossibility-prob-bounds}(\ref{itm:impossibility-prob-bound-ogp}), which shows that the forbidden structure in $\Sogp$ does not occur with high probability.

\subsection{Proof Outline}

We first give a high level overview of the proof, which is by another first moment computation.
Throughout this section, for $I\in [n]^k$ and $x\in \{\T,\F\}^n$, let $x[I] = (x_{I_1}, \ldots, x_{I_k})$ be the subsequence of $x$ indexed by $I$.
We begin with the following lemma, which bounds the exponential rate of $\P(\Sogp^c)$ in terms of a maximum over overlap profiles.
We will prove this lemma in Subsection~\ref{subsec:impossibility-sogp-exp-rate-proof}.

\begin{lemma}
    \label{lem:impossibility-sogp-exp-rate}
    Let $P$ denote the set of overlap profiles $\pi = \pi(\yp{0}, \ldots, \yp{k})$ over $\yp{0}, \ldots, \yp{k} \in \{\T,\F\}^n$ satisfying that for all $1\le \ell \le k$, $H\lt( \pi(\yp{\ell} | \yp{0}, \ldots, \yp{\ell-1}) \rt) \in \lt[\bm \f{\log k}{k}, \bp \f{\log k}{k}\rt]$.
    Then,
    \begin{equation}
        \label{eq:impossibility-sogp-exp-rate}
        \f{1}{n}
        \log \P(\Sogp^c)
        \le
        \log 2 +
        \max_{\pi \in P}
        \lt[
            H(\pi) -
            \kappa \f{\log k}{k}
            \E_{I\sim \unif\lt([n]^k\rt)}
            \lt|\lt\{\yp{\ell}[I]: 0\le \ell\le k\rt\}\rt|
        \rt] + o_k(1) + o(1),
    \end{equation}
    where in the expectation, $\yp{0}, \ldots, \yp{k} \in \{\T,\F\}^n$ is a sequence of assignments with overlap profile $\pi$.
\end{lemma}
Note that the expectation $\E_{I\sim \unif\lt([n]^k\rt)} \lt|\lt\{\yp{\ell}[I]: 0\le \ell\le k\rt\}\rt|$ has the same value for any $\yp{0}, \ldots, \yp{k} \in \{\T,\F\}^n$ with overlap profile $\pi$.
So, the quantity inside the maximum is a function of $\pi$.

The negative term in \eqref{eq:impossibility-sogp-exp-rate} arises as an upper bound on the exponential rate of the probability that $\yp{0}, \ldots, \yp{k}$ all respectively $\nu$-satisfy $\phip{t_0}, \ldots, \phip{t_k}$, for fixed $\yp{0}, \ldots, \yp{k}$ and $t_0, \ldots, t_k$.
Let us first argue heuristically that this bounds the exponential rate; we will formalize this reasoning in Lemma~\ref{lem:impossibility-sogp-exp-rate-probability-term} below.
We expect this probability to be maximized when $t_0 = \cdots = t_k$, because making the $t_i$ different only introduces additional randomness (see Remark~\ref{rem:ensemble-commentary}).
So, let $\phip{t_0}, \ldots, \phip{t_k}$ all equal the same $k$-SAT instance $\Phi \sim \Phi_k(n,m)$.
We also focus on the probability that $\yp{0}, \ldots, \yp{k}$ all \emph{satisfy} $\Phi$; we will see that the relaxation to $\nu$-satisfy only changes the exponential rate by $o_k(1)$.
The probability that $\yp{0}, \ldots, \yp{k}$ all satisfy the first clause $\Phi_1$ is $1 - 2^{-k} \E_{I\sim \unif\lt([n]^k\rt)} \lt|\lt\{\yp{\ell}[I]: 0\le \ell\le k\rt\}\rt|$, because if $\Phi_1$ contains the variables $x_{I_1}, \ldots, x_{I_k}$, there are $\lt|\lt\{\yp{\ell}[I]: 0\le \ell\le k\rt\}\rt|$ ways to set these variables' polarities in $\Phi_1$ so that one of $\yp{0}, \ldots, \yp{k}$ does not satisfy $\Phi_1$.
Then, the probability that $\yp{0}, \ldots, \yp{k}$ all satisfy $\Phi$ is upper bounded by
\[
    \lt(
        1 - 2^{-k}
        \E_{I\sim \unif\lt([n]^k\rt)}
        \lt|\lt\{\yp{\ell}[I]: 0\le \ell\le k\rt\}\rt|
    \rt)^m
    \le
    \exp\lt(
        -\f{m}{2^k}
        \E_{I\sim \unif\lt([n]^k\rt)}
        \lt|\lt\{\yp{\ell}[I]: 0\le \ell\le k\rt\}\rt|
    \rt),
\]
and $\f{m}{2^k} \approx n\kappa \f{\log k}{k}$.
The second ingredient in the proof of Proposition~\ref{prop:impossibility-prob-bounds}(\ref{itm:impossibility-prob-bound-ogp}) is the following proposition, which lower bounds the expectation in the negative term of \eqref{eq:impossibility-sogp-exp-rate}.
We will prove this proposition in Subsection~\ref{subsec:impossibility-energy-lb}.
Proving the bound in this proposition is one of the main technical challenges of this paper, which we overcome via a surprising probabilistic reformulation of the left-hand expectation.

\begin{proposition}
    \label{prop:impossibility-ogp-energy-term}
    Let $\beta_1,\ldots,\beta_k \in [\bm, \bp]$, and let $\yp{0}, \ldots, \yp{k} \in \{\T,\F\}^n$ be assignments satisfying that $H\lt( \pi(\yp{\ell} | \yp{0}, \ldots, \yp{\ell-1}) \rt) = \beta_\ell \f{\log k}{k}$.
    Then,
    \[
        \E_{I\sim \unif([n]^k)}
        \lt|\lt\{\yp{\ell}[I] : 0\le \ell \le k\rt\}\rt|
        \ge
        (1 - o_k(1))
        \sum_{\ell=1}^k \lt(1 - \beta_\ell e^{-(\beta_\ell-1)}\rt).
    \]
\end{proposition}

From Lemma~\ref{lem:impossibility-sogp-exp-rate} and Proposition~\ref{prop:impossibility-ogp-energy-term}, we can see the main ideas of the proof of Proposition~\ref{prop:impossibility-prob-bounds}(\ref{itm:impossibility-prob-bound-ogp}) and understand the motivation of the definition of $\Sogp$.
The ideas are as follows.

We will prove Proposition~\ref{prop:impossibility-prob-bounds}(\ref{itm:impossibility-prob-bound-ogp}) by showing that the right-hand side of \eqref{eq:impossibility-sogp-exp-rate} is negative.
For each $\pi\in P$, this quantity can be regarded as a free entropy, with entropy term $\log 2 + H(\pi)$ and energy term $\kappa \f{\log k}{k} \E_{I\sim \unif\lt([n]^k\rt)} \lt|\lt\{\yp{\ell}[I]:  0\le \ell\le k\rt\}\rt|$.
This free entropy exhibits a tradeoff where as the entropy term increases, the assignments $\yp{0}, \ldots, \yp{k}$ become more diverse, and so the energy term increases too.
The event $\Sogp$ is selected so that for overlap profiles $\pi\in P$, where $P$ is defined in Lemma~\ref{lem:impossibility-sogp-exp-rate}, the energy term is larger than the entropy term, which makes the free entropy negative.
In particular, (due to Fact~\ref{fac:overlap-profile-properties}(\ref{itm:overlap-profile-property-chainrule})) we think of $H\lt( \pi(\yp{\ell} | \yp{0}, \ldots, \yp{\ell-1}) \rt)$ as the amount that $\yp{\ell}$ contributes to the entropy term.
Given this contribution, Proposition~\ref{prop:impossibility-ogp-energy-term} lower bounds the amount that $\yp{\ell}$ contributes to the energy term.
In the definition of $\Sogp$, we require the entropy contribution to be in a medium range $\lt[\bm \f{\log k}{k}, \bp \f{\log k}{k}\rt]$ because (recalling the definition of $\bm, \bp$) in this range the energy-to-entropy ratio is favorable to the energy term.
Specifically, we show that if $\yp{\ell}$ contributes an entropy in this range, the energy it contributes is at least $\eps \f{\log k}{k}$ more.
Thus each $\yp{\ell}$ decreases the free entropy by at least $\eps \f{\log k}{k}$.
Together, the $k$ assignments $\yp{1}, \ldots, \yp{k}$ contribute a free entropy decrease of $\eps \log k$, which dominates the starting free entropy of $\log 2$ and makes the overall free entropy negative.

We now prove Proposition~\ref{prop:impossibility-prob-bounds}(\ref{itm:impossibility-prob-bound-ogp}) given Lemma~\ref{lem:impossibility-sogp-exp-rate} and Proposition~\ref{prop:impossibility-ogp-energy-term}.

\begin{proof}[Proof of Proposition~\ref{prop:impossibility-prob-bounds}(\ref{itm:impossibility-prob-bound-ogp})]
    Let $P$ be as in Lemma~\ref{lem:impossibility-sogp-exp-rate}.
    Let $\pi \in P$, and consider assignments $\yp{0}, \ldots, \yp{k} \in \{\T,\F\}^n$ with $\pi(\yp{0},\ldots,\yp{k}) = \pi$.
    For $1\le \ell \le k$, define $\beta_\ell$ by $H\lt(\pi(\yp{\ell} | \yp{0}, \ldots, \yp{\ell-1})\rt) = \beta_\ell \f{\log k}{k}$.
    Note that the $\beta_\ell$ are determined given $\pi$ and satisfy $\beta_1,\ldots,\beta_k\in [\bm, \bp]$.
    By Fact~\ref{fac:overlap-profile-properties}(\ref{itm:overlap-profile-property-chainrule}), $H(\pi) = \f{\log k}{k} \sum_{\ell=1}^k \beta_\ell$.
    By Proposition~\ref{prop:impossibility-ogp-energy-term},
    \begin{align*}
        -\kappa \f{\log k}{k} \E_{I\sim \unif([n]^k)} \lt|\lt\{\yp{\ell}[I] : 0\le \ell \le k\rt\}\rt|
        &\le
        -(1 - o_k(1)) \f{\log k}{k}
        \sum_{\ell=1}^k
        \kappa \lt(1 - \beta_\ell e^{-(\beta_\ell-1)}\rt) \\
        &\le
        -(1 - o_k(1)) \f{\log k}{k}
        \sum_{\ell=1}^k
        \lt(\beta_\ell + \eps\rt).
    \end{align*}
    The last inequality uses that $\f{\beta + \eps}{1 - \beta e^{-(\beta - 1)}} \le \kappa$ for all $\beta \in [\bm, \bp]$.
    Therefore,
    \begin{align*}
        H(\pi) - \kappa \f{\log k}{k} \E_{I\sim \unif([n]^k)} \lt|\lt\{\yp{\ell}[I] : 0\le \ell \le k\rt\}\rt|
        &\le
        \f{\log k}{k}
        \sum_{\ell=1}^k \beta_\ell
        -(1 - o_k(1)) \f{\log k}{k}
        \sum_{\ell=1}^k
        \lt(\beta_\ell + \eps\rt) \\
        &\le
        o_k(1) \f{\log k}{k} \sum_{\ell=1}^k \beta_\ell
        - (1-o_k(1)) \eps \log k \\
        &\le
        o_k(1) \bp \log k - (1 - o_k(1)) \eps \log k \\
        &=
        - (1 - o_k(1)) \eps \log k.
    \end{align*}
    This bound holds for an arbitrary $\pi \in P$, and thus for the maximum over $\pi \in P$.
    By Lemma~\ref{lem:impossibility-sogp-exp-rate},
    \[
        \f{1}{n} \log \P(\Sogp^c)
        \le
        \log 2 - (1 - o_k(1)) \eps \log k + o_k(1) + o(1) < 0
    \]
    for sufficiently large $k$ and $n$.
    Thus $\P(\Sogp^c) \le \exp(-\Omega(n))$.
\end{proof}

\subsection{Bounding the Exponential Rate by a Free Entropy}
\label{subsec:impossibility-sogp-exp-rate-proof}

In this subsection, we will prove Lemma~\ref{lem:impossibility-sogp-exp-rate}.
We begin with the following lemma, which bounds the probability term arising in the first moment upper bound of $\P(\Sogp^c)$.
\begin{lemma}
    \label{lem:impossibility-sogp-exp-rate-probability-term}
    Suppose $\yp{0}, \ldots, \yp{k} \in \{\T,\F\}^n$ is a sequence of assignments and $0\le t_0\le t_1\le \cdots \le t_k \le T$.
    Then,
    \[
        \f{1}{n}
        \log
        \P \lt[
            \text{$\yp{\ell}$ $\nu$-satisfies $\phip{t_\ell}$ for all $0\le \ell\le k$}
        \rt]
        \le
        -\kappa \f{\log k}{k}
        \E_{I\sim \unif([n]^k)}
        \lt|\lt\{\yp{\ell}[I] : 0\le \ell \le k\rt\}\rt|
        + o_k(1) + o(1).
    \]
\end{lemma}
\begin{proof}
    Say a clause index $i\in [m]$ is \vocab{interrupted} if for some $0\le \ell \le k$, $t_\ell$ satisfies $1\le \sigma(t_\ell) - (i-1)k \le k-1$, where $\sigma(\cdot)$ is defined in Definition~\ref{defn:interpolation}.
    Informally, $i$ is interrupted if there is some $\ell$ such that $\phip{t_\ell}$ is partway through resampling the $i$th clause.
    Let $U$ denote the set of interrupted clause indices.
    Note that each $t_\ell$ interrupts at most one clause, so $|U| \le k+1$.

    Say a clause index $i\in [m]$ is \vocab{bad} if some $\yp{\ell}$ fails to satisfy $\phip{t_\ell}_i$.
    Let $S$ denote the set of bad clause indices.
    If $\yp{\ell}$ $\nu$-satisfies $\phip{t_\ell}$ for all $0\le \ell \le k$, then each $\yp{\ell}$ fails to satisfy at most $\nu m$ clauses of $\phip{t_\ell}$, so $|S| \le (k+1)\nu m$.

    We will see that because so few clause indices are interrupted or bad, it does not hurt our analysis to throw them out.
    We have that
    \begin{align}
        \notag
        &\P \lt[
            \text{$\yp{\ell}$ satisfies $\phip{t_\ell}$ for all $0\le \ell\le k$}
        \rt] \\
        \notag
        &\le
        \sum_{S\subseteq [m], |S| \le (k+1) \nu m}
        \P \lt[
            \text{$\yp{\ell}$ satisfies $\phip{t_\ell}_i$ for all $0\le \ell\le k$, $i\in [m]\setminus S$}
        \rt] \\
        \notag
        &\le
        (m+1) \binom{m}{(k+1)\nu m}
        \max_{S\subseteq [m], |S| \le (k+1) \nu m}
        \P \lt[
            \text{$\yp{\ell}$ satisfies $\phip{t_\ell}_i$ for all $0\le \ell\le k$, $i\in [m]\setminus (S\cup U)$}
        \rt] \\
        \label{eq:impossibility-sat-ensemble-to-clause-ensemble}
        &=
        (m+1) \binom{m}{(k+1)\nu m}
        \max_{S\subseteq [m], |S| \le (k+1) \nu m}
        \prod_{i\in [m]\setminus (S\cup U)}
        \P \lt[
            \text{$\yp{\ell}$ satisfies $\phip{t_\ell}_i$ for all $0\le \ell\le k$}
        \rt].
    \end{align}
    The last step uses that over $i \in [m]$, the collections of clauses $\{\phip{t}_i: 0\le t\le T\}$ are mutually independent.

    We now fix a single $i\in [m] \setminus (S\cup U)$ and analyze the last probability.
    We exploit the following stochastic property of non-interrupted clauses: if $i$ is not interrupted, then the clauses $\phip{t_0}_i, \phip{t_1}_i,\ldots, \phip{t_k}_i$ can be partitioned into equivalence classes, such that all clauses in the same equivalence class are identical and all clauses in different equivalence classes are mutually independent.
    Formally, for some $1\le r\le k+1$, there is a surjective map $\tau : \{0,\ldots,k\} \to [r]$ (dependent only on the indices $t_0,\ldots,t_k$ and $i$) such that for i.i.d. clauses $C_1,\ldots,C_r \sim \Phi_k(n,1)$,
    \[
        \lt(\phip{t_0}_i, \phip{t_1}_i,\ldots, \phip{t_k}_i\rt)
        =_d
        \lt(C_{\tau(0)}, C_{\tau(1)}, \ldots, C_{\tau(k)}\rt).
    \]
    For $1\le s\le r$, let $B_s = \tau^{-1}(s)$ be the set of $\ell\in \{0,\ldots,k\}$ such that $\phip{t_\ell}_i$ corresponds to $C_s$.
    Thus $B_1, \ldots, B_r$ partition $\{0,\ldots,k\}$.
    Now,
    \begin{equation}
        \label{eq:impossibility-clause-ensemble-to-clause}
        \P \lt[
            \text{$\yp{\ell}$ satisfies $\phip{t_\ell}_i$ for all $0\le \ell\le k$}
        \rt]
        =
        \prod_{s=1}^r
        \P \lt[
            \text{$\yp{\ell}$ satisfies $C_s$ for all $\ell \in B_s$}
        \rt].
    \end{equation}
    Let $I \in [n]^k$ be the indices of the $k$ variables sampled by $C_s$, so $I \sim \unif([n]^k)$.
    Given $I$, there are $\lt|\lt\{\yp{\ell}[I] : \ell \in B_s\rt\}\rt|$ ways to assign polarities to these $k$ variables such that for some $\ell \in B_s$, $\yp{\ell}$ does not satisfy $C_s$.
    Thus, conditioned on $I$, the probability that $\yp{\ell}$ satisfies $C_s$ for all $\ell \in B_s$ is $1 - 2^{-k} \lt|\lt\{\yp{\ell}[I] : \ell \in B_s\rt\}\rt|$.
    It follows that
    \begin{align*}
        \P \lt[
            \text{$\yp{\ell}$ satisfies $C_s$ for all $\ell \in B_s$}
        \rt]
        &=
        1 - 2^{-k}
        \E_{I\sim \unif([n]^k)}
        \lt|\lt\{\yp{\ell}[I] : \ell \in B_s\rt\}\rt| \\
        &\le
        \exp\lt(
            - 2^{-k}
            \E_{I\sim \unif([n]^k)}
            \lt|\lt\{\yp{\ell}[I] : \ell \in B_s\rt\}\rt|
        \rt).
    \end{align*}
    So, using \eqref{eq:impossibility-clause-ensemble-to-clause} and recalling that $B_1,\ldots,B_r$ partition $\{0,\ldots,k\}$, we have
    \begin{align*}
        \P \lt[
            \text{$\yp{\ell}$ satisfies $\phip{t_\ell}_i$ for all $0\le \ell\le k$}
        \rt]
        &\le
        \exp\lt(
            - 2^{-k}
            \E_{I\sim \unif([n]^k)}
            \sum_{s=1}^r
            \lt|\lt\{\yp{\ell}[I] : \ell \in B_s\rt\}\rt|
        \rt) \\
        &\le
        \exp\lt(
            - 2^{-k}
            \E_{I\sim \unif([n]^k)}
            \lt|\lt\{\yp{\ell}[I] : 0\le \ell \le k\rt\}\rt|
        \rt).
    \end{align*}
    Next, we substitute into \eqref{eq:impossibility-sat-ensemble-to-clause-ensemble}.
    Since $|S| \le (k+1)\nu m$, $|U| \le k+1$, and $m = \lfloor \alpha n \rfloor \ge \alpha n - 1$,
    \begin{align*}
      |m \setminus (S \cup U)|
      &\ge
      (1-(k+1)\nu) m - (k+1) \\
      &\ge
      (1-(k+1)\nu) \alpha n - k-2.
    \end{align*}
    So,
    \begin{align*}
        &\P \lt[
            \text{$\yp{\ell}$ satisfies $\phip{t_\ell}$ for all $0\le \ell\le k$}
        \rt] \\
        &\le
        (m+1) \binom{m}{(k+1)\nu m}
        \exp\lt(
            -\f{(1-(k+1)\nu) \alpha n - k-2}{2^k}
            \E_{I\sim \unif([n]^k)}
            \lt|\lt\{\yp{\ell}[I] : 0\le \ell \le k\rt\}\rt|
            +o(n)
        \rt).
    \end{align*}
    Thus, using that $\binom{a}{b} \le \lt(\f{ea}{b}\rt)^b$,
    \begin{align*}
        &\f{1}{n}
        \log
        \P \lt[
            \text{$\yp{\ell}$ satisfies $\phip{t_\ell}$ for all $0\le \ell\le k$}
        \rt] \\
        &\le
        (k+1)\nu \alpha \log\f{e}{(k+1)\nu}
        -\f{(1-(k+1)\nu) \alpha}{2^k}
        \E_{I\sim \unif([n]^k)}
        \lt|\lt\{\yp{\ell}[I] : 0\le \ell \le k\rt\}\rt|
        + o(1) \\
        &\le
        -\f{\alpha}{2^k}
        \E_{I\sim \unif([n]^k)}
        \lt|\lt\{\yp{\ell}[I] : 0\le \ell \le k\rt\}\rt|
        + o_k(1) + o(1).
    \end{align*}
    The result follows from $\alpha = \kappa 2^k \log k / k$.
\end{proof}

\begin{proof}[Proof of Lemma~\ref{lem:impossibility-sogp-exp-rate}]
    By Markov's inequality, $\P(\Sogp^c)$ is upper bounded by the expected number of $0 \le t_0 \le t_1 \le \cdots \le t_k \le T$ and $(\yp{0},\ldots,\yp{k})$ satisfying conditions (\ref{itm:def-sogp-satisfy}) and (\ref{itm:def-sogp-entropy}).
    There are at most $(T+1)^{k+1}$ choices of $(t_0,\ldots,t_k)$, and (by Fact~\ref{fac:overlap-profile-properties}(\ref{itm:overlap-profile-property-count})) $|P| \le n^{2^k}$.
    By linearity of expectation,
    \begin{align*}
        \P(\Sogp^c)
        &\le
        (T+1)^{k+1}
        n^{2^k}
        \max_{\substack{0\le t_0\le \cdots \le t_k \le T \\ \pi \in P}}
        \E
        \lt[ \# \lt(
            \begin{array}{l}
                (\yp{0}, \ldots, \yp{k}) \in \{\T,\F\}^{n\times (k+1)}: \\
                \text{$\yp{\ell}$ $\nu$-satisfies $\phip{t_\ell}$ for all $0\le \ell\le k$} \\
                \text{and~} \pi(\yp{0}, \ldots, \yp{k}) = \pi
            \end{array}
        \rt) \rt].
    \end{align*}
    Let $\pi n$ be the scalar product of $\pi$, treated as a vector, by $n$.
    There are $2^n \binom{n}{\pi n}$ sequences of assignments $(\yp{0}, \ldots, \yp{k})$ with $\pi(\yp{0}, \ldots, \yp{k}) = \pi$: $2^n$ ways to choose $\yp{0}$, and then $\binom{n}{\pi n}$ ways to assign the positions $[n]$ to the partitions of $\{0,\ldots,k\}$.
    Over all of these sequences of assignments, the probability of the event that $\yp{\ell}$ satisfies $\phip{t_\ell}$ for all $0\le \ell\le k$ is uniformly upper bounded by Lemma~\ref{lem:impossibility-sogp-exp-rate-probability-term}.
    By linearity of expectation, the last expectation is upper bounded by
    \[
        2^n \binom{n}{\pi n}
        \exp\lt(
            -n \lt(
              \kappa \f{\log k}{k}
              \E_{I\sim \unif([n]^k)}
              \lt|\lt\{\yp{\ell}[I] : 0\le \ell \le k\rt\}\rt|
              + o_k(1) + o(1)
            \rt)
        \rt).
    \]
    Because $\binom{n}{\pi n} = \exp\lt(n(H(\pi) + o(1))\rt)$, the result follows.
\end{proof}

\begin{remark}
    \label{rem:ensemble-commentary}
    The step in the proof of Lemma~\ref{lem:impossibility-sogp-exp-rate-probability-term} where we lower bound $\sum_{s=1}^r \lt|\lt\{\yp{\ell}[I] : \ell \in B_s\rt\}\rt|$ by $\lt|\lt\{\yp{\ell}[I] : 0\le \ell \le k\rt\}\rt|$ is tight when $t_0,\ldots,t_k$ are all equal, because in this case $r=1$ and $B_1 = \{0,1,\ldots,k\}$.
    Thus the exponential rate of $\P(\Sogp^c)$ is dominated by the case when the $t_i$ are equal.
    In other words, $\P(\Sogp^c)$ has the same exponential rate as if, in the definition of $\Sogp$, we required all the $\yp{\ell}$ to $\nu$-satisfy the same $\phip{t}$.
    This shows the power of the ``ensemble" part of the ensemble multi-OGP: for no cost in the exponential rate, we can generalize the forbidden structure to an ensemble.
    All ensemble (multi-)OGPs in the literature share and leverage this property, see \cite{GJW20, Wei20}.
\end{remark}

\subsection{Lower Bounding the Energy Term}
\label{subsec:impossibility-energy-lb}

In this subsection, we will prove Proposition~\ref{prop:impossibility-ogp-energy-term}.
Let $\yp{0}, \ldots, \yp{k}$ and $\beta_1,\ldots,\beta_k$ be as in Proposition~\ref{prop:impossibility-ogp-energy-term}.
Without loss of generality, we can set $\yp{0} = \T^n$.

To analyze the expectation in Proposition~\ref{prop:impossibility-ogp-energy-term}, we introduce the following probabilistic quantities.
For $0\le \ell \le k$ and $\sigma \in \{\T,\F\}^k$, define
\[
    E_{\ell}(\sigma)
    =
    \lt\{I \in [n]^k : \yp{\ell'}[I] = \sigma \text{~for some~} 0\le \ell' \le \ell\rt\}
    \qquad
    \text{and}
    \qquad
    p_{\ell}(\sigma)
    =
    \P_{I \sim \unif([n]^k)} \lt(E_{\ell}(\sigma)\rt).
\]
In other words, $E_{\ell}(\sigma)$ is the event that $\sigma$ appears in the set $\lt\{\yp{\ell'}[I] : 0\le \ell' \le \ell\rt\}$, and $p_{\ell}(\sigma)$ is the probability of this event.
The probabilities $p_k(\sigma)$ will be relevant to our analysis by the following identity \eqref{eq:impossibility-ogp-energy-term-to-probabilities}, while the probabilities $p_\ell(\sigma)$ for $\ell < k$ will arise in our inductive analysis below, where we lower bound $p_k(\sigma)$ by peeling off one of $\yp{1},\ldots,\yp{k}$ at a time.
We have that
\begin{equation}
    \label{eq:impossibility-ogp-energy-term-to-probabilities}
    \E_{I\sim \unif([n]^k)} \lt|\lt\{\yp{\ell}[I] : 0\le \ell \le k\rt\}\rt|
    =
    \E_{I\sim \unif([n]^k)}
    \lt[
        \sum_{\sigma \in \{\T,\F\}^k}
        \ind{I\in E_k(\sigma)}
    \rt]
    =
    \sum_{\sigma \in \{\T,\F\}^k}
    p_k(\sigma).
\end{equation}

To prove Proposition~\ref{prop:impossibility-ogp-energy-term}, we will need to lower bound the right-hand side of \eqref{eq:impossibility-ogp-energy-term-to-probabilities}.
This task will require several definitions; to motivate these definitions, we first outline our technique for deriving this lower bound.

Our first step is a conditional expansion.
Let $I \sim \unif([n]^k)$.
We reveal the $k$ bit strings $\yp{1}[I], \ldots, \yp{k}[I]$ one by one.
(Recall that we fixed $\yp{0}=\T^n$, so $\yp{0}[I]=\T^k$ is known.)
Conditioned on its predecessors $\yp{1}[I], \ldots, \yp{\ell-1}[I]$, the distribution of $\yp{\ell}[I]$ can be described in terms of the conditional overlap profile $\pi(\yp{\ell} | \yp{0}, \ldots, \yp{\ell-1})$.
Then, $1 - p_k(\sigma)$, the probability that $\sigma$ does not appear in $\lt\{\yp{\ell}[I] : 0\le \ell \le k\rt\}$, can be expanded as a product of $k$ factors: the $\ell$th factor is the conditional probability that the revealed value of $\yp{\ell}$ does not equal $\sigma$.
The $\ell$th factor of this product can be thought of as ($1$ minus) the contribution of $\yp{\ell}$ to $p_k(\sigma)$.

Our second step is to estimate this product by a \emph{sum}, whose $\ell$th summand is the contribution of $\yp{\ell}$ to this estimate of $p_k(\sigma)$.
The purpose of this estimation is to decouple the contributions of the $\yp{\ell}$, so that we can analyze the overall contribution of $\yp{\ell}$ by summing over $\sigma \in \{\T,\F\}^k$.
We achieve this by truncating the factors in the product at $1 - \f{1}{k\log k}$; any factor smaller than this gets rounded up to $1$.
Because $\f{1}{k\log k} \ll \f{1}{k}$, we can separate the contributions of $\yp{1},\ldots,\yp{k}$ to $p_k(\sigma)$ by the estimate
\[
    1 - (1-\eps_1)(1-\eps_2)\cdots (1-\eps_k)
    \approx
    \eps_1 + \eps_2 + \cdots + \eps_k
\]
up to $1-o_k(1)$ multiplicative error.
Propositions~\ref{prop:impossibility-ogp-decouple-joint-event-aux} and \ref{prop:impossibility-ogp-decouple-joint-event} below carry out this technique.

Finally, our third step is to collect the (now additive) contributions of each $\yp{\ell}$ to the estimated $p_k(\sigma)$ over all $\sigma \in \{\T,\F\}^k$.
Miraculously, we can interpret this sum as a probability of a sum of $k$ i.i.d. random variables, which can be controlled by a Chernoff bound.
This step is carried out in Proposition~\ref{prop:impossibility-ogp-dart-game}.

Formally, for $0\le \ell \le k$ and $i\in [n]$, let $\yp{\le \ell}_i = (\yp{1}_i,\ldots,\yp{\ell}_i)$.
Similarly, for $I\in [n]^k$, let $\yp{\le \ell}[I] = (\yp{1}[I],\ldots,\yp{\ell}[I])$.
Because $\yp{0} = \T^n$, the overlap profile $\pi$ determines the distribution of $\yp{\le k}_i$ over $i \sim \unif([n])$.
Namely, for $\xi \in \{\T,\F\}^k$,
\[
    \P_{i\sim \unif([n])}
    \lt[
        \yp{\le k}_i = \xi
    \rt]
    =
    \pi_{S \cup \{0\},T}
\]
where $S = \{\ell \in [k]: \xi_\ell = \T\}$ and $T = \{\ell \in [k]: \xi_\ell = \F\}$.
Moreover, the distribution of $\yp{\le k}[I]$, where $I \sim \unif([n]^k)$, is the product of $k$ i.i.d. copies of this distribution.
For $1\le \ell \le k$, $b\in \{\T,\F\}$, and $\xi \in \{\T,\F\}^{\ell-1}$, define
\[
    \phi_\ell (b | \xi)
    =
    \P_{i\sim \unif([n])}
    \lt[
        \yp{\ell}_i = b |
        \yp{\le \ell-1}_i = \xi
    \rt].
\]
The probabilities in the aforementioned conditional expansion are products of conditional probabilities $\phi_\ell(b | \xi)$.
Namely, the probability that $\yp{\ell}[I] \neq \sigma$ given $\yp{\le \ell-1}[I]$ is $1-\prod_{r=1}^k \phi_\ell (\sigma_{r} | \yp{\le \ell-1}_{I_r})$.

For $1 \le \ell \le k$, $\sigma \in \{\T,\F\}^k$ and $I\in [n]^k$, further define
\[
    Q_\ell(\sigma, I) =
    \lt(
        \prod_{r=1}^k
        \phi_\ell (\sigma_{r} | \yp{\le \ell-1}_{I_r})
    \rt)
    \ind{
        \prod_{r=1}^k
        \phi_\ell (\sigma_{r} | \yp{\le \ell-1}_{I_r})
        \le
        \f{1}{k\log k}
    }
    \quad
    \text{and}
    \quad
    q_\ell(\sigma)
    =
    \E_{I\sim \unif([n]^k)}
    \lt[
    Q_\ell(\sigma,I)
    \rt].
\]
Thus, $1 - Q_\ell(\sigma, I)$ is a term in the conditional expansion, truncated at $1 - \f{1}{k\log k}$ in the aforementioned sense, and $q_\ell(\sigma)$ is its expectation.

For each $\sigma \in \{\T,\F\}^k$, the following two propositions lower bound $p_k(\sigma)$ in terms of $q_1(\sigma),\ldots,q_k(\sigma)$ by peeling off one of $\yp{1},\ldots,\yp{k}$ at a time.

\begin{proposition}
    \label{prop:impossibility-ogp-decouple-joint-event-aux}
    For each $\sigma \in \{\T,\F\}^k$ and $1\le \ell \le k$, we have that
    \[
        p_{\ell}(\sigma) \ge \lt(1 - \f{1}{k\log k}\rt) p_{\ell-1}(\sigma) + q_{\ell}(\sigma).
    \]
\end{proposition}
\begin{proof}
    Note that
    \begin{align*}
        1 - p_{\ell}(\sigma)
        &=
        \P_{I\sim \unif([n]^k)}
        \lt[
            \yp{\ell'}[I] \neq \sigma \text{~for all~} 0\le \ell' \le \ell
        \rt] \\
        &=
        \E_{I\sim \unif([n]^k)}
        \lt[
            \ind{
                \yp{\ell'}[I] \neq \sigma \text{~for all~} 0\le \ell' \le \ell-1
            }
            \lt(
                1 -
                \prod_{r=1}^k
                \phi_\ell (\sigma_{r} | \yp{\le \ell-1}_{I_r})
            \rt)
        \rt].
    \end{align*}
    Here, we use that the event inside the indicator is $\yp{\le \ell-1}[I]$-measurable, and conditioned on $\yp{\le \ell-1}[I]$ the probability that $\yp{\ell}[I] = \sigma$ is $\prod_{r=1}^k \phi_\ell (\sigma_{r} | \yp{\le \ell-1}_{I_r})$.
    Moreover, we have $\prod_{r=1}^k \phi_\ell (\sigma_{r} | \yp{\le \ell-1}_{I_r}) \ge Q_\ell(\sigma, I)$ by definition.
    So,
    \begin{align*}
        1 - p_{\ell}(\sigma)
        &\le
        \E_{I\sim \unif([n]^k)}
        \lt[
            \lt(1 - \ind{I \in E_{\ell-1}(\sigma)}\rt)
            \lt(
                1 -
                Q_\ell(\sigma, I)
            \rt)
        \rt] \\
        &\le
        \E_{I\sim \unif([n]^k)}
        \lt[
            1 - \lt(1 - \f{1}{k\log k}\rt) \ind{I \in E_{\ell-1}(\sigma)} - Q_\ell(\sigma, I)
        \rt] \\
        &=
        1 - \lt(1 - \f{1}{k\log k}\rt) p_{\ell-1}(\sigma) - q_\ell(\sigma).
    \end{align*}
    The second-last line uses the fact that $Q_\ell(\sigma, I) \le \f{1}{k\log k}$ almost surely, and the last line uses the definitions of $p_{\ell-1}(\sigma)$ and $q_\ell(\sigma)$.
    Rearranging yields the desired bound.
\end{proof}

\begin{proposition}
    \label{prop:impossibility-ogp-decouple-joint-event}
    For each $\sigma \in \{\T,\F\}^k$, we have that
    \[
        p_k(\sigma) \ge \lt(1 - \f{1}{\log k}\rt) \sum_{\ell=1}^k q_\ell(\sigma).
    \]
\end{proposition}
\begin{proof}
    By iterating Proposition~\ref{prop:impossibility-ogp-decouple-joint-event-aux}, we get
    \[
        p_k(\sigma)
        \ge
        \lt(1 - \f{1}{k\log k}\rt)^k
        p_0(\sigma)
        +
        \sum_{\ell=1}^k
        \lt(1 - \f{1}{k\log k}\rt)^{k-\ell}
        q_\ell(\sigma)
        \ge
        \lt(1 - \f{1}{k\log k}\rt)^k
        \sum_{\ell=1}^k q_\ell(\sigma).
    \]
    The result follows from the bound $\lt(1 - \f{1}{k\log k}\rt)^k \ge 1 - \f{1}{\log k}$, by Bernoulli's inequality.
\end{proof}

Equation \eqref{eq:impossibility-ogp-energy-term-to-probabilities} and Proposition~\ref{prop:impossibility-ogp-decouple-joint-event} leave the task of lower bounding $\sum_{\sigma \in \{\T,\F\}^k} \sum_{\ell=1}^k q_\ell(\sigma)$.
This is achieved by the following proposition, which reinterprets $\sum_{\sigma \in \{\T,\F\}^k} q_\ell(\sigma)$, the total contribution of $\yp{\ell}$, as a probability.

\begin{proposition}
    \label{prop:impossibility-ogp-dart-game}
    For each $1\le \ell \le k$, we have that
    \[
        \sum_{\sigma \in \{\T,\F\}^k}
        q_\ell(\sigma)
        \ge
        1 - \beta_\ell e^{-(\beta_\ell-1)} - o_k(1).
    \]
\end{proposition}
\begin{proof}
    Using the definition of $q_\ell(\sigma)$, we have
    \begin{align*}
        \sum_{\sigma \in \{\T,\F\}^k}
        q_\ell(\sigma)
        &=
        \E_{I \sim \unif([n]^k)} \lt[
            \sum_{\sigma \in \{\T,\F\}^k}
            \lt(
                \prod_{r=1}^k
                \phi_\ell (\sigma_{r} | \yp{\le \ell-1}_{I_r})
            \rt)
            \ind{
                \prod_{r=1}^k
                \phi_\ell (\sigma_{r} | \yp{\le \ell-1}_{I_r})
                \le
                \f{1}{k\log k}
            }
        \rt] \\
        &=
        \E_{I \sim \unif([n]^k)} \lt[
            \sum_{\sigma \in \{\T,\F\}^k}
            \lt(
                \prod_{r=1}^k
                \phi_\ell (\sigma_{r} | \yp{\le \ell-1}_{I_r})
            \rt)
            \ind{
                -
                \sum_{r=1}^k
                \log \phi_\ell (\sigma_{r} | \yp{\le \ell-1}_{I_r})
                \ge
                \log k + \log \log k
            }
        \rt].
    \end{align*}
    This quantity is the success probability of the following experiment.
    Sample positive random variables $u_1,\ldots,u_k$ by the following procedure, repeated independently for each $r \in [k]$.
    Sample $i\in \unif([n])$; this determines the value of $\yp{\le \ell-1}_i$.
    Then, sample $b \in \{\T,\F\}$ from the measure $\phi_\ell (\cdot | \yp{\le \ell-1}_{i})$.
    Finally, set $u_r = - \log \phi_\ell (b | \yp{\le \ell-1}_{I_r})$.
    The experiment succeeds if $\sum_{r=1}^k u_r \ge \log k + \log \log k$.

    For $r\in [k]$, let $v_r = \min(u_r, \log k)$.
    This is a proxy for $u_r$ with an almost sure upper bound, which allows us to control the experiment's failure probability by a Chernoff bound.
    This failure probability is bounded by
    \[
        \P \lt[\sum_{r=1}^k u_r < \log k + \log \log k\rt]
        \le
        \P \lt[\sum_{r=1}^k v_r < \log k + \log \log k\rt]
        =
        \P \lt[\sum_{r=1}^k \f{v_r}{\log k} < 1 + \f{\log \log k}{\log k}\rt].
    \]
    Note that the $\f{v_r}{\log k}$ are i.i.d. random variables in $[0,1]$ almost surely.
    To bound this last probability by a Chernoff bound, we will lower bound $\E [v_r]$.
    By the definition of $\phi_\ell$,
    \begin{align*}
        \E [u_r]
        &=
        \E_{i\sim \unif([n])}
        \E_{b\sim \phi_\ell (\cdot | \yp{\le \ell-1}_{i})}
        \lt[
            - \log \phi_\ell (b | \yp{\le \ell-1}_{i})
        \rt] \\
        &=
        H(\pi(\yp{\ell} | \yp{0}, \ldots, \yp{\ell-1}))
        =
        \beta_\ell \f{\log k}{k}.
    \end{align*}
    Moreover,
    \begin{align*}
        \E [u_r - v_r]
        &=
        \E \lt[
            (u_r - \log k)\ind{u_r \ge \log k}
        \rt] \\
        &=
        \E_{i\sim \unif([n])}
        \lt[
            \sum_{b\in \{\T,\F\}}
            \phi_\ell (b | \yp{\le \ell-1}_{i})
            \log \f{1}{k \phi_\ell (b | \yp{\le \ell-1}_{i})}
            \ind{\phi_\ell (b | \yp{\le \ell-1}_{i}) \le \f{1}{k}}
        \rt].
    \end{align*}
    For each $i\in [n]$, the quantity inside the last expectation is nonzero for at most one $b\in \{\T,\F\}$ (for $k\ge 3$).
    Moreover, on the interval $[0, \f{1}{k}]$, the function $x \mapsto x \log \f{1}{kx}$ has maximum value $\f{1}{ek}$, attained at $x = \f{1}{ek}$.
    Thus, $\E [u_r-v_r] \le \f{1}{ek}$.
    It follows that $\E [v_r] \ge \beta_\ell \f{\log k}{k} - \f{1}{ek}$.
    So,
    \[
        \E \lt[\sum_{r=1}^k \f{v_r}{\log k}\rt]
        \ge
        \beta_\ell
        - \f{1}{e \log k}.
    \]
    Furthermore, $(1 + \f{\log \log k}{\log k})/(\beta_\ell - \f{1}{e \log k}) = \f{1}{\beta_\ell} + o_k(1)$.
    So, by a Chernoff bound,
    \[
        \P \lt[\sum_{r=1}^k \f{v_r}{\log k} < 1 + \f{\log \log k}{\log k}\rt]
        \le
        \lt(
            \f{e^{-(1 - \f{1}{\beta_\ell} - o_k(1))}}{\lt(\f{1}{\beta_\ell} + o_k(1)\rt)^{\f{1}{\beta_\ell} + o_k(1)}}
        \rt)^{\beta_\ell - o_k(1)}
        =
        \beta_\ell e^{-(\beta_\ell - 1)} + o_k(1).
    \]
    Hence,
    \[
        \P \lt[\sum_{r=1}^k u_r \ge \log k + \log \log k\rt]
        \ge
        1 - \beta_\ell e^{-(\beta_\ell - 1)} - o_k(1),
    \]
    as desired.
\end{proof}

We can now combine these propositions to prove Proposition~\ref{prop:impossibility-ogp-energy-term}.

\begin{proof}[Proof of Proposition~\ref{prop:impossibility-ogp-energy-term}]
    By combining \eqref{eq:impossibility-ogp-energy-term-to-probabilities}, Proposition~\ref{prop:impossibility-ogp-decouple-joint-event}, and Proposition~\ref{prop:impossibility-ogp-dart-game}, we have
    \begin{align*}
        \E_{I\sim \unif([n]^k)} \lt|\lt\{\yp{\ell}[I] : 0\le \ell \le k\rt\}\rt|
        &\ge
        \lt(1 - \f{1}{\log k}\rt)
        \sum_{\ell=1}^k
        \sum_{\sigma \in \{\T,\F\}^k}
        q_\ell(\sigma) \\
        &\ge
        \lt(1 - o_k(1)\rt)
        \sum_{\ell=1}^k
        \lt(1 - \beta_\ell e^{-(\beta_\ell-1)}\rt).
    \end{align*}
\end{proof}

%% file: tex/6-ldp-stability.tex
\section{Stability of Low Degree Polynomials}
\label{sec:impossibility-ldp-stability}

In this section, we will prove Proposition~\ref{prop:impossibility-prob-bounds}(\ref{itm:impossibility-prob-bound-consecutive}), which lower bounds the probability that $\xp{t}$ satisfies $\phip{t}$ for all $0\le t\le T$ and the sequence $\xp{t}$ has no large jumps in Hamming distance.

The proof is a mild generalization of the stability analysis in \cite[Subsection 4.1]{GJW20} and \cite[Subsection 2.3]{Wei20} from a biased Boolean hypercube to a product of discrete uniform measures.
Like in these two works, the proof proceeds in two steps.
The interpolation path can be modeled as a walk on a product graph whose vertices are the elements of $\Omega_k(n,m)$, where two vertices are adjacent if they differ by one literal.
An edge $(\Phi, \Phi')$ is bad if the output of our polynomial $f$ has a large jump between inputs $\Phi$ and $\Phi'$.
In the first step, we will use Fourier analysis to upper bound the fraction of bad edges.
In the second step, we translate this bound to a lower bound on the probability that our walk encounters no bad edges.

\subsection{An Upper Bound on the Rate of Bad Steps}
\label{subsec:impossibility-influence}

We begin by formalizing the notion of $c$-badness.
Recall that $N = m \cdot k \cdot 2n$, and each $\Phi \in \Omega_k(n,m)$ is identified with a vector of indicators in $\{0,1\}^N$, which is the input of a low degree polynomial.
\begin{definition}[$c$-badness]
    \label{defn:impossibility-cbad}
    Let $c>0$ and let $f: \bR^N \to \bR^n$ be a deterministic degree-$D$ polynomial.
    A pair of formulas $(\Phi, \Phi')\in \Omega_k(n,m)^2$ is \vocab{$c$-bad} (with respect to $f$) if $\norm{f(\Phi) - f(\Phi')}_2^2 > c\E_{\Phi \sim \Phi_k(n,m)} \norm{f(\Phi)}_2^2$.
\end{definition}

Recall the interpolation path $\phip{0},\phip{1},\ldots,\phip{T}$ defined in Definition~\ref{defn:interpolation}.
We will prove Proposition~\ref{prop:impossibility-prob-bounds}(\ref{itm:impossibility-prob-bound-consecutive}) via the following proposition, which controls the probability that the output of $f$ does not have a large jump between any pair of consecutive assignments in the interpolation path.
\begin{proposition}
    \label{prop:impossibility-no-cbad-edge}
    Let $f: \bR^N \to \bR^n$ be a deterministic degree-$D$ polynomial.
    With probability at least $(2n)^{-4Dk/c}$, $(\phip{t-1}, \phip{t})$ is not $c$-bad with respect to $f$ for any $1\le t\le T$.
\end{proposition}

We will prove this proposition in Subsection~\ref{subsec:impossibility-ldp-stability-proof}.
The objective of this subsection is to prove Proposition~\ref{prop:impossibility-total-influence} below, which upper bounds the fraction of all possible steps that are bad.
To this end, for $1\le j\le km$, define $\Phi_k(n,m;j)$ as the measure of a sample $(\Phi, \Phi') \in \Omega_k(n,m)^2$ obtained by sampling $\Phi\sim \Phi_k(n,m)$, and then obtaining $\Phi'$ from $\Phi$ by resampling the $j$th lexicographic literal $\Phi'_{L(j)}$ from $\unif\lt(\cL \setminus \{\Phi_{L(j)}\} \rt)$.
(Recall the definition of $L(j)$ before Definition~\ref{defn:interpolation}.)
Define
\[
    \lambda_j = \P_{(\Phi,\Phi') \sim \Phi_k(n,m;j)} \lt(\text{$(\Phi, \Phi')$ is $c$-bad with respect to $f$}\rt).
\]
This is the fraction of pairs of formulas in $\Omega_k(n,m)$, differing in exactly the $j$th lexicographic literal, that are $c$-bad with respect to $f$.

\begin{proposition}
    \label{prop:impossibility-total-influence}
    If $f$ is a deterministic degree-$D$ polynomial, then $\sum_{j=1}^{km} \lambda_j \le \f{4D}{c}$.
\end{proposition}

We recall the following orthogonal decomposition property of functions on product measures, which can be thought of as a generalization of Fourier analysis on the Boolean cube.
We will give brief self-contained proofs of the relevant facts; a full discussion can be found in \cite[Chapter 8.3]{Odo14}.
Let $(\cX, \P_X)$ be an arbitrary probability space, and let $J$ be a positive integer.
Let $X = (X_1, \ldots, X_J) \in \cX^J$.
For $j\in [J]$, define the operators $\sD_j$ and $\sE_j$ as follows.
For any function $g : \cX^J \to \bR$, $\sE_j g$ is the function satisfying
\[
    \sE_j g(X) = \E_{X_j\sim (\cX, \P_X)} g(X),
\]
where in the right-hand side the coordinate $X_j$ is resampled from $(\cX, \P_X)$.
Let $\sD_j g = g - \sE_j g$.
Note that the operators $\{\sD_j, \sE_j\}_{j\in [J]}$ commute.
For $S\subseteq [J]$, define the functions
\[
    \hg_S = \prod_{j\in S} \sD_j \prod_{j\in [J]\setminus S} \sE_j g.
\]
Note that $g = \sum_{S\subseteq [J]} \hg_S$.
Moreover, $\hg_S$ depends only on the inputs $\{X_j: j\in S\}$.
For any $j$,
\[
    \E_{X_j \sim (\cX, \P_X)} g(X)^2 = \E_{X_j \sim (\cX, \P_X)} \lt[(\sD_j g)(X)^2\rt] + (\sE_j g)(X)^2,
\]
and so by induction
\[
    \E_{X\sim (\cX, \P_X)^{\otimes J}} g(X)^2
    =
    \sum_{S\subseteq [J]}
    \E_{X\sim (\cX, \P_X)^{\otimes J}}
    \hg_S(X)^2.
\]
For $j\in [J]$, define ${\Var}_j g(X) = \E_{X_j\sim (\cX, \P_X)} \lt[(\sD_j g)(X)^2\rt]$.
We begin with the following inequality, which can be considered a converse to the Efron-Stein inequality.

\begin{lemma}
    \label{lem:impossibility-reverse-efron-stein}
    Suppose a function $g: \cX^J \to \bR$ can be written in the form $g(X) = \sum_{i=1}^I g_i(X)$, where each $g_i(X)$ depends on at most $D$ coordinates of $X$. Then,
    \[
        D \Var_{X\sim (\cX, \P_X)^{\otimes J}} g(X)
        \ge
        \sum_{j=1}^J
        \E_{X\sim (\cX, \P_X)^{\otimes J}}
        {\Var}_j g(X).
    \]
\end{lemma}
\begin{proof}
    By the orthogonal expansion above, we have
    \[
        \Var_{X\sim (\cX, \P_X)^{\otimes J}} g(X)
        =
        \sum_{\substack{S\subseteq [J] \\ S\neq \emptyset}}
        \E_{X\sim (\cX, \P_X)^{\otimes j}}
        \hg_S(X)^2
        \qquad
        \text{and}
        \qquad
        \E_{X\sim (\cX, \P_X)^{\otimes J}}
        {\Var}_j g(X)
        =
        \sum_{\substack{S\subseteq [J] \\ S \ni j}}
        \E_{X\sim (\cX, \P_X)^{\otimes j}}
        \hg_S(X)^2.
    \]
    We claim that for all $S\subseteq [J]$ with $|S|>D$, we have $\hg_S \equiv 0$.
    For each $i\in [I]$, we have $\prod_{j\in S} \sD_j g_i \equiv 0$, because $S$ contains at least one $j$ such that $g_i(X)$ does not depend on $X_j$.
    Thus, $\prod_{j\in S} \sD_j g \equiv 0$, and so $\hg_S \equiv 0$, as desired.
    Hence,
    \[
        \sum_{j=1}^J
        \E_{X\sim (\cX, \P_X)^{\otimes J}}
        {\Var}_j g(X)
        =
        \sum_{\substack{S\subseteq [J] \\ |S| \le D}}
        |S|
        \E_{X\sim (\cX, \P_X)^{\otimes j}}
        \hg_S(X)^2
        \le
        D \Var_{X\sim (\cX, \P_X)^{\otimes J}} g(X).
    \]
\end{proof}

\begin{proof}[Proof of Proposition~\ref{prop:impossibility-total-influence}]
    Note that $\Phi_k(n,m)$ is composed of $km$ i.i.d. literals, and thus can be thought of as the product measure $\unif(\cL)^{\otimes km}$.
    By slight abuse of notation, for $1\le j\le km$, we can define $\sD_j$ and $\sE_j$ as the above operators with respect to the $j$th lexicographic literal $\Phi_{L(j)}$ of $\Phi$.

    For $1\le \ell\le n$, let $f_\ell$ denote the $\ell$th component of $f$.
    By Markov's inequality and the inequality $(a-b)^2 \le 2a^2+2b^2$, we have
    \begin{align*}
        \sum_{j=1}^{km}
        \lambda_j
        &\le
        \sum_{j=1}^{km}
        \f{
            \E_{(\Phi,\Phi') \sim \Phi_k(n,m;j)}
            \norm{f(\Phi) - f(\Phi')}_2^2
        }{
            c\E_{\Phi \sim \Phi_k(n,m)}
            \norm{f(\Phi)}_2^2
        } \\
        &=
        \f{
            \sum_{\ell=1}^n
            \sum_{j=1}^{km}
            \E_{(\Phi,\Phi') \sim \Phi_k(n,m;j)}
            ((\sD_j f_\ell)(\Phi) - (\sD_j f_\ell)(\Phi'))^2
        }{
            c
            \sum_{\ell=1}^n
            \E_{\Phi \sim \Phi_k(n,m)}
            f_\ell(\Phi)^2
        } \\
        &\le
        \f{
            2
            \sum_{\ell=1}^n
            \sum_{j=1}^{km}
            \E_{(\Phi,\Phi') \sim \Phi_k(n,m;j)}
            \lt((\sD_j f_\ell)(\Phi)^2 + (\sD_j f_\ell)(\Phi')^2\rt)
        }{
            c
            \sum_{\ell=1}^n
            \E_{\Phi \sim \Phi_k(n,m)}
            f_\ell(\Phi)^2
        } \\
        &=
        \f{
            4
            \sum_{\ell=1}^n
            \sum_{j=1}^{km}
            \E_{\Phi \sim \Phi_k(n,m)}
            {\Var}_j f_\ell(\Phi)
        }{
            c
            \sum_{\ell=1}^n
            \E_{\Phi \sim \Phi_k(n,m)}
            f_\ell(\Phi)^2
        }.
    \end{align*}
    Now, each $f_\ell$ is a degree-$D$ polynomial in the indicators $\Phi_{i,j,s}$ that $\Phi_{i,j}$ is the $s$th literal in $\cL$.
    So, each monomial of each $f_\ell$ depends on at most $D$ literals of $\Phi$.
    By Lemma~\ref{lem:impossibility-reverse-efron-stein},
    \[
        \sum_{j=1}^{km}
        \E_{\Phi \sim \Phi_k(n,m)}
        {\Var}_j f_\ell(\Phi)
        \le
        D
        \Var_{\Phi\sim \Phi_k(n,m)} f_\ell(\Phi)
        \le
        D
        \E_{\Phi \sim \Phi_k(n,m)}
        f_\ell(\Phi)^2.
    \]
    So, $\sum_{j=1}^{km} \lambda_j \le \f{4D}{c}$.
\end{proof}

\subsection{Bounding the Probability of no Bad Step}
\label{subsec:impossibility-no-bad-step}

Proposition~\ref{prop:impossibility-total-influence} bounds the combined rate of $c$-bad steps.
To derive Proposition~\ref{prop:impossibility-no-cbad-edge}, we must translate this bound on the rate of $c$-bad steps to a bound on the probability that interpolation path never takes a $c$-bad step.
To make the ideas in our argument more clear, we abstract to the following graph theoretic problem, which is interesting in its own right.

Let $\Sigma$ be a set of symbols and $J,T$ be positive integers.
Let $G$ be a graph on $\Sigma^J$, where two nodes are adjacent if their Hamming distance is exactly $1$.
Each edge has a \vocab{direction} $j\in [J]$, the index on which its endpoints disagree.
Let an arbitrary subset of edges be \vocab{bad}; for adjacent vertices $v,w$, let $B(v,w)$ denote the event that the edge $(v,w)$ is bad.
For $j\in [J]$, let $\lambda_j$ denote the fraction of edges in direction $j$ that are bad.
Equivalently, $\lambda_j = \P(B(v,w))$, where $v\sim \unif(G)$ and $w$ is obtained from $v$ by resampling $w_j$ from $\unif\lt(\Sigma \setminus \{v_j\}\rt)$.

Let $\sigma : [T] \to [J]$ be an arbitrary map.
Consider the (lazy) random walk $\vp{0}, \vp{1}, \ldots, \vp{T}$ such that $\vp{0} \sim \unif(G)$ and for $1\le t\le T$, $\vp{t}$ is obtained from $\vp{t-1}$ by resampling $\vp{t}_{\sigma(t)}$ from $\unif(\Sigma)$.

\begin{lemma}
    \label{lem:impossibility-graphwalk}
    With probability at least $|\Sigma|^{-\sum_{t=1}^T \lambda_{\sigma(t)}}$, no step of the random walk $\vp{0}, \vp{1}, \ldots, \vp{T}$ traverses a bad edge.
\end{lemma}
Note that at each step, the random walk either traverses an edge or does not move; we say that the steps that do not move do not traverse a bad edge.
The lemma is sharp, for example, when all the $\lambda_j$ are $0$ or $1$: in this case, the random walk does not traverse a bad edge if it does not move at all times $t$ with $\lambda_{\sigma(t)} = 1$.
\begin{proof}
    For $v\in \Sigma^J$, let $q(v)$ be the probability that the random walk $\vp{0}, \vp{1}, \ldots, \vp{T}$ does not traverse a bad edge, starting from $\vp{0} = v$.
    We will prove by induction on $T$ that
    \[
        \E_{v\sim \unif(G)}
        \log q(v)
        \ge
        - \log |\Sigma| \cdot \sum_{t=1}^T \lambda_{\sigma(t)}.
    \]
    The lemma then follows from Jensen's inequality, because $\log \E q(v) \ge \E \log q(v)$.

    The base case of the claim, $T=0$, follows trivially.
    For the inductive step, let $\tq(v)$ be the probability that the random walk $\vp{1}, \vp{2}, \ldots, \vp{T}$ does not traverse a bad edge, starting from $\vp{1} = v$.
    Let $j = \sigma(1)$.
    Let $v_{\sim j}\in \Sigma^{J-1}$ denote an element of $\Sigma^J$ with the $j$th coordinate left blank.
    For $s\in \Sigma$, let $v_{\sim j}[s]\in \Sigma^{J}$ denote $v_{\sim j}$ with the $j$th coordinate set to $s$.

    For now, fix some $v_{\sim j} \in \Sigma^{J-1}$.
    For $s\in \Sigma$, we have that
    \begin{align*}
        q(v_{\sim j}[s])
        &=
        \sum_{s'\in \Sigma}
        \f{1}{|\Sigma|}
        \ind{\text{$s'=s$ or $B\lt(v_{\sim j}[s], v_{\sim j}[s']\rt)^c$}}
        \tq(v_{\sim j}[s']) \\
        &=
        \sum_{s' \in \Sigma \setminus \{s\} }
        \f{1}{|\Sigma|-1}
        \lt(
            \f{1}{|\Sigma|}
            \tq(v_{\sim j}[s])
            +
            \f{|\Sigma|-1}{|\Sigma|}
            \ind{B\lt(v_{\sim j}[s], v_{\sim j}[s']\rt)^c}
            \tq(v_{\sim j}[s'])
        \rt).
    \end{align*}
    By Jensen's inequality, this implies
    \[
        \log q(v_{\sim j}[s])
        \ge
        \sum_{s' \in \Sigma \setminus \{s\} }
        \f{1}{|\Sigma|-1}
        \log
        \lt(
            \f{1}{|\Sigma|}
            \tq(v_{\sim j}[s])
            +
            \f{|\Sigma|-1}{|\Sigma|}
            \ind{B\lt(v_{\sim j}[s], v_{\sim j}[s']\rt)^c}
            \tq(v_{\sim j}[s'])
        \rt).
    \]
    Taking an expectation over $s \sim \unif(\Sigma)$, we have
    \begin{align}
        \notag
        \E_{s\sim \unif(\Sigma)}
        \log q(v_{\sim j}[s])
        &\ge
        \sum_{\substack{s, s' \in \Sigma \\ s \neq s'}}
        \f{1}{|\Sigma|(|\Sigma|-1)} \lt[
            \log
            \lt(
                \f{1}{|\Sigma|}
                \tq(v_{\sim j}[s])
                +
                \f{|\Sigma|-1}{|\Sigma|}
                \ind{B\lt(v_{\sim j}[s], v_{\sim j}[s']\rt)^c}
                \tq(v_{\sim j}[s'])
            \rt)
        \rt] \\
        &=
        \label{eq:impossibility-graphwalk-slice-ev}
        \sum_{\substack{s, s' \in \Sigma \\ s \neq s'}}
        \f{1}{2|\Sigma|(|\Sigma|-1)}
        \xi(v_{\sim j}, s, s'),
    \end{align}
    where for $s\neq s'$,
    \begin{align*}
        \xi(v_{\sim j}, s, s')
        &=
        \log
        \lt(
            \f{1}{|\Sigma|}
            \tq(v_{\sim j}[s])
            +
            \f{|\Sigma|-1}{|\Sigma|}
            \ind{B\lt(v_{\sim j}[s], v_{\sim j}[s']\rt)^c}
            \tq(v_{\sim j}[s'])
        \rt) \\
        &\qquad +
        \log
        \lt(
            \f{1}{|\Sigma|}
            \tq(v_{\sim j}[s'])
            +
            \f{|\Sigma|-1}{|\Sigma|}
            \ind{B\lt(v_{\sim j}[s], v_{\sim j}[s']\rt)^c}
            \tq(v_{\sim j}[s])
        \rt).
    \end{align*}
    If $B\lt(v_{\sim j}[s], v_{\sim j}[s']\rt)$ holds, then $\xi(v_{\sim j}, s, s') = \log \tq(v_{\sim j}[s]) + \log \tq(v_{\sim j}[s']) - 2 \log |\Sigma|$.
    Otherwise, by Jensen's inequality we have
    \[
        \log
        \lt(
            \f{1}{|\Sigma|}
            \tq(v_{\sim j}[s])
            +
            \f{|\Sigma|-1}{|\Sigma|}
            \ind{B\lt(v_{\sim j}[s], v_{\sim j}[s']\rt)^c}
            \tq(v_{\sim j}[s'])
        \rt)
        \ge
        \f{1}{|\Sigma|}
        \log \tq(v_{\sim j}[s]) +
        \f{|\Sigma|-1}{|\Sigma|}
        \log \tq(v_{\sim j}[s'])
    \]
    and similarly for the other term of $\xi(v_{\sim j}, s, s')$.
    In this case, $\xi(v_{\sim j}, s, s') \ge \log \tq(v_{\sim j}[s]) + \log \tq(v_{\sim j}[s'])$.
    So, in all cases
    \[
        \xi(v_{\sim j}, s, s')
        \ge
        \log \tq(v_{\sim j}[s])
        + \log \tq(v_{\sim j}[s'])
        - 2 \ind{B\lt(v_{\sim j}[s], v_{\sim j}[s']\rt)} \log |\Sigma|.
    \]
    Substituting into \eqref{eq:impossibility-graphwalk-slice-ev}, we have
    \[
        \E_{s\sim \unif(\Sigma)}
        \log q(v_{\sim j}[s])
        \ge
        \E_{s\sim \unif(\Sigma)}
        \log \tq(v_{\sim j}[s])
        -
        \log |\Sigma| \cdot
        \sum_{\substack{s, s' \in \Sigma \\ s \neq s'}}
        \f{\ind{B\lt(v_{\sim j}[s], v_{\sim j}[s']\rt)}}{|\Sigma|(|\Sigma|-1)}.
    \]
    Taking an expectation over $v_{\sim j}$ yields
    \[
        \E_{v\sim \unif(G)} \log q(v)
        \ge
        \E_{v\sim \unif(G)} \log \tq(v)
        -
        \log |\Sigma| \cdot \lambda_j.
    \]
    By induction, we have
    \[
        \E_{v\sim \unif(G)} \log \tq(v)
        \ge
        -\log |\Sigma| \cdot
        \sum_{t=2}^T
        \lambda_{\sigma(t)},
    \]
    and the result follows.
\end{proof}

\subsection{Completing the Proof of Stability}
\label{subsec:impossibility-ldp-stability-proof}

\begin{proof}[Proof of Proposition~\ref{prop:impossibility-no-cbad-edge}]
    Our interpolation scheme can be modeled as the random walk in Subsection~\ref{subsec:impossibility-no-bad-step}, with $\Sigma = \cL$, $J=km$, $T = k^2m$, $\sigma(t)$ defined in Definition~\ref{defn:interpolation}, and where the bad edges are the $c$-bad edges.
    This correspondence is consistent because the steps in the interpolation path where a formula transitions to itself are never $c$-bad.

    Since $\sigma$ maps to every value in $[km]$ $k$ times and $|\cL| = 2n$, Proposition~\ref{prop:impossibility-total-influence} and Lemma~\ref{lem:impossibility-graphwalk} imply that the probability of never traversing a $c$-bad edge is at least $(2n)^{-4Dk/c}$.
\end{proof}

\begin{proof}[Proof of Proposition~\ref{prop:impossibility-prob-bounds}(\ref{itm:impossibility-prob-bound-consecutive})]
    Set $c = \f{\bp - \bm}{\gamma k}$.
    Let $\Snobad$ be the event that for all $1\le t\le T$, $(\phip{t-1}, \phip{t})$ is not $c$-bad with respect to $f$.
    By Proposition~\ref{prop:impossibility-no-cbad-edge}, $\P(\Snobad) \ge (2n)^{-4Dk^2\gamma/(\bp - \bm)}$.

    By a union bound, $\P(\Svalid) \ge 1 - (T+1) \delta$.
    Thus, $\P(\Svalid \cap \Snobad) \ge (2n)^{-4Dk^2\gamma/(\bp - \bm)} - (T+1) \delta$.
    We claim that on $\Svalid \cap \Snobad$, the event $\Sconsec$ also occurs.

    Suppose for sake of contradiction that $\Svalid \cap \Snobad$ holds and for some $1\le t \le T$, we have that $\Delta(\xp{t-1}, \xp{t}) > \f{\bp - \bm}{2k}$.
    Because $(\phip{t-1}, \phip{t})$ is not $c$-bad, we have
    \[
        \norm{f(\phip{t-1}) - f(\phip{t})}_2^2
        \le
        c\E_{\Phi \sim \Phi_k(n,m)} \norm{f(\Phi)}_2^2
        \le
        c\gamma n
        =
        \f{\bp - \bm}{k} n.
    \]
    Let $I = \{i\in [n]: \xp{t-1}_i \neq \xp{t}_i\}$, so $|I| > \f{\bp - \bm}{2k} n$.
    Define
    \[
        B^{(t)}
        =
        \lt\{
            i\in [n]:
            \xp{t}_i \neq (\round \circ f)(\phip{t})_i
        \rt\}.
    \]
    Because $\xp{t}_i = \cA(\phip{t})_i$ and the assistance subroutine $\cB$ in $\cA$ can edit only an $\eta$ fraction of bits of the assignment, $|B^{(t)}| \le \eta n = \f{\bp - \bm}{8k}n$.
    For similarly defined $B^{(t-1)}$, we likewise have $|B^{(t-1)}|\le \f{\bp - \bm}{8k}n$.

    Let $J = I \setminus (B^{(t-1)} \cup B^{(t)})$, so $|J| > \f{\bp - \bm}{4k} n$.
    For all $i \in J$, one of $f_i(\phip{t-1})$ and $f_i(\phip{t})$ is at least $1$ and the other is at most $-1$, so $|f_i(\phip{t-1}) - f_i(\phip{t})| \ge 2$.
    So,
    \[
        \norm{f(\phip{t-1}) - f(\phip{t})}_2^2
        \ge
        \sum_{i\in J}
        |f_i(\phip{t-1}) - f_i(\phip{t})|^2
        >
        \f{\bp - \bm}{k} n.
    \]
    This is a contradiction.
    Therefore $\Sconsec \supseteq \Svalid \cap \Snobad$, and so
    \[
        \P(\Svalid \cap \Sconsec)
        \ge
        \P(\Svalid \cap \Snobad)
        \ge
        (2n)^{-4Dk^2\gamma/(\bp - \bm)} - (T+1) \delta.
    \]
\end{proof}

%% file: tex/7-local-hardness.tex
\section{Proof of Impossibility for Local Algorithms}
\label{sec:local-algs}

This section proves our impossibility result for local algorithms, Theorem~\ref{thm:impossibility-local-algs}.
Throughout, fix $\kappa > \kappast$ and $r\in \bN$.
Fix a probability space $(\Omega, \P_\omega)$, and let $\cA$ be an $r$-local algorithm that, on input $\Phi\in \Omega_k(n,m)$ with factor graph $(G,\rho)$, samples internal randomness $\varphi \sim (\Omega, \P_\omega)^{\otimes (V_G \cup E_G)}$.

Let $\varphi_V = \varphi\big|_{V_G}$ and $\varphi_E = \varphi\big|_{E_G}$.
We will actually prove Theorem~\ref{thm:impossibility-local-algs} conditioned on any realization of $\varphi_V$.
Fix once and for all a realization of $\varphi_V$; all probabilities and expectations in this section will implicitly be conditioned on this realization.

\subsection{A Different Interpolation}
\label{subsec:local-algs-interpolation}

Instead of the interpolation path of $k^2m$ problem instances used in the proof of Theorem~\ref{thm:impossibility}, we now use an interpolation structured as $k$ paths of length $km$ originating at a common point.
We also couple to this interpolation the internal randomness $\varphi_E$ of $\cA$ run on these problem instances.

\begin{definition}[Interplation structure]
    Let $\oT=km$.
    We will sample $\phip{0}$ and $\phip{\ell,t} \in \Omega_k(n,m)$ for $1\le \ell \le k$, $1\le t\le \oT$.
    Let the factor graphs of these $k$-SAT instances be $(\Gp{0},\rhop{0})$ and $(\Gp{\ell,t},\rhop{\ell,t})$.
    We also sample maps $\varphiep{0} : E_{\Gp{0}} \to \Omega$ and $\varphiep{\ell,t} : E_{\Gp{\ell,t}} \to \Omega$.

    We sample $\phip{0} \sim \Phi_k(n,m)$ and $\varphiep{0} \sim (\Omega, \P_\omega)^{\otimes E_{\Gp{0}}}$.
    For $1\le \ell \le k$, $1\le t\le \oT$, we obtain $\phip{\ell,t}$ from $\phip{\ell,t-1}$ (take $\phip{\ell,0} = \phip{0}$ for all $\ell$) by resampling $\phip{\ell,t}_{L(t)}$ from $\unif(\cL)$.
    (Recall that $L(t)$ is the $t$th pair $(a,b) \in [m]\times [k]$ in lexicographic order.)

    This resampling deletes an edge $e$ from $\Gp{\ell,t-1}$, adds an edge $e'$ to $\Gp{\ell,t}$ (possibly in the same location), and samples $\rhop{\ell,t}(e) \sim \unif(\{\T,\F\})$.
    We obtain $\varphiep{\ell,t}$ from $\varphiep{\ell,t-1}$ by deleting the entry for $e'$ and sampling $\varphiep{\ell, t}(e) \sim (\Omega, \P_\omega)$.
\end{definition}

In other words, starting from a random $k$-SAT instance we sample $k$ interpolation paths, where in each path we resample the literals one by one in the same order.
Resampling a literal resamples an edge of the factor graph, and we also resample the output of $\varphi_E$ on that edge.

Note that each $(\phip{\ell, t}, \varphiep{\ell, t})$ is marginally distributed as $(\Phi, \varphi_E)$ where $\Phi \sim \Phi_k(n,m)$ and $\varphi_E \sim (\Omega, \P_\omega)^{\otimes E_G}$, where $(G,\rho)$ is the factor graph of $\Phi$.
Moreover, $(\phip{\ell, \oT}, \varphiep{\ell, \oT})$ is independent of $(\phip{0}, \varphiep{0})$ and $(\phip{\ell', t}, \varphiep{\ell', t})$ for all $\ell'\neq \ell$.

\subsection{Selecting Problem Instances Yielding the Forbidden Structure}

We set $\bm, \bp, \eps$ as in Subsection~\ref{subsec:impossibility-outline}.
Recall that these numbers depend only on $\kappa$, and $\f{\beta+\eps}{1-\beta e^{-(\beta-1)}} \le \kappa$ for all $\beta \in [\bm, \bp]$.
Moreover, set $\bz = \f{\bm+\bp}{2}$.

In the proof of Theorem~\ref{thm:impossibility}, we selected the problem instances $\phip{t_0}, \phip{t_1}, \ldots, \phip{t_k}$ where the algorithm outputs form a forbidden structure after observing the entire interpolation.
Here, because we can leverage concentration properties of local algorithms (instead of stability properties of low degree polynomials), we know in advance which problem instances to choose.
This allows us to immediately restrict our attention to $k+1$ problem instances, instead of the full interpolation structure.

We will choose the problem instances $\phip{0}$ and $\phip{\ell, t_\ell}$ for $1\le \ell \le k$, for indices $t_1,\ldots,t_k$ we now determine.
Consider random variables (which depend on the $t_\ell$)
\[
    \xp{0} = \cA(\phip{0}, \varphiep{0})
    \qquad
    \text{and}
    \qquad
    \xp{\ell} = \cA(\phip{\ell, t_\ell}, \varphiep{\ell, t_{\ell}})
\]
for $1\le \ell \le k$.
Here, $\cA(\Phi, \varphi_E)$ denotes $\cA$ run with input $\Phi$ and internal randomness $\varphi_E$ (we suppress the dependence on $\varphi_V$, which is fixed).

We inductively define $t_\ell$ in terms of $t_1,\ldots,t_{\ell-1}$ as the smallest number satisfying $1\le t_\ell \le \oT$ and
\[
    \E H\lt(\pi(\xp{\ell} | \xp{0}, \ldots, \xp{\ell-1})\rt)
    \ge
    \bz \f{\log k}{k}.
\]
If no such $t_\ell$ exists, set $t_\ell = \oT$ and say $\ell$ is \emph{deficient}.
Note that the $t_\ell$ are a deterministic function of $\cA$ and $\varphi_V$.
For $1\le \ell \le k$, define $\phip{\ell} = \phip{\ell, t_\ell}$ and $\varphiep{\ell} = \varphiep{\ell, t_\ell}$.
Let $(\Gp{\ell},\rhop{\ell})$ be the factor graph of $\phip{\ell}$.

The following lemma follows from the stability of the expected conditional overlap entropy when we resample one literal of $\phip{\ell, t}$.
We defer its proof to Subsection~\ref{subsec:impossibility-local-bdd-differences}.
The exponent $-\f12$ can be replaced by any constant larger than $-1$.
\begin{lemma}
    \label{lem:local-alg-e-bd}
    If $\ell \in [k]$ is not deficient, then
    \[
        \bz \f{\log k}{k}
        \le
        \E H\lt(\pi(\xp{\ell} | \xp{0}, \ldots, \xp{\ell-1})\rt)
        \le
        \bz \f{\log k}{k} + O(n^{-1/2}).
    \]
\end{lemma}

\subsection{Outline of the Proof}

For the rest of this proof, take $\eta = \f{\bp-\bm}{32k^2}$ and $\nu = \f{1}{k^2 2^k}$.
We now define events $\Svalid, \Sconc, \Sindep, \Sogp$, which are measurable in the $(\phip{\ell}, \varphiep{\ell})$ for $0\le \ell \le k$.

Let $\Svalid$ be the event that for all $0\le \ell \le k$, $\xp{\ell}$ $(\eta, \nu)$-satisfies $\phip{\ell}$.
Let $\Sconc$ be the event that for all $1\le \ell \le k$,
\[
    \lt|
        H\lt(\pi(\xp{\ell} | \xp{0}, \ldots, \xp{\ell-1})\rt)
        -
        \E H\lt(\pi(\xp{\ell} | \xp{0}, \ldots, \xp{\ell-1})\rt)
    \rt|
    \le
    \f{1}{\log n}.
\]
For $1\le \ell \le k$, let $\Sindep^\ell$ be the event that there does not exist an assignment $y\in \{\T,\F\}^n$ such that
\begin{enumerate}[label=(IND-\Alph*), ref=IND-\Alph*, leftmargin=50pt]
    \item \label{itm:def-sindep-local-satisfy} $y$ $\nu$-satisfies $\phip{\ell}$;
    \item \label{itm:def-sindep-local-entropy} $H\lt(\pi(y | \xp{0}, \ldots,\xp{\ell-1}) \rt) \le \bp \f{\log k}{k}$.
\end{enumerate}
Let $\Sogp$ be the event that there does not exist assignments $\yp{0},\ldots,\yp{k}\in \{\T,\F\}^n$ such that
\begin{enumerate}[label=(OGP-\Alph*), ref=OGP-\Alph*, leftmargin=50pt]
    \item \label{itm:def-sogp-local-satisfy} For all $0\le \ell\le k$, $\xp{\ell}$ $\nu$-satisfies $\phip{\ell}$;
    \item \label{itm:def-sogp-local-entropy} For all $1\le \ell\le k$, $H\lt(\pi(\yp{\ell} | \yp{0}, \ldots,\yp{\ell-1}) \rt) \in \lt[\bm \f{\log k}{k}, \bp \f{\log k}{k}\rt]$.
\end{enumerate}

Finally, define
\[
    p = \P \lt[
        \text{$\cA(\Phi, \varphi_E)$ $(\eta, \nu)$-satisfies $\Phi$}
    \rt],
\]
where the probability is over the randomness of $\Phi \sim \Phi_k(n,m)$ and $\varphi_E \sim (\Omega, \P_\omega)^{\otimes E_G}$, where $(G,\rho)$ is the factor graph of $\Phi$.
This is the probability upper bounded by Theorem~\ref{thm:impossibility-local-algs}.

We will derive Theorem~\ref{thm:impossibility-local-algs} from the following two propositions.

\begin{proposition}
    \label{prop:impossibility-local-algs-events-relation}
    For all sufficiently large $k$ and $n$, the following relations hold.
    \begin{enumerate}[label=(\alph*), ref=\alph*]
        \item \label{itm:impossibility-local-algs-deficient}
        If $\ell\in [k]$ is deficient, then $\Svalid \cap \Sconsec \cap \Sindep^\ell = \emptyset$.
        \item \label{itm:impossibility-local-algs-no-deficient}
        If no $\ell \in [k]$ is deficient, then $\Svalid \cap \Sconsec \cap \Sogp = \emptyset$.
    \end{enumerate}
\end{proposition}

\begin{proposition}
    \label{prop:impossibility-local-prob-bounds}
    For all sufficiently large $k$ and $n$, the following inequalities hold.
    \begin{enumerate}[label=(\alph*), ref=\alph*]
        \item \label{itm:impossibility-local-prob-bound-valid}
        $\P(\Svalid) \ge p^{k+1}$.
        \item \label{itm:impossibility-local-prob-bound-conc}
        $\P(\Sconc) \ge 1 - \exp(-\tOmega(n^{1/3}))$.
        \item \label{itm:impossibility-local-prob-bound-indep}
        If $\ell \in [k]$ is deficient, then $\P(\Sindep^\ell) \ge 1 - \exp(-\Omega(n))$.
        \item \label{itm:impossibility-local-prob-bound-ogp}
        If no $\ell \in [k]$ is deficient, then $\P(\Sogp) \ge 1 - \exp(-\Omega(n))$.
    \end{enumerate}
\end{proposition}

We will prove Proposition~\ref{prop:impossibility-local-algs-events-relation} in Subsection~\ref{subsec:impossibility-local-algs-events-relation}, Proposition~\ref{prop:impossibility-local-prob-bounds}(\ref{itm:impossibility-local-prob-bound-valid}) in Subsection~\ref{subsec:impossibility-local-prob-bound-valid}, and Proposition~\ref{prop:impossibility-local-prob-bounds}(\ref{itm:impossibility-local-prob-bound-conc}) in Subsection~\ref{subsec:impossibility-local-bdd-differences}.
Proposition~\ref{prop:impossibility-local-prob-bounds}(\ref{itm:impossibility-local-prob-bound-indep},\ref{itm:impossibility-local-prob-bound-ogp}) are analogous to Proposition~\ref{prop:impossibility-prob-bounds}(\ref{itm:impossibility-prob-bound-independent},\ref{itm:impossibility-prob-bound-ogp}).
The proofs are exactly the same, except we no longer need to union bound over all possible choices of the $t_\ell$.

First, let us see how these bounds imply Theorem~\ref{thm:impossibility-local-algs}.
\begin{proof}[Proof of Theorem~\ref{thm:impossibility-local-algs}]
    Set $\kst$ such that for all $k\ge \kst$, Propositions~\ref{prop:impossibility-local-algs-events-relation} and \ref{prop:impossibility-local-prob-bounds} both hold.

    Suppose some $\ell \in [k]$ is deficient.
    By Proposition~\ref{prop:impossibility-local-algs-events-relation}(\ref{itm:impossibility-local-algs-deficient}), $\P(\Svalid) + \P(\Sconc) + \P(\Sindep^\ell) \le 2$.
    By Proposition~\ref{prop:impossibility-local-prob-bounds}, this implies
    \[
        p^{k+1} \le \exp(-\tOmega(n^{1/3})) + \exp(-\Omega(n)),
    \]
    whence $p \le \exp(-\tOmega(n^{1/3}))$.
    If no $\ell \in [k]$ is deficient, then by Proposition~\ref{prop:impossibility-local-algs-events-relation}(\ref{itm:impossibility-local-algs-no-deficient}), $\P(\Svalid) + \P(\Sconc) + \P(\Sogp) \le 2$.
    By Proposition~\ref{prop:impossibility-local-prob-bounds}, we get the same conclusion.
    So, $p \le \exp(-\tOmega(n^{1/3}))$ conditionally on any realization of $\varphi_V$.
\end{proof}

\subsection{Successful Algorithm Outputs Contradict $\Sindep^\ell$ or $\Sogp$}
\label{subsec:impossibility-local-algs-events-relation}

The following corollary to Lemma~\ref{lem:impossibility-hamming-to-entropy} is obvious.
\begin{corollary}
    \label{cor:entropy-stability}
    Let $\ell \in \bN$ be arbitrary and let $x,x',\yp{0},\ldots,\yp{\ell-1} \in \{\T,\F\}^n$. If $\Delta(x,x') \le \f12$, then
    \[
        \lt|
            H\lt(\pi(x, \yp{0},\ldots,\yp{\ell-1})\rt)
            - H\lt(\pi(x', \yp{0},\ldots,\yp{\ell-1})\rt)
        \rt|
        \le
        H(\Delta(x,x')).
    \]
\end{corollary}
\begin{proof}
    Use Fact~\ref{fac:overlap-profile-properties}(\ref{itm:overlap-profile-property-chainrule}).
\end{proof}

\begin{proof}[Proof of Proposition~\ref{prop:impossibility-local-algs-events-relation}]
    Suppose $\Svalid$ and $\Sconc$ both hold.
    We will construct an example of the structure forbidden by $\Sindep^\ell$ or $\Sogp$.
    Since $\Svalid$ holds, there exists $\yp{0}, \ldots, \yp{\ell}$ such that for all $0\le \ell \le k$, $
    \Delta(\xp{\ell}, \yp{\ell}) \le \eta$ and $\yp{\ell}$ $\nu$-satisfies $\phip{\ell}$.

    Suppose $\ell \in [k]$ is deficient.
    Then, $\E H\lt(\pi(\xp{\ell} | \xp{0}, \ldots, \xp{\ell-1})\rt) \le \bz \f{\log k}{k}$ by definition.
    By Lemma~\ref{lem:impossibility-hamming-to-entropy},
    \[
        \lt|
            H\lt(\pi(\yp{\ell} | \xp{0}, \ldots, \xp{\ell-1})\rt)
            - H\lt(\pi(\xp{\ell} | \xp{0}, \ldots, \xp{\ell-1})\rt)
        \rt|
        \le
        H(\eta)
        \le
        \f{\bp-\bm}{8} \cdot \f{\log k}{k^2}
    \]
    for sufficiently large $k$, using the bound $H(x) \le x\log \f{e}{x}$.
    In tandem with $\Sconc$, this implies
    \[
        H\lt(\pi(\yp{\ell} | \xp{0}, \ldots, \xp{\ell-1})\rt)
        \le
        \bz \f{\log k}{k}
        +\f{\bp-\bm}{8} \cdot \f{\log k}{k^2}
        +\f{1}{\log n}
        \le
        \bp \f{\log k}{k}
    \]
    for sufficiently large $n$.
    This is an example of the structure forbidden by $\Sindep^\ell$.
    This proves part (\ref{itm:impossibility-local-algs-deficient}).

    Otherwise, suppose no $\ell \in [k]$ is deficient.
    Writing (by Fact~\ref{fac:overlap-profile-properties}(\ref{itm:overlap-profile-property-chainrule}))
    \[
        H\lt(\pi(\xp{\ell} | \xp{0}, \ldots, \xp{\ell-1})\rt)
        =
        H\lt(\pi(\xp{0}, \ldots, \xp{\ell})\rt) - H\lt(\pi(\xp{0}, \ldots, \xp{\ell-1})\rt)
    \]
    and applying Corollary~\ref{cor:entropy-stability} repeatedly, we have, for all $1\le \ell \le k$,
    \[
        \label{eq:overlap-entropy-small-move}
        \lt|
            H\lt(\pi(\xp{\ell} | \xp{0}, \ldots, \xp{\ell-1})\rt)
            - H\lt(\pi(\yp{\ell} | \yp{0}, \ldots, \yp{\ell-1})\rt)
        \rt|
        \le
        2(k+1) H(\eta)
        \le
        \f{\bp-\bm}{4} \cdot \f{\log k}{k}
    \]
    for sufficiently large $k$.
    By Lemma~\ref{lem:local-alg-e-bd} and $\Sconc$, this implies
    \[
        \lt|
            H\lt(\pi(\yp{\ell} | \yp{0}, \ldots, \yp{\ell-1})\rt)
            - \bz \f{\log k}{k}
        \rt|
        \le
        \f{\bp-\bm}{4} \cdot \f{\log k}{k} + \f{1}{\log n} + O(n^{-1/2})
        \le
        \f{\bp-\bm}{2} \cdot \f{\log k}{k}
    \]
    for sufficiently large $n$.
    So, $H\lt(\pi(\yp{\ell} | \yp{0}, \ldots, \yp{\ell-1})\rt) \in \lt[\bm \f{\log k}{k}, \bp \f{\log k}{k}\rt]$ for all $1\le \ell \le k$.
    This is an example of the structure forbidden by $\Sogp$.
    This proves part (\ref{itm:impossibility-local-algs-no-deficient}).
\end{proof}

\subsection{Lower Bound on the All-Success Probability}
\label{subsec:impossibility-local-prob-bound-valid}

Consider the random variable $\Psi = (\Phi, \varphi_E)$, for $\Phi \sim \Phi_k(n,m)$ and $\varphi_E \sim (\Omega, \P_\omega)^{\otimes E_G}$ where $(G,\rho)$ is the factor graph of $\Phi$.
In this and the next subsection, the following representation of $\Psi$ as a sequence of $km$ i.i.d. random variables will be useful.
We can reformat $\Psi = (\psi_j)_{1\le j\le km}$, where $\psi_j = (\Phi_{L(j)}, \varphi_E(e))$ and $e$ is the edge in $G$ corresponding to $\Phi_{L(j)}$.
Each $\psi_j$ is an i.i.d. sample from $\Upsilon = \unif(\cL) \times (\Omega, \P_\omega)$.

For $0\le \ell \le k$, let $\Psip{\ell} = (\phip{\ell}, \varphiep{\ell})$, which is marginally distributed as $\Psi$.
We similarly can reformat $\Psip{\ell} = (\psip{\ell}_j)_{1\le j\le km}$.

\begin{proof}[Proof of Proposition~\ref{prop:impossibility-local-prob-bounds}(\ref{itm:impossibility-local-prob-bound-valid})]
    For $\Psi$ as above, let
    \[
        f(\Psi) = \ind{\text{$\cA(\Phi, \varphi_E)$ $(\eta, \nu)$-satisfies $\Phi$}}.
    \]
    Note that $\E f(\Psi) = p$ by definition, and
    \[
        \P(\Svalid) = \E \lt[\prod_{\ell=0}^k f(\Psip{\ell})\rt].
    \]
    We wish to show this expectation is at least $p^{k+1}$.

    The constituent random variables $\psip{\ell}_j$ of $\Psip{\ell} = (\psip{\ell}_j)_{1\le j\le km}$ have the following stochastic structure.
    For each $1\le \ell \le k$, the last $km-t_\ell$ variables $(\psip{\ell}_j)_{t_\ell < j \le km}$ in $\Psip{\ell}$ are identical to the corresponding variables in $\Psip{0}$, and the first $t_\ell$ variables $(\psip{\ell}_j)_{1\le j\le t_\ell}$ are fresh i.i.d. draws from $\Upsilon$.

    Let $\tau : [k] \to [k]$ be a permutation such that $t_{\tau(1)} \le t_{\tau(2)} \le \cdots \le t_{\tau(k)}$, and let $s_\ell = t_{\tau(\ell)}$.
    Then, $\Psip{0},\ldots,\Psip{k}$ all share their last $km - s_k$ variables $\psip{\ell}_j$; all but $\Psip{\tau(k)}$ share the next $s_k - s_{k-1}$ variables; all but $\Psip{\tau(k)}$ and $\Psip{\tau(k-1)}$ share the next $s_{k-1}-s_{k-2}$ variables, and so on.

    For $0\le i \le k$, let $\xi_i$ be a sequence of $s_{i+1} - s_i$ i.i.d. draws from $\Upsilon$, where $s_0 = 0$ and $s_{k+1} = km$.
    Let $\xi_i^0, \xi_i^1, \ldots$ be a sequence of i.i.d. copies of $\xi_i$.
    By the above discussion, we can generate $\Psip{0}, \ldots, \Psip{k}$ by generating $\Psip{0} = (\xi_i^0)_{i=0}^{k}$, and for $1\le \ell \le k$, generating $\Psip{\tau(\ell)} = (\xi_i^{(\ell - i)_+})_{i=0}^k$.
    For example, when $k=3$,
    \begin{align*}
        \Psip{0}       &= (\xi_0^0, \xi_1^0, \xi_2^0, \xi_3^0), \\
        \Psip{\tau(1)} &= (\xi_0^1, \xi_1^0, \xi_2^0, \xi_3^0), \\
        \Psip{\tau(2)} &= (\xi_0^2, \xi_1^1, \xi_2^0, \xi_3^0), \\
        \Psip{\tau(3)} &= (\xi_0^3, \xi_1^2, \xi_2^1, \xi_3^0).
    \end{align*}
    Let $f(\xi_0,\ldots,\xi_k)$ denote $f(\Psi)$, for the $\Psi$ that can be formatted (by the above discussion) as $(\xi_0,\ldots,\xi_k)$.
    Let $f_0 = f$, and for $0\le d \le k$, define
    \[
        f_{d+1}(\xi_{d+1}, \ldots, \xi_k)
        =
        \E_{\xi_d}
        f_{d}(\xi_{d}, \ldots, \xi_k).
    \]
    Note that $f_{k+1}$ takes no inputs and outputs $p$.
    Further, for $0\le d \le k+1$ define
    \[
        P_d = \E \lt[
            \prod_{\ell=0}^k
            f_d\lt((\xi_i^{(\ell-i)_+})_{i=d}^k\rt)
        \rt].
    \]
    In particular $P_0 = \P(\Svalid)$ and $P_{k+1} = p^{k+1}$.
    To finish the proof we will show that $P_d \ge P_{d+1}$ for all $0\le d \le k$.
    By Jensen's inequality,
    \begin{align*}
        P_d
        &=
        \E \lt[
            f_d\lt((\xi_i^0)_{i=d}^k\rt)^{d+1}
            \prod_{\ell=d+1}^k
            f_d\lt((\xi_i^{(\ell-i)_+})_{i=d}^k\rt)
        \rt] \\
        &=
        \E \lt[
            \E_{\xi_d^0}\lt[
                f_d\lt((\xi_i^0)_{i=d}^k\rt)^{d+1}
            \rt]
            \prod_{\ell=d+1}^k
            \E_{\xi_d^{\ell-d}}\lt[
                f_d\lt((\xi_i^{(\ell-i)_+})_{i=d}^k\rt)
            \rt]
        \rt] \\
        &\ge
        \E \lt[
            \E_{\xi_d^0}\lt[
                f_d\lt((\xi_i^0)_{i=d}^k\rt)
            \rt]^{d+1}
            \prod_{\ell=d+1}^k
            \E_{\xi_d^{\ell-d}}\lt[
                f_d\lt((\xi_i^{(\ell-i)_+})_{i=d}^k\rt)
            \rt]
        \rt] \\
        &=
        \E \lt[
            f_{d+1}\lt((\xi_i^0)_{i=d+1}^k\rt)^{d+1}
            \prod_{\ell=d+1}^k
            f_{d+1}\lt((\xi_i^{(\ell-i)_+})_{i=d+1}^k\rt)
        \rt] \\
        &= P_{d+1}.
    \end{align*}
\end{proof}

\subsection{Bounded Differences and Concentration of Local Algorithms}
\label{subsec:impossibility-local-bdd-differences}

We will use the following variant of McDiarmid's inequality, which allows a bad event on which bounded differences are large.
\begin{lemma}[{\cite[Theorem 3.3]{Kut02}}]
    \label{lem:mcdiarmid-with-bad}
    Let $M\in \bN$.
    Let $\Omega_1,\ldots,\Omega_M$ be probability spaces and $\Omega = \prod_{i=1}^M \Omega_i$.
    Let $S\subset \Omega$ and $f:\Omega \to \bR$ have the following properties.
    \begin{enumerate}[label=(\roman*), ref=\roman*]
        \item If $X, X'\in S$ differ in coordinate $i$, then $|f(X) - f(X')| \le c_i$.
        \item If $X, X'\in \Omega$ differ in coordinate $i$, then $|f(X) - f(X')| \le b_i$.
    \end{enumerate}
    Then,
    \[
        \P_{X\sim \Omega}
        \lt[
            |f(X) - \E f(X)| \ge t
        \rt]
        \le
        2 \exp\lt(
            -\f{t^2}{8\sum_{i=1}^M c_i^2}
        \rt)
        + 2\P(S^c) \sum_{i=1}^M \f{b_i}{c_i}.
    \]
\end{lemma}

The following definition gives the complement of the bad event we will use.
The exponent $\f13$ is chosen to minimize the failure probability in Lemma~\ref{lem:mcdiarmid-with-bad} by balancing the two terms.
\begin{definition}
    \label{defn:r-locally-small}
    A $k$-SAT formula $\Phi\in \Omega_k(n,m)$ is \emph{$r$-locally small} if, for $(G,\rho)$ the factor graph of $\Phi$, $|N_r(v,G)|\le n^{1/3}$ for all $v\in \Va_G$.
\end{definition}

\begin{fact}
    \label{fac:r-locally-small}
    If $\Phi\sim \Phi_k(n,m)$ and $r\in \bN$ is constant, $\Phi$ is $r$-locally small with probability $1-\exp(-\Omega(n^{1/3}))$.
\end{fact}
\begin{proof}
    This follows from Lemma~\ref{lem:simulation-dfg-nbd-growth-rate} and a union bound.
    Note that the $r$-neighborhood of any $v\in \Cl_G$ is contained in the $(r+1)$-neighborhood of any of its neighbors.
\end{proof}

\begin{proof}[Proof of Proposition~\ref{prop:impossibility-local-prob-bounds}(\ref{itm:impossibility-local-prob-bound-conc})]
    We present the argument for $\ell = k$; showing concentration for the other conditional overlap entropies is similar.
    The random variable $Y = H\lt(\pi(\xp{k} | \xp{0}, \ldots, \xp{k-1})\rt)$ is measurable in $\Psip{0},\ldots,\Psip{k}$.
    For each $0\le \ell \le k$, we can write $\Psip{\ell} = (\psip{\ell}_j)_{1\le j\le km}$.
    The constituent random variables $\psip{\ell}_j$ can be partitioned into equivalence classes, where variables in the same equivalence class are identical and different equivalence classes are mutually independent.
    Let $\zeta = (\psi_1,\ldots,\psi_M)$ contain one representative from each equivalence class.
    Note that $\psi_1,\ldots,\psi_M$ i.i.d. samples from $\Upsilon$ and $M \le k^2m$.

    All the $\phip{\ell}$ $(\Gp{\ell},\rhop{\ell},\varphiep{\ell})$, $\xp{\ell}$, overlap profiles of the $\xp{\ell}$, and $Y$ are $\zeta$-measurable.
    We will use $Y(\zeta)$ to denote the $Y$ given by this realization of $\zeta$, and similarly for the remaining random variables.

    Let $S\subseteq \Upsilon^M$ be the event that $\phip{\ell}(\zeta)$ is $r$-locally small for all $0\le \ell \le k$.
    By a union bound on Fact~\ref{fac:r-locally-small}, $\P(S^c) \le \exp(-\Omega(n^{1/3}))$.

    Suppose $\zeta,\zeta'\in S$ differ in only one coordinate $\psi_i$.
    For each $0\le \ell \le k$, the decorated factor graphs $(\Gp{\ell},\rhop{\ell},\varphiep{\ell})(\zeta)$ and $(\Gp{\ell},\rhop{\ell},\varphiep{\ell})(\zeta')$ differ in at most one edge.
    Because $\zeta,\zeta'\in S$ and $\cA$ is local, $\xp{\ell}(\zeta)$ and $\xp{\ell}(\zeta')$ differ in $O(n^{1/3})$ bits.
    So, corresponding entries in $\pi(\xp{0}, \ldots, \xp{k})(\zeta)$ and $\pi(\xp{0}, \ldots, \xp{k})(\zeta')$ differ by $O(n^{-2/3})$.
    Thus, $|Y(\zeta)-Y(\zeta')| \le O(n^{-2/3} \log n)$.

    Moreover, for any $\zeta,\zeta'\in \Upsilon^M$, we have $|Y(\zeta)-Y(\zeta')| \le \log 2$ because the conditional overlap entropy attains values in $[0, \log 2]$.
    By Lemma~\ref{lem:mcdiarmid-with-bad},
    \begin{align*}
        \P \lt[|Y-\E Y| \ge \f{1}{\log n}\rt]
        &\le
        2 \exp\lt(-\f{1/\log^2 n}{8k^2m O((n^{-2/3} \log n)^2)}\rt)
        +
        \exp(-\Omega(n^{1/3}))
        O\lt(\f{2k^2m\log 2}{n^{-2/3} \log n}\rt) \\
        &\le
        \exp(-\tOmega(n^{1/3})).
    \end{align*}
\end{proof}

Similar ideas prove Lemma~\ref{lem:local-alg-e-bd}.

\begin{proof}[Proof of Lemma~\ref{lem:local-alg-e-bd}]
    The lower bound follows from the definition of deficient.
    For $1\le t\le T$, let $\xp{\ell, t} = \cA(\phip{\ell, t}, \varphiep{\ell, t})$.
    We will show that
    \[
        \lt|
            \E H(\pi(\xp{\ell, t} | \xp{0}, \ldots, \xp{\ell-1}))
            -\E H(\pi(\xp{\ell, t-1} | \xp{0}, \ldots, \xp{\ell-1}))
        \rt|
        \le
        O(n^{-1/2}).
    \]
    Since $\E H(\pi(\xp{\ell, t_\ell-1} | \xp{0}, \ldots, \xp{\ell-1})) < \bz \f{\log k}{k}$, the above inequality implies the result.

    Let $S$ be the event that $\phip{\ell, t-1}$ and $\phip{\ell, t}$ are both $r$-locally small.
    By Fact~\ref{fac:r-locally-small} and a union bound, $\P(S^c) \le \exp(-\Omega(n^{1/3}))$.
    The decorated factor graphs $(\Gp{\ell,t-1}, \rhop{\ell,t-1}, \varphiep{\ell,t-1})$ and $(\Gp{\ell,t}, \rhop{\ell,t}, \varphiep{\ell,t})$ differ in at most one edge.
    On the event $S$, $\xp{\ell, t}$ and $\xp{\ell, t-1}$ differ in at most $O(n^{1/3})$ bits, and by Lemma~\ref{lem:impossibility-hamming-to-entropy},
    \[
        \lt|
            H(\pi(\xp{\ell, t} | \xp{0}, \ldots, \xp{\ell-1}))
            -H(\pi(\xp{\ell, t-1} | \xp{0}, \ldots, \xp{\ell-1}))
        \rt|
        \le
        H(\Delta(\xp{\ell, t}, \xp{\ell, t-1}))
        \le
        O(n^{-2/3} \log n).
    \]
    Moreover, this difference is always at most $\log 2$.
    Thus
    \begin{align*}
        \lt|
            \E H(\pi(\xp{\ell, t} | \xp{0}, \ldots, \xp{\ell-1}))
            -\E H(\pi(\xp{\ell, t-1} | \xp{0}, \ldots, \xp{\ell-1}))
        \rt|
        &\le
        O(n^{-2/3} \log n)
        +\P(S^c) \log 2 \\
        &\le
        O(n^{-1/2}).
    \end{align*}
\end{proof}

%% file: tex/8-simulation.tex
\section{Simulation of Local Memory Algorithms}
\label{sec:simulation}

In this section, we introduce the class of \emph{local memory algorithms}.
These algorithms are a natural generalization of local algorithms, which make local decisions in series (in a random vertex order) and allow earlier decisions to leave local information that later decisions can see.
This class includes the first phase of \verb|Fix|, as well as the sequential local algorithms considered in \cite{GS17}.
We show, somewhat surprisingly, that any local memory algorithm can be simulated by a local algorithm of larger radius.
We then show that any local algorithm can be simulated by a constant degree polynomial.

The main results of this section are the following two propositions.
Throughout this section, fix arbitrary $\alpha = \alpha(k)$ independent of $n$ and let $m = \lfloor \alpha n \rfloor$.
\begin{proposition}[Local algorithms simulate local memory algorithms]
    \label{prop:local-memory-to-local}
    Suppose $\alpha k, k-1 \ge 2$ and $\eta > 0$.
    Let $\cA$ be an $r$-local memory algorithm (defined in Definition~\ref{defn:local-memory}) with output in $\{\T,\F\}^n$.
    There exists $R\in \bN$ depending on $\alpha, k, r, \eta$ and an $R$-local algorithm $\cA'$ such that, for some coupling of the internal randomnesses of $\cA, \cA'$,
    \[
        \P\lt[
            \Delta(\cA(\Phi), \cA'(\Phi)) \ge \eta
        \rt]
        \le \exp(-\Omega(n^{1/3})),
    \]
    where the probability is over $\Phi \sim \Phi_k(n,m)$ and the randomnesses of $\cA, \cA'$.
\end{proposition}
We parse the outputs of a low degree polynomial with the function $\strictround : \bR \to \{\T,\F,\Q\}$, defined by
\[
    \strictround(x) =
    \begin{cases}
        \T & x = 1, \\
        \F & x = -1, \\
        \Q & \text{otherwise}.
    \end{cases}
\]
When applied to a real-valued vector, $\strictround$ is applied coordinate-wise.
Note that this is a more stringent parsing scheme than $\round$.
Let $N = m\cdot k \cdot 2n$.
Recall that each $\Phi \in \Omega_k(n,m)$ can be identified with a vector in $\{0,1\}^N$, as described below Definition~\ref{defn:ldp}.
\begin{proposition}[Low degree polynomials simulate local algorithms]
    \label{prop:local-to-ldp}
    Suppose $\alpha k, k-1 \ge 2$ and $\eta > 0$.
    Let $\cA$ be an $r$-local algorithm with output in $\{\T,\F\}^n$.
    There exist $D, \gamma > 0$ depending on $\alpha, k, r, \eta$ and a (random) degree-$D$ polynomial $f : \bR^N \times \Omega \to \bR^n$ such that the following holds.
    Let $\cA' = \strictround \circ f$.
    For some coupling of the internal randomnesses of $\cA$ and $f$,
    \[
        \P\lt[
            \Delta(\cA(\Phi), \cA'(\Phi)) \ge \eta
        \rt]
        \le \exp(-\Omega(n^{1/3})),
    \]
    where the probability is over $\Phi \sim \Phi_k(n,m)$ and the randomnesses of $\cA, \cA'$.
    Moreover, $\E_{\Phi, \omega} \norm{f(\Phi, \omega)}_2^2 \le \gamma n$.
\end{proposition}

Both simulation results incur an error tolerance $\eta$ independent of $n$ which can be made arbitrarily small in $k$ and fail with probability only $\exp(-\Omega(n^{1/3}))$.

These simulation results imply that our hardness theorems, Theorems~\ref{thm:impossibility} and \ref{thm:impossibility-local-algs}, apply to any local memory algorithm.
We will also use these results in Section~\ref{sec:achievability} with the fact that the first phase of \verb|Fix| is a local memory algorithm to show that local algorithms and low degree polynomials solve random $k$-SAT at clause density $\alpha = (1-\eps) 2^k \log k / k$.

This section is structured as follows.
In Subsection~\ref{subsec:simulation-local-algs} we review properties of local algorithms and the $k$-SAT factor graph.
In Subsection~\ref{subsec:simulation-local-memory-algs} we define local memory algorithms.
In Subsection~\ref{subsec:simulation-local-memory-to-local} we prove Proposition~\ref{prop:local-memory-to-local}, and in Subsection~\ref{subsec:simulation-local-to-ldp} we prove Proposition~\ref{prop:local-to-ldp}.
Subsection~\ref{subsec:simulation-deferred} contains deferred proofs.

\subsection{Properties of Local Algorithms and the $k$-SAT Factor Graph}
\label{subsec:simulation-local-algs}

Throughout this section, fix a probability space $(\Omega, \P_\omega)$.
Let $\DFG(n,m,k,(\Omega, \P_\omega))$ denote the law of the decorated random $k$-SAT factor graph $(G, \rho, \varphi)$, where $\Phi \sim \Phi_k(n,m)$, $(G,\rho)$ is the factor graph of $\Phi$, and $\varphi \sim (\Omega, \P_\omega)^{\otimes (V_G \cup E_G)}$.
We write this as $\DFG(n,m,k)$ when $(\Omega, \P_\omega)$ is unambiguous.

Equivalently, $(G, \rho, \varphi) \sim \DFG(n,m,k)$ can be sampled as follows.
$\Va_G = \{v_1,\ldots,v_n\}$ and $\Cl_G = \{c_1,\ldots,c_m\}$ are fixed.
$E_G$ consists of $k$ edges from each $c\in \Cl_G$ to i.i.d. uniformly random vertices in $\Va_G$, and $\rho, \varphi$ are sampled by $\rho \sim \unif(\{\T,\F\})^{\otimes E_G}$, $\varphi \sim (\Omega, \P_\omega)^{\otimes (V_G\cup E_G)}$.

A (possibly infinite) graph is \emph{locally finite} if every vertex has finite degree.
The formalism in Definitions~\ref{defn:rooted-decorated-bipartite-graph}, \ref{defn:r-nbd}, and \ref{defn:local-fn} applies verbatim to locally finite $G$.
The local geometry of a sample from $\DFG(n,m,k)$ can be understood in analogy to the following locally finite tree.

\begin{definition}[Decorated Alternating Galton-Watson Tree]
    Let $d_1 > 0$, $d_2 \in \bN$.
    Let $\DGW(d_1, d_2, (\Omega, \P_\omega))$ denote the law of the following rooted decorated tree $(o, T, \rho, \varphi)$.
    The rooted tree $(o, T)$ is sampled by the following procedure.
    \begin{itemize}
        \item Start with a root vertex $o$ in layer $0$.
        \item For $\ell \ge 1$:
        \begin{itemize}
             \item If $\ell$ is even, each vertex in layer $\ell$ independently spawns $\Pois(d_1)$ children in layer $\ell+1$.
             \item If $\ell$ is odd, each vertex in layer $\ell$ spawns $d_2$ children in layer $\ell+1$.
        \end{itemize}
        Each non-root vertex is connected to its parent by an edge.
    \end{itemize}
    Let $\Va_T$ and $\Cl_T$ be the sets of even and odd depth vertices of $T$.
    Further, let $V_T = \Va_T \cup \Cl_T$ and let $E_T$ be the edge set of $T$.
    Sample $\rho \sim \unif(\{\T,\F\})^{\otimes E_T}$ and $\varphi \sim (\Omega, \P_\omega)^{\otimes (V_T\cup E_T)}$.
\end{definition}

When $(\Omega, \P_\omega)$ is unambiguous, we write this as $\DGW(d_1,d_2)$.
The significance of this tree is that as $n\to\infty$, local neighborhoods $N_r(v, G, \rho, \varphi)$ of a sample $(G,\rho,\varphi) \sim \DFG(n,m,k)$, where $v\in \Va_G$ is fixed, converge weakly to local neighborhoods of the root of $\DGW(\alpha k, k-1)$.
This is analogous to the fact that local neighborhoods of the sparse Erd\H{o}s-R\'enyi graph $G(n, d/n)$ converge weakly to local neighborhoods of the root of the Poisson Galton-Watson tree $\PGW(d)$.

We now state several lemmas pertaining to local geometry of samples from $\DGW(d_1,d_2)$ and $\DFG(n,m,k)$.
Lemmas~\ref{lem:simulation-dgw-nbd-growth-rate} and \ref{lem:simulation-dfg-nbd-growth-rate} control the local neighborhood sizes of the root of $\DGW(d_1,d_2)$ and of a  left-vertex in $\DFG(n,m,k)$.
Lemma~\ref{lem:simulation-dfg-to-dgw} makes precise the sense in which local neighborhoods of left-vertices of $\DFG(n, m, k)$ converge to local neighborhoods of the root of $\DGW(\alpha k, k-1)$.
Lemma~\ref{lem:simulation-dfg-concentration} shows concentration for the sum of a local function.
These lemmas are analogous to \cite[Lemma 11.1, Lemma 11.2, Lemma 12.4, Proposition 12.3]{BCN20}, which give the analogous results with $\cG(n,m,k)$ and $\DGW(\alpha k, k-1)$ replaced by $G(n,d/n)$ and $\PGW(d)$ (and without the decorations $\rho, \varphi$, which do not affect the results).
We omit their proofs, which are easily adapted from the corresponding proofs of \cite{BCN20}.

\begin{lemma}
    \label{lem:simulation-dgw-nbd-growth-rate}
    Let $d_1, d_2 \ge 2$ and $(o, T, \rho, \varphi) \sim \DGW(d_1, d_2)$.
    There are universal constants $c_0, c_1 > 0$ such that for all $\lambda > 0$,
    \[
        \P \lt[
            \text{$|N_{2r}(o,T)| \le \lambda (d_1d_2)^r$ for all positive integers $r$}
        \rt]
        \ge
        1 - c_1 e^{-c_0 \lambda}.
    \]
\end{lemma}

\begin{lemma}
    \label{lem:simulation-dfg-nbd-growth-rate}
    Let $\alpha k, k-1 \ge 2$.
    Let $(G, \rho, \varphi) \sim \DFG(n,m,k)$, and let $v\in \Va_G$ be fixed.
    There are universal constants $c_0, c_1 > 0$ such that for all $\lambda > 0$,
    \[
        \P \lt[
            \text{$|N_{2r}(v,G)| \le \lambda (\alpha k(k-1))^r$ for all positive integers $r$}
        \rt]
        \ge
        1 - c_1 e^{-c_0 \lambda}.
    \]
\end{lemma}

Recall that $\Lambda$ is the set of (possibly infinite, locally finite) rooted decorated bipartite graphs.

\begin{lemma}
    \label{lem:simulation-dfg-to-dgw}
    Let $\alpha k, k-1 \ge 2$.
    Let $(o, T, \rho, \varphi) \sim \DGW(\alpha k, k-1)$, $(G, \rho', \varphi') \sim \DFG(n,m,k)$, and let $v\in \Va_G$ be fixed.
    Let $f : \Lambda \to [-1,1]$ be a $2r$-local function.
    There exists $c>0$ (depending on $\alpha,k,r$) such that for all $n$,
    \[
        \lt|
            \E f(o,T,\rho,\varphi)
            -\E f(v,G,\rho',\varphi')
        \rt|
        \le
        \f{c\log n}{n^{1/2}}.
    \]
\end{lemma}

\begin{lemma}
    \label{lem:simulation-dfg-concentration}
    Let $\alpha k, k-1 \ge 2$, and let $(G, \rho, \varphi) \sim \DFG(n,m,k)$.
    Let $f : \Lambda \to [-1, 1]$ be a $2r$-local function.
    There exists $c>0$ (depending on $\alpha,k,r$) such that for all $p\ge 2$,
    \[
        \E \lt[\lt|
            \sum_{v\in \Va_G}
            f(v,G,\rho,\varphi)
            -
            \E
            \sum_{v\in \Va_G}
            f(v,G,\rho,\varphi)
        \rt|^p\rt]
        \le
        \lt(c n^{1/2} p^{3/2}\rt)^p.
    \]
\end{lemma}

We can translate Lemma~\ref{lem:simulation-dfg-concentration}, into the following tail bound for sums of local functions.
\begin{corollary}
    \label{cor:simulation-dfg-concentration}
    Let $\alpha k, k-1 \ge 2$, and let $(G, \rho, \varphi) \sim \DFG(n,m,k)$.
    Let $f : \Lambda \to [-1, 1]$ be a $2r$-local function.
    There exists $c>0$ (depending on $\alpha,k,r$) such that for all $t \ge cn^{1/2}$,
    \[
        \P \lt[ \lt|
            \sum_{v\in \Va_G}
            f(v,G,\rho,\varphi)
            -
            \E
            \sum_{v\in \Va_G}
            f(v,G,\rho,\varphi)
        \rt| \ge t \rt]
        \le
        \exp\lt(
            -\f{t^{2/3}}{cn^{1/3}}
        \rt).
    \]
\end{corollary}
\begin{proof}
    Let $c$ be as in Lemma~\ref{lem:simulation-dfg-concentration}, and suppose $t \ge 2^{3/2}ecn^{1/2}$.
    Set $p = \lt(\f{t}{ecn^{1/2}}\rt)^{2/3} \ge 2$, so by Lemma~\ref{lem:simulation-dfg-concentration},
    \begin{align*}
        \P \lt[ \lt|
            \sum_{v\in \Va_G}
            f(v,G,\rho,\varphi)
            -
            \E
            \sum_{v\in \Va_G}
            f(v,G,\rho,\varphi)
        \rt| \ge t \rt]
        &\le
        t^{-p}
        \E \lt[\lt|
            \sum_{v\in \Va_G}
            f(v,G,\rho,\varphi)
            -
            \E
            \sum_{v\in \Va_G}
            f(v,G,\rho,\varphi)
        \rt|^p\rt] \\
        &\le
        \lt(\f{cn^{1/2}p^{3/2}}{t}\rt)^p = \exp(-p) \\
        &=
        \exp\lt(
            -\f{t^{2/3}}{(ec)^{2/3}n^{1/3}}
        \rt).
    \end{align*}
    The result follows by adjusting the constant $c$.
\end{proof}

\subsection{Local Memory Algorithms}
\label{subsec:simulation-local-memory-algs}

We now define local memory algorithms.
In addition to the usual features of a local algorithm, these algorithms have access to a mutable memory map $\mu : V_G \to \bZ_{\ge 0}$, which we think of as an unlimited notepad on each variable.
The algorithm processes vertices $v\in V_G$ (both variables and clauses) in a uniformly random order.
Each step, the algorithm accesses the $r$-local neighborhood of a vertex and can overwrite the data written on any vertex in that neighborhood.
In the end, each variable $v\in \Va_G$ decides to be true or false depending on the final value $\mu(v)$ on its notepad.

To formalize this algorithm class, we will define memory-augmented versions of Definitions~\ref{defn:rooted-decorated-bipartite-graph}, \ref{defn:r-nbd}, and \ref{defn:local-fn}.
\begin{definition}[Rooted memory-augmented decorated bipartite graph]
    A \emph{memory-augmented decorated bipartite graph} is a tuple $(G, \rho, \varphi, \mu)$, where $(G, \rho, \varphi)$ is a decorated bipartite graph and $\mu$ is a function $\mu : V_G \to \bZ_{\ge 0}$.
    A \emph{rooted memory-augmented decorated bipartite graph} is a tuple $(v, G, \rho, \varphi, \mu)$, where $(G, \rho, \varphi, \mu)$ is a memory-augmented decorated bipartite graph and $v\in V_G$.
\end{definition}

Let $\tLambda$ denote the set of rooted memory-augmented decorated bipartite graphs.
Two such graphs are isomorphic of there exists a bijection between them preserving $v, \Va_G, \Cl_G, E_G, \rho, \varphi, \mu$.

\begin{definition}[$r$-neighborhood]
    Let $(v, G, \rho, \varphi, \mu) \in \tLambda$ and $r\in \bN$.
    Define $N_r(v, G, \rho, \varphi, \mu)$ to be $(v, G', \rho', \varphi', \mu') \in \tLambda$, where $(v, G', \rho', \varphi') = N_r(v, G, \rho, \varphi)$ and $\mu' = \mu \big|_{G'}$ is the restriction of $\mu$ to $G'$.
\end{definition}

\begin{definition}[$r$-local subroutine]
    An algorithm $f$ with input space $\tLambda$ is an \emph{$r$-local subroutine} if the execution of $f(v, G, \rho, \varphi)$ depends only on the isomorphism class of $N_r(v, G, \rho, \varphi, \mu) = (v, G', \rho', \varphi', \mu')$, and $f$ interacts with its input by editing the outputs of $\mu'$.
\end{definition}

We are now ready to define a local memory algorithm.
In the following definition, $\psi$ is an auxiliary random variable on each vertex that determines the order in which vertices are processed.

\begin{definition}[$r$-local memory algorithm]
    \label{defn:local-memory}
    Let $f_1$ be an $r$-local subroutine and $f_2 : \bZ_{\ge 0} \to \{\T,\F\}$ be a function.
    The $r$-local memory algorithm based on $(f_1, f_2)$, denoted $\cA_{f_1, f_2}$, runs as follows on input $\Phi \in \Omega_k(n,m)$ with factor graph $(G, \rho)$.
    \begin{enumerate}[label=(\arabic*), ref=\arabic*]
        \item Initialize $\mu : V_G \to \bZ_{\ge 0}$ to the all-$0$ map.
        Sample $\varphi \sim (\Omega, \P_\omega)^{\otimes (V_G \cup E_G)}$ and $\psi \sim \unif([0,1])^{\otimes V_G}$.
        \item Loop through vertices $v\in V_G$ (both variables and clauses) in increasing order of $\psi(v)$.
        For each $v$, run $f_1(v, G, \rho, \varphi, \mu)$.
        \item Output $x\in \{\T,\F\}^n$ where $x_i = f_2(\mu(v_i))$.
    \end{enumerate}
\end{definition}

We will see (Fact~\ref{fac:fix-local-memory}) that the first phase of \verb|Fix| is in this class.
The following variant of the sequential local algorithms in \cite{GS17} is also in this class.
\begin{definition}[Sequential $r$-local algorithm]
    \label{defn:seq-local}
    Let $f : \Lambda \to [0,1]$ be an $r$-local function.
    The sequential $r$-local algorithm based on $f$, denoted $\cB_f$, runs as follows on input $\Phi \in \Omega_k(n,m)$ with factor graph $(G, \rho)$.
    \begin{enumerate}[label=(\arabic*), ref=\arabic*]
        \item Sample $\varphi \sim (\Omega, \P_\omega)^{\otimes (V_G \cup E_G)}$ and $\psi \sim \unif([0,1])^{\otimes \Va_G}$.
        \item Loop through $v\in \Va_G$ in increasing order of $\psi(v)$.
        For each $v = v_i$:
        \begin{enumerate}[label=(\alph*), ref=\alph*]
            \item Compute $p = f(v, G, \rho, \varphi)$. Set $x_i = \T$ with probability $p$, and otherwise $x_i = \F$.
            \item Simplify $\Phi$ by deleting clauses satisfied by this setting of $x_i$ and appearances of $x_i$ in clauses not satisfied by this setting. Furthermore, delete any clauses that become empty (thus not satisfied) as a result of the latter operation.
            \item Let $G'$ be the corresponding simplification of $G$, and let $\rho' = \rho \big|_{G'}$ and $\varphi' = \varphi \big|_{G'}$.
            \item Set $(G, \rho, \varphi) \leftarrow (G', \rho', \varphi')$.
        \end{enumerate}
        \item Output $(x_1,\ldots,x_n) \in \{\T,\F\}^n$.
    \end{enumerate}
\end{definition}

\begin{fact}
    \label{fac:seq-local-local-memory}
    For any $r\in \bN$, a sequential $r$-local algorithm can simulated by a $\max(r,2)$-local memory algorithm.
\end{fact}
\begin{proof}
    Let $\cB_f$ be a sequential $r$-local algorithm, whose randomness is sampled i.i.d. from $(\Omega, \P_\omega)$.
    We will construct an $r$-local memory algorithm $\cA_{f_1, f_2}$ simulating $\cB_f$.

    This algorithm maintains the invariant that for $v = v_i\in \Va_G$, $\mu(v)=0$ if $x_i$ is not yet set, $1$ if $x_i$ is set true, and $2$ if $x_i$ is set false.
    For clause vertices $c \in \Cl_G$, $\mu(c)=1$ if the clause corresponding to $c$ has been deleted in the simplification, and otherwise $\mu(c)=0$.

    Thus, $\cA_{f_1, f_2}$ uses randomness sampled from $(\Omega, \P_\omega) \times \unif([0,1])$.
    That is, its internal randomness is $\varphist = (\varphi, q)$, which is sampled by $\varphi \sim (\Omega, \P_\omega)^{\otimes (V_G \cup E_G)}$ and $q \sim \unif([0,1])^{\otimes (V_G \cup E_G)}$.

    The $r$-local subroutine $f_1$ runs as follows on input $(v, G, \rho, \varphist, \mu)$.
    If $v\in \Cl_G$, do nothing.
    Note that the remaining loop over $v\in \Va_G$ runs over these vertices in a uniformly random order, as desired.
    If $v=v_i\in \Va_G$, let $G'$ be the simplification of $G$ determined by the information recorded in $\mu$, and let $\rho' = \rho \big|_{G'}$, $\varphi' = \varphi\big|_{G'}$.
    We can simulate the computation of $p = f(v, G', \rho', \varphi')$ because simplification only deletes vertices and edges, so any $r$-local decision in the simplified factor graph is still $r$-local in the simulation.
    We then set $x_i = \T$ if $p < q(v)$, and otherwise $x_i = \F$.
    We update $\mu$ to record this value of $x_i$ and any clause simplifications that result (which is a $2$-local operation).

    At the end of the algorithm, $\mu(v) \in \{1,2\}$ for all $v\in \Va_G$.
    Let $f_2(x) = \T$ if $x=1$ and $\F$ if $x=2$.
\end{proof}

Definition~\ref{defn:seq-local} differs slightly from the presentation in \cite{GS17} in the following way.
\cite{GS17} studies NAE-$k$-SAT, in which a clause is satisfied if it contains at least one true and false literal.
In partially simplfied formulas of this problem, clauses can exist in four states: ``removed," ``already contains true," ``already contains false," and ``contains neither true nor false," and the sequential local algorithms of \cite{GS17} track this information.
Of course, we can just as well simulate this by a local memory algorithm by having $\mu$ track these clause states.

\subsection{Local Algorithms Simulate Local Memory Algorithms}
\label{subsec:simulation-local-memory-to-local}

In this subsection, we prove Proposition~\ref{prop:local-memory-to-local}, that any local memory algorithm can be simulated by a local algorithm of larger (but still constant) radius.

The simulation is the natural one: we expand $\varphi$ to also generate the auxiliary randomness $\psi$ determining the vertex order, and then determine the output at each $v\in \Va_G$ by simulating the local memory algorithm on the $R$-neighborhood of $v$.
Formally, we expand $\varphi$ to $\varphist$, whose outputs are sampled from $(\Omega, \P_\omega) \times \unif([0,1])$.
We collect the first coordinates of the outputs into $\varphi$ and the second coordinates into $\psi$.
(This generates $\psi : V_G \cup E_G \to [0,1]$, and we ignore $\psi \big|_{E_G}$.)

Because sequentiality usually does not create long dependence chains, this simulation will often faithfully capture the local memory algorithm's behavior.

\begin{definition}[$R$-local simulation]
    \label{defn:local-sim}
    Let $\cA_{f_1, f_2}$ be an $r$-local memory algorithm, with i.i.d. internal randomness from $(\Omega, \P_\omega)$.
    For $R\in \bN$, the $R$-local simulation of $\cA_{f_1, f_2}$ is the $R$-local algorithm $\cA_f$ that runs as follows on input $\Phi \in \Omega_k(n,m)$ with factor graph $(G, \rho)$.
    \begin{enumerate}[label=(\arabic*), ref=\arabic*]
        \item Sample $\varphist = (\varphi, \psi)$, where $\varphi \sim (\Omega, \P_\omega)^{\otimes (V_G \cup E_G)}$ and $\psi \sim \unif([0,1])^{\otimes (V_G \cup E_G)}$.
        \item For each $v = v_i \in \Va_G$, set $x_i = f(v,G,\rho,\varphist)$.
        Here $f(v,G,\rho,\varphist)$ is the following $R$-local function.
        \begin{enumerate}[label=(\alph*), ref=\alph*]
            \item Let $N_R(v, G, \rho, \varphist) = (v, G', \rho', \varphist')$.
            Let $\varphist' = (\varphi', \psi')$, where $\varphi' = \varphi \big|_{G'}$ and $\psi' = \psi \big|_{G'}$.
            \item Initialize $\mu : V_{G'} \to \bZ_{\ge 0}$ to the all-0 map.
            \item For $u\in V_{G'}$ in increasing order of $\psi'(u)$, run $f_1(u, G', \rho', \varphi', \mu)$.
            \item Output $f(v,G,\rho,\varphist) = f_2(\mu(v))$.
        \end{enumerate}
        \item Output $(x_1,\ldots,x_n)$.
    \end{enumerate}
\end{definition}

The main idea of the proof of Proposition~\ref{prop:local-memory-to-local} is that dependencies caused by sequentiality all arise from the following structure.

\begin{definition}[$r$-hop $\psi$-dependence chain]
    Let $G$ be a locally finite graph and $\psi : V_G \to [0,1]$ be a function.
    Let $r\in \bN$.
    A sequence $v_1,v_2,\ldots,v_s \in V_G$ is an \emph{$r$-hop $\psi$-dependence chain} if consecutive vertices in the sequence are at most distance $r$ apart and $\psi(v_1),\psi(v_2),\ldots,\psi(v_s)$ is decreasing.
\end{definition}

We can now define a notion of insulation in terms of these dependence chains.
The key point of the following definition is that if in the $R$-local simulation in Definition~\ref{defn:local-sim}, $v\in \Va_G$ is $(r,R,\psi)$-insulated, then the $R$-local simulation's output at $v$ must match that of the local memory algorithm run with the same $\varphi, \psi$.

\begin{definition}[$(r,R,\psi)$-insulated]
    Let $G$ be a locally finite graph and $\psi : V_G \to [0,1]$ be a function.
    Let $v\in V_G$ and $r,R\in \bN$ with $R\ge 2r$.
    $v$ is \emph{$(r,R,\psi)$-insulated} if there is no $2r$-hop $\psi$-dependence chain $v_1,v_2,\ldots,v_s\in V_G$ with $v=v_1$ and $v_s \in N_R(v,G) \setminus N_{R-2r}(v,G)$.
\end{definition}

In Definition~\ref{defn:local-sim}, if $\Phi \sim \Phi_k(n,m)$, then $(G, \rho, \varphist)$ is a sample from the decorated $k$-SAT factor graph $\DFG(n,m,k,(\Omega,\P_\omega)\times \unif([0,1]))$.
To prove Proposition~\ref{prop:local-memory-to-local}, it suffices to upper bound the fraction of $v\in \Va_G$ that are not $(r,R,\psi)$-insulated.
To achieve this, we will control the probability that the root of $\DGW(d_1,d_2,(\Omega,\P_\omega)\times \unif([0,1]))$ is not $(r,R,\psi)$-insulated.
Then, because $(r,R,\psi)$-insulatedness is an $R$-local property, we can translate this bound to the $k$-SAT factor graph by the machinery of Lemma~\ref{lem:simulation-dfg-to-dgw} and Corollary~\ref{cor:simulation-dfg-concentration}.

\begin{proposition}
    \label{prop:dgw-root-insulated}
    Let $d_1,d_2\ge 2$, $r\in \bN$, and $\eta\in (0,1)$.
    Let $(o,T,\rho,\varphist) \sim \DGW(d_1,d_2,(\Omega,\P_\omega)\times \unif([0,1]))$, and write $\varphist = (\varphi, \psi)$ for $\varphi : V_T \cup E_T \to \Omega$ and $\psi : V_T \cup E_T \to [0,1]$.
    There exists $R$ dependent on $d_1,d_2,r,\eta$ such that
    \[
        \P\lt[\text{$o$ is $(r,R,\psi)$-insulated in $T$}\rt] \ge 1-\eta.
    \]
\end{proposition}

The proof of this proposition relies on the following technical lemma, whose proof we defer to Subsection~\ref{subsec:simulation-deferred}.
\begin{lemma}
    \label{lem:simulation-max-nbd}
    Let $d_1, d_2 \ge 2$ and $(o, T, \rho, \varphi) \sim \DGW(d_1, d_2)$.
    For any $r\in \bN$ and $\eta \in (0,1)$, there exist $C, \Rst>0$ depending on $d_1, d_2, r, \eta$ such that for all integers $R\ge
    \Rst$,
    \[
        \max_{v\in N_{R}(o, T)}
        |N_{2r}(v, T)|
        \le
        \f{CR}{\log^{(r+1)}R}
    \]
    with probability at least $1-\eta$.
    Here, $\log^{(r+1)}$ denotes the $(r+1)$th iterate of $\log$.
\end{lemma}

\begin{proof}[Proof of Proposition~\ref{prop:dgw-root-insulated}]
    Let $R\in \bN$ be a number we will determine later.
    Lemma~\ref{lem:simulation-max-nbd} gives $\Rst$ such that if $R \ge \Rst$, then the conclusion of Lemma~\ref{lem:simulation-max-nbd} holds with probability at least $1-\eta/2$.
    Consider a realization of $T, \rho, \varphi$ such that this event holds.
    We will control the probability over $\psi$ that $o$ is not $(r,R,\psi)$-insulated in $T$.

    If $o$ is not $(r,R,\psi)$-insulated, there exists a $2r$-hop $\psi$-dependence chain $o=v_1, v_2, \ldots, v_s \in V_T$ where $v_s \in N_R(o,T) \setminus N_{R-2r}(o,T)$.
    By taking an initial subsequence, we get a $2r$-hop $\psi$-dependence chain $o=v_1, v_2, \ldots, v_t \in V_T$ of length $t = \lceil \f{R}{2r}\rceil$.
    By Markov's inequality,
    \[
        \P \lt[\text{$o$ is not $(r,R,\psi)$-insulated in $T$}\rt]
        \le
        \E \# \lt(\text{$2r$-hop $\psi$-dependence chains $o=v_1, v_2, \ldots, v_t \in V_T$}\rt).
    \]
    The last expectation is bounded as follows.
    By Lemma~\ref{lem:simulation-max-nbd}, there are at most $\lt(\f{CR}{\log^{(r+1)}R}\rt)^t$ sequences $o=v_1,v_2,\ldots,v_t$ with consecutive vertices at most distance $2r$ apart, and for each one, $\psi(v_1),\psi(v_2),\ldots,\psi(v_t)$ is decreasing with probability $\f{1}{t!}$.
    So (using $t! \ge (t/e)^t$) the last expectation is at most
    \begin{equation}
        \label{eq:super-inefficient-power-tower}
        \f{1}{t!} \lt(\f{CR}{\log^{(r+1)}R}\rt)^t
        \le
        \lt(\f{eCR}{t \log^{(r+1)}R}\rt)^t
        \le
        \lt(\f{2eCr}{\log^{(r+1)}R}\rt)^t
        \le \eta/2
    \end{equation}
    for a large enough choice of $R$.
    Thus, over the randomness of $\psi$,
    \[
        \P \lt[\text{$o$ is $(r,R,\psi)$-insulated in $T$}\rt]
        \ge
        1-\eta/2.
    \]
    The result follows by a union bound.
\end{proof}

Unfortunately, due to the last inequality in \eqref{eq:super-inefficient-power-tower}, the $R$ needed to make this proposition hold is approximately the power tower $\exp^{(r+1)} 2eCr$.
This is the $R$ we will need to simulate an $r$-local memory algorithm by an $R$-local algorithm.
While this $R$ is a constant for any constant $r$, it would of course be nice to improve this dependence.

Finally, we can prove Proposition~\ref{prop:local-memory-to-local}.
\begin{proof}[Proof of Proposition~\ref{prop:local-memory-to-local}]
    We let $\cA'$ be the $R$-local simulation of $\cA$, for $R$ to be determined.
    We couple the runs of $\cA, \cA'$ to use the same $\varphi, \psi$.
    If $(G,\rho)$ is the factor graph of $\Phi$, then
    \[
        \Delta(\cA(\Phi), \cA'(\Phi))
        \le
        \f{1}{n}
        \sum_{v\in \Va_G}
        \ind{\text{$v$ is $(r,R,\psi)$-insulated in $G$}}.
    \]
    Recall that $(G,\rho,(\varphi,\psi)) \sim \DFG(n,m,k,(\Omega,\P_\omega)\times \unif([0,1]))$.
    The last indicator is an $R$-local function taking values in $[-1, 1]$.
    By Corollary~\ref{cor:simulation-dfg-concentration} with $t=\eta n/3$,
    \[
        \f{1}{n}
        \sum_{v\in \Va_G}
        \ind{\text{$v$ is $(r,R,\psi)$-insulated in $G$}}
        \le
        \eta/3 +
        \E
        \ind{\text{$v$ is $(r,R,\psi)$-insulated in $G$}}
    \]
    with probability $1-\exp(-\Omega(n^{1/3}))$.
    Let $(o,T,\rho,(\varphi,\psi))\sim \DGW(\alpha k, k-1, (\Omega,\P_\omega)\times \unif([0,1]))$.
    By Lemma~\ref{lem:simulation-dfg-to-dgw},
    \[
        \E
        \ind{\text{$v$ is $(r,R,\psi)$-insulated in $G$}}
        \le
        \f{c\log n}{n^{1/2}} +
        \E
        \ind{\text{$o$ is $(r,R,\psi)$-insulated in $T$}}.
    \]
    By Proposition~\ref{prop:dgw-root-insulated}, for sufficiently large $R$ depending on $\alpha, k, r, \eta$,
    \[
        \E
        \ind{\text{$o$ is $(r,R,\psi)$-insulated in $T$}}
        \le
        \eta/3.
    \]
    Putting this all together, with probability $1-\exp(-\Omega(n^{1/3}))$,
    \[
        \Delta(\cA(\Phi), \cA'(\Phi))
        \le
        2\eta/3 + \f{c\log n}{n^{1/2}}
        \le \eta
    \]
    for sufficiently large $n$.
\end{proof}

\subsection{Low Degree Polynomials Simulate Local Algorithms}
\label{subsec:simulation-local-to-ldp}

In this subsection, we prove Proposition~\ref{prop:local-to-ldp}, that any local algorithm can be simulated by a constant degree polynomial.
The proof closely resembles the proof of \cite[Theorem 1.4]{Wei20}.
The main idea is to construct a low degree polynomial by inclusion-exclusion that simulates the behavior of the local algorithm on any $r$-neighborhood that is a tree without too many edges.
We now define this simulation.

Consider $\Phi \in \Omega_k(n,m)$ with factor graph $(G, \rho)$.
Recall that $\Phi$ is encoded by indicators $\Phi_{i,j,s}$ ($i\in [m]$, $j\in [k]$, $s\in [2n]$) that $\Phi_{i,j}$ is the $s$th literal of $\cL$.
For each $s\in [2n]$, let $v(s) \in [n]$ be the index of the underlying variable of the $s$th literal of $\cL$.
Each triple $(i,j,s)$ is naturally associated with the edge $e = (v_{v(s)}, c_i)$ of the factor graph.
For a set $S\subseteq [m]\times [k]\times [2n]$, let $e(S)$ be the (multi-)set of edges associated in this manner to triples $(i,j,s)\in S$.
For $D\in \bN$ and $v\in \Va_G$, let $\cG_{v,r,D}$ be the collection of sets $S \subseteq [m]\times [k] \times [2n]$ such that
\begin{enumerate}[label=(\alph*), ref=\alph*]
    \item The bipartite graph $G(S) = (\Va_G, \Cl_G, e(S))$ is a tree in which every non-isolated vertex has a path to $v$ of length at most $r$.
    (This includes that $G(S)$ does not have multiple edges.)
    \item $|S|\le D$.
\end{enumerate}
Equivalently, $\cG_{v,r,D}$ is the collection of sets of $(i,j,s)$ corresponding to all possible tree shaped $r$-neighborhoods of $v$ in $G$ of size at most $D$.

\begin{definition}[Degree-$D$ simulation]
    Let $\cA_g$ be an $r$-local algorithm, with i.i.d. internal randomness from $(\Omega, \P_\omega)$.
    For $D\in \bN$, the degree-$D$ simulation of $\cA_g$ is the random polynomial that runs as follows on input $\Phi \in \Omega_k(n,m)$ with factor graph $(G,\rho)$.
    \begin{enumerate}[label=(\arabic*), ref=\arabic*]
        \item Sample $\varphi \sim (\Omega, \P_\omega)^{\otimes (V_G\cup E_G)}$.
        \item For each $v=v_i\in \Va_G$, set
        \begin{equation}
            \label{eq:ldp-def-f}
            f_i(\Phi, \varphi)
            =
            \sum_{S \in \cG_{v,r,D}}
            h(v, G(S), \rho, \varphi)
            \prod_{(i,j,s)\in S}
            \Phi_{(i,j,s)},
        \end{equation}
        where the coefficients $h(v, G(S), \rho, \varphi)$ are given recursively by
        \begin{equation}
            \label{eq:ldp-def-h}
            h(v, G(S), \rho, \varphi)
            =
            (\strictround^{-1} \circ g)(v, G(S), \rho\big|_{G(S)}, \varphi\big|_{G(S)})
            -
            \sum_{S' \subsetneq S}
            h(v, G(S'), \rho, \varphi).
        \end{equation}
    \end{enumerate}
\end{definition}

The internal randomness of $f$ is the map $\varphi$.
It is clear that this $f$ is a degree-$D$ polynomial.
We will analyze the performance of the degree-$D$ simulation by analogy to the following local function.
\begin{definition}[$D$-truncation]
    If $g: \Lambda \to \{\T,\F\}$ is an $r$-local function, the \emph{$D$-truncation} $g_{\le D} : \Lambda \to \{\T,\F,\Q\}$ is defined by
    \[
        g_{\le D}(v,G,\rho,\varphi)
        =
        \begin{cases}
            g(v,G,\rho,\varphi) & \text{$N_r(v,G)$ is a tree and $|N_r(v,G)| \le D$}, \\
            \Q & \text{otherwise}.
        \end{cases}
    \]
\end{definition}

By inclusion-exclusion, \eqref{eq:ldp-def-f} and \eqref{eq:ldp-def-h} immediately imply the following fact.
\begin{fact}
    \label{fac:deg-D-sim-good}
    For all $v=v_i\in \Va_G$ where $g_{\le D}(v,G,\rho,\varphi) \neq \Q$,
    \[
        (\strictround \circ f_i)(\Phi,\varphi)
        =
        g(v,G,\rho,\varphi)
        =
        g_{\le D}(v,G,\rho,\varphi).
    \]
\end{fact}
In other words, when $\cA_g$, $\cA_{g_{\le D}}$, and the degree-$D$ simulation $f$ of $\cA_g$ are run with the same $\varphi$, $\strictround \circ f$ correctly simulates any output of $\cA_g$ that $\cA_{g_{\le D}}$ correctly simulates.
Therefore, we can upper bound the fraction of variables where the simulation $f$ fails by bounding the fracton of variables where $\cA_{g_{\le D}}$ fails.
We achieve this by controlling the corresponding probability in $\DGW(d_1,d_2)$, and then translating this bound to the $k$-SAT factor graph by the machinery of Lemma~\ref{lem:simulation-dfg-to-dgw} and Corollary~\ref{cor:simulation-dfg-concentration}.

\begin{lemma}
    \label{lem:local-to-truncation}
    Suppose $\alpha k, k-1 \ge 2$ and $\eta > 0$.
    Let $\cA_g$ be an $r$-local algorithm with output in $\{\T,\F\}^n$.
    There exists $D > 0$ depending on $\alpha, k, r, \eta$ such that if $\cA_g$ and $\cA_{g_{\le D}}$ are run with the same $\varphi$, then
    \[
        \P\lt[
            \Delta(\cA_g(\Phi), \cA_{g_{\le D}}(\Phi)) \ge \eta
        \rt]
        \le \exp(-\Omega(n^{1/3})),
    \]
    where the probability is over the randomness of $\Phi \sim \Phi_k(n,m)$ and $\varphi$.
\end{lemma}
\begin{proof}
    By Corollary~\ref{cor:simulation-dfg-concentration} with $t=\eta n/3$,
    \begin{align*}
        \Delta(\cA_g(\Phi), \cA_{g_{\le D}}(\Phi))
        &=\f{1}{n}
        \sum_{v\in \Va_G}
        \ind{\text{$N_r(v,G)$ is not a tree or $|N_r(v,G)|>D$}} \\
        &\le
        \eta/3 +
        \E \ind{\text{$N_r(v,G)$ is not a tree or $|N_r(v,G)|>D$}}
    \end{align*}
    with probability $1-\exp(-\Omega(n^{1/3}))$, because the indicator is an $r$-local function taking values in $[-1, 1]$.
    Let $(o,T,\rho,\varphi) \sim \DGW(\alpha k, k-1)$.
    By Lemma~\ref{lem:simulation-dfg-to-dgw},
    \begin{align*}
        \E \ind{\text{$N_r(v,G)$ is not a tree or $|N_r(v,G)|>D$}}
        &\le
        \f{c\log n}{n^{1/2}} +
        \E \ind{|N_r(o,T)|>D},
    \end{align*}
    where we use that $N_r(o,T)$ is always a tree.
    By Lemma~\ref{lem:simulation-dgw-nbd-growth-rate}, we can pick $D$ large enough (depending on $\alpha,k,r,\eta$) that
    \[
        \E \ind{|N_r(o,T)|>D} \le \eta/3.
    \]
    Putting this all together, with probability $1-\exp(-\Omega(n^{1/3}))$,
    \[
        \Delta(\cA_g(\Phi), \cA_{g_{\le D}}(\Phi))
        \le
        2\eta/3 + \f{c\log n}{n^{1/2}}
        \le \eta
    \]
    for sufficiently large $n$.
\end{proof}

We get the second conclusion of Proposition~\ref{prop:local-to-ldp} from the following lemma.
\begin{lemma}
    \label{lem:deg-D-sim-2mt}
    If $\cA_g$ is an $r$-local algorithm and $f$ is its degree-$D$ simulation, then there exists $\gamma$ depending on $\alpha, k, r, D$ such that
    \[
        \E \norm{f(\Phi, \varphi)}_2^2
        \le
        \gamma n.
    \]
\end{lemma}
\begin{proof}
    We will upper bound each $\E_{\Phi, \varphi}[f_i(\Phi, \varphi)^2]$ by a constant depending only on $\alpha,k,r,\eta$.
    Fix $i\in [n]$.
    Let $v=v_i \in \Va_G$ and define the random variable $X = |N_r(v,G)|$.
    In the expansion \eqref{eq:ldp-def-f}, the monomial indexed by $S\in \cG_{v,r,D}$ is only nonzero if $e(S)$ is a subset of the edges of $N_r(v,G)$.
    So, the number of nonzero monomials is at most
    \[
        k^D \sum_{d=0}^D \binom{X}{d} \le k^D(X+1)^D.
    \]
    Moreover, by \eqref{eq:ldp-def-h}, each of the coefficients $h(v, G(S), \rho, \varphi)$ is upper bounded by a constant $a$ dependent on $\alpha,k,r,D$.
    Thus
    \[
        f_i(\Phi, \varphi)^2 \le a^2k^{2D}(X+1)^{2D}
    \]
    pointwise, and so
    \[
        \E \lt[f_i(\Phi, \varphi)^2\rt]
        \le
        a^2k^{2D}
        \E \lt[(X+1)^{2D}\rt].
    \]
    Lemma~\ref{lem:simulation-dfg-nbd-growth-rate} gives an exponential bound on the tail probability of $X$.
    Integration by tails gives the result.
\end{proof}

\begin{proof}[Proof of Proposition~\ref{prop:local-to-ldp}]
    Set $D$ such that Lemma~\ref{lem:local-to-truncation} holds, and let $f$ be the degree-$D$ simulation of $\cA$.
    We couple $f, \cA_g, \cA_{g_{\le D}}$ to all use the same $\varphi$.
    Fact~\ref{fac:deg-D-sim-good} and Lemma~\ref{lem:local-to-truncation} imply the first conclusion.
    Since $D$ depends on only $\alpha,k,r,\eta$, so does the $\gamma$ given by Lemma~\ref{lem:deg-D-sim-2mt}.
    This implies the second conclusion.
\end{proof}

\subsection{Deferred Proofs}
\label{subsec:simulation-deferred}

In this subsection, we give the deferred proof of Lemma~\ref{lem:simulation-max-nbd}.
We first prove a sharper version of Lemma~\ref{lem:simulation-dgw-nbd-growth-rate} for a \emph{specific} $r$, where the bound is improved by an $r$-iterated logarithmic factor.

\begin{lemma}
    \label{lem:simulation-dgw-better-growth-rate}
    Let $r\in \bN$, $d_1, d_2 \ge 2$ and $(o, T, \rho, \varphi) \sim \DGW(d_1, d_2)$.
    There exists $\tst$ (depending on $r, d_1, d_2$) such that for all $t\ge \tst$,
    \[
        \P \lt[
            |N_{2r}(o,T)|
            \le
            \f{2t}{\log^{(r)}t}(d_1d_2)^{r}
        \rt] \ge 1 - e^{-t}.
    \]
\end{lemma}
\begin{proof}
    For $0\le \ell \le 2r$, let $S_\ell$ denote the number of vertices in $\DGW(d_1, d_2)$ at depth $\ell$.
    The $S_\ell$ have the following distribution.
    First, $S_0 = 1$.
    For $\ell \ge 1$, $S_\ell$ is the sum of $S_{\ell-1}$ i.i.d. copies of $\Pois(d_1)$ if $\ell$ is odd, and $S_\ell = d_2 S_{\ell-1}$ if $\ell$ is even.

    For $i\in [r]$, define the event
    \[
        E_i = \lt\{
            S_{2i-1}
            \le
            \f{t}{\log^{(i)} t}
            d_1^i d_2^{i-1}
        \rt\}.
    \]
    This is equivalent to the event that $S_{2i} \le \f{t}{\log^{(i)} t} d_1^id_2^i$.
    For convenience, also define $E_0 = \{S_0 = 1\}$, which holds almost surely.
    On $\bigcap_{i=0}^r E_i$, we have
    \[
        |N_{2r}(o,T)|
        =
        \sum_{\ell=0}^{2r} S_\ell
        \le
        \f{t}{\log^{(r)} t}
        \cdot
        \f{d_1^rd_2^{r-1} + d_1^rd_2^r}{1-(d_1d_2)^{-1}}
        \le
        \f{2t}{\log^{(r)} t}
        (d_1d_2)^r.
    \]
    So, it remains to show that $\P \lt[\bigcap_{i=0}^r E_i\rt] \ge 1- e^{-t}$.

    Consider $i\in [r]$; we will upper bound $\P(E_i^c | E_{i-1})$.
    Let $N = \f{t}{\log^{(i-1)}t} d_1^{i-1}d_2^{i-1}$ (where $\log^{(0)} t = t$).
    Conditioned on $E_{i-1}$, we have $S_{2i-2}\le N$, so $S_{2i-1}$ is stochastically dominated by $\sum_{j=1}^N \xi_i$, where the $\xi_i$ are i.i.d. samples from $\Pois(d_1)$.
    By a standard Chernoff bound,
    \begin{align*}
        \P(E_i^c | E_{i-1})
        &\le
        \P\lt[
            \sum_{j=1}^N \xi_i
            \ge
            \f{\log^{(i-1)}t}{\log^{(i)}t}
            d_1 N
        \rt] \\
        &\le
        \lt[
            \inf_{s>0}
            \E \exp(s\xi_1)
            \cdot
            \exp \lt(
                -\f{\log^{(i-1)}t}{\log^{(i)}t}
                sd_1
            \rt)
        \rt]^N \\
        &=
        \lt[
            \inf_{s>0}
            \exp\lt(
                (e^s-1)d_1
                -
                \f{\log^{(i-1)}t}{\log^{(i)}t} sd_1
            \rt)
        \rt]^N \\
        &=
        \exp\lt(
            -Nd_1 \gamma\lt(
                \f{\log^{(i-1)}t}{\log^{(i)}t}
            \rt)
        \rt),
    \end{align*}
    where $\gamma(x) = x\log x - x + 1$.
    For large enough $t$,
    \[
        \gamma\lt(
            \f{\log^{(i-1)}t}{\log^{(i)}t}
        \rt)
        \ge
        \f34
        \f{\log^{(i-1)}t}{\log^{(i)}t}
        \log \f{\log^{(i-1)}t}{\log^{(i)}t}
        \ge
        \f23 \log^{(i-1)}t,
    \]
    while (as $d_1, d_2 \ge 2$ implies $d_1^i d_2^{i-1} \ge 2^{2i-1} \ge 2i$ for $i\ge 1$)
    \[
        Nd_1
        =
        \f{t}{\log^{(i-1)}t} d_1^id_2^{i-1}
        \ge
        2i \cdot \f{t}{\log^{(i-1)}t}.
    \]
    Thus, for large enough $t$, $\P(E_i^c | E_{i-1}) \le \exp(-\f43 it)$.
    So,
    \[
        \P \lt[
            \bigcap_{i=0}^r E_i
        \rt]
        \ge
        1 -
        \sum_{i=1}^r
        \P(E_i^c | E_{i-1})
        \ge
        1 - \f{\exp(-\f43 t)}{1-\exp(-\f43 t)}
        \ge
        1 - e^{-t}
    \]
    for sufficiently large $t$.
\end{proof}

\begin{proof}[Proof of Lemma~\ref{lem:simulation-max-nbd}]
    Set $\lambda > 0$ such that the conclusion of Lemma~\ref{lem:simulation-dgw-nbd-growth-rate} holds with probability $1 - \eta/2$.
    Denote this event $S$; on this event, $|N_{R}(o,T)| \le \lambda (d_1d_2)^{\lceil R/2\rceil}$ for all $R\in \bN$.

    For $R\in \bN$, let $t(R)$ be the smallest positive integer such that $\lambda (d_1d_2)^{\lceil R/2\rceil} e^{-t(R)} \le \eta / 2$; note that $t(R) = \Theta(R)$ for $d_1, d_2, r, \eta$ fixed.
    Set $\Rst$ such that $t(\Rst) \ge \tst$ for the $\tst$ in Lemma~\ref{lem:simulation-dgw-better-growth-rate}.
    Henceforth let $R\ge \Rst$ and $t = t(R) \ge \tst$.

    For $v\in \Va_T$, let $N^{\downarrow}_{2r}(v,T)$ denote the subset of $N_{2r}(v,T)$ in the descendant subtree of $v$.
    Note that the descendant subtree of $v$ has distribution $\DGW(d_1,d_2)$, so $|N^{\downarrow}_{2r}(v,T)| =_d |N_{2r}(o,T)|$.
    By Lemma~\ref{lem:simulation-dgw-better-growth-rate}, for each $v\in \Va_T$,
    \[
        \P \lt[
            |N^{\downarrow}_{2r}(v,T)|
            \le
            \f{2t}{\log^{(r)}t}(d_1d_2)^{r}
        \rt] \ge 1 - e^{-t}.
    \]
    By a union bound,
    \[
        \P \lt[
            S
            \text{~and~}
            \max_{v\in \Va_T \cap N_{R}(o,T)}
            |N^{\downarrow}_{2r}(v,T)|
            \le
            \f{2t}{\log^{(r)}t}(d_1d_2)^{r}
        \rt]
        \ge
        1 - \f{\eta}{2} - \lambda (d_1d_2)^{\lceil R/2\rceil} e^{-t}
        \ge
        1-\eta.
    \]
    Let $S'$ be the event in this probability.
    Note that for $v\in \Va_T$,
    \[
        N_{2r}(v,T)
        \subseteq
        N^{\downarrow}_{2r}(v,T)
        \cup
        N^{\downarrow}_{2r}(\gr(v),T)
        \cup
        \cdots
        \cup
        N^{\downarrow}_{2r}(\gr^r(v),T),
    \]
    where $\gr(v)$ denotes the grandparent of $v$.
    Thus, on the event $S'$, we have
    \[
        \max_{v\in \Va_T \cap N_{R}(o,T)}
        |N_{2r}(v,T)|
        \le
        \f{2t(r+1)}{\log^{(r)}t}(d_1d_2)^{r}.
    \]
    For $v\in \Cl_T \cap N_{R}(o,T)$, simply note that $N_{2(r-1)}(v,T) \subseteq N_{2r}(\pa(v),T)$, where $\pa(v)$ denotes the parent of $v$.
    It follows that on $S'$,
    \[
        \max_{v\in N_{R}(o,T)}
        |N_{2(r-1)}(v,T)|
        \le
        \f{2t(r+1)}{\log^{(r)}t}(d_1d_2)^{r}
        \le
        \f{CR}{\log^{(r)}R},
    \]
    using that $t = t(R)$ and $t(R) = \Theta(R)$.
    The result follows by renaming $r$ to $r+1$.
\end{proof}

%% file: tex/9-achievability.tex
\section{Proof of Achievability}
\label{sec:achievability}

Throughout this section, let $\eps > 0$, $\alpha = (1-\eps) 2^k \log k / k$, and $m = \lfloor \alpha n \rfloor$.
In this section we will prove Theorem~\ref{thm:achievability}, that local algorithms and low degree polynomials can solve random $k$-SAT at this clause density $\alpha$.

We will prove this theorem by simulating the first phase of \verb|Fix|, which we denote \verb|Fix1|, by these two computation classes.
Parts (\ref{itm:achievability-local},\ref{itm:achievability-ldp}) of Theorem~\ref{thm:achievability} follow immediately from guarantees on \verb|Fix1| in \cite{Coj10} and our simulation results, Propositions~\ref{prop:local-memory-to-local} and \ref{prop:local-to-ldp}.
To prove parts (\ref{itm:achievability-local-conc},\ref{itm:achievability-ldp-conc}), we use the fact that \verb|Fix1| is simulated by a local algorithm to argue concentration of the number of clauses satisfied, in order to prove the stronger bound on the failure probability.

This section is structured as follows.
In Subsection~\ref{subsec:achievability-fix} we define \verb|Fix1| and introduce its guarantees.
This immediately implies Theorem~\ref{thm:achievability}(\ref{itm:achievability-local},\ref{itm:achievability-ldp}).
In Subsection~\ref{subsec:achievability-concentration} we show concentration of the number of clauses satisfied and prove Theorem~\ref{thm:achievability}(\ref{itm:achievability-local-conc},\ref{itm:achievability-ldp-conc}).

\subsection{Review of Fix}
\label{subsec:achievability-fix}

At clause density $\alpha = (1-\eps) 2^k \log k / k$, \verb|Fix| produces a (exactly) satisfying assignment with high probability.
At a high level, \verb|Fix| runs in three phases.
In the first phase, it produces a almost-satisfying assignment.
In the second phase, it modifies this assignment in a small fraction of variables, at most $k^{-12}$ with high probability, to ``don't know."
This is done in such a way that the remaining problem of assigning truth values to the ``don't know" variables is equivalent to a very subcritical random $3$-SAT instance.
The third phase solves the remaining problem with a maxflow algorithm.

We will only show that local algorithms and low degree polynomials simulate the first phase \verb|Fix1|.
Because the the rest of \verb|Fix| changes at most a $k^{-12}$ fraction of variables with high probability, simulating \verb|Fix1| within normalized Hamming distance $\eta' > 0$ simulates \verb|Fix| within error $\eta = k^{-12} + \eta'$.
This is why Theorem~\ref{thm:achievability} requires $\eta > k^{-12}$.
Let us record the guarantees on \verb|Fix1| proved in \cite{Coj10}.
\begin{theorem}[Implicit in {\cite[Section 3]{Coj10}}]
    \label{thm:fix1-guarantee}
    Let \verb|Fix1| be defined in Algorithm~\ref{alg:fix1} below. Then,
    \[
        \P \lt[
            \verb|Fix1|(\Phi)~\text{$(k^{-12}, 0)$-satisfies $\Phi$}
        \rt]
        \ge 1-o(1).
    \]
    The probability is over $\Phi \sim \Phi_k(n,m)$ and the (independent) internal randomness of \verb|Fix1|.
\end{theorem}

We now define \verb|Fix1|.
This phase starts from the all-true assignment $x=\T^n$ and selects a set of indices $Z\subseteq [n]$ such that if $\{x_i : i\in Z\}$ are set false, most clauses are satisfied.
To do this, it scans through the clauses of the input formula $\Phi$.
When it encounters an all-negative clause that does not contain any variable from $Z$, it tries to find a true variable $x_i$ from this clause that when made false does not create more unsatisfied clauses.
It adds this $i$ to $Z$.
Formalizing this idea, we say $x_i$ (for $i \in [n]\setminus Z$) is \vocab{$Z$-safe} if, when we set $\{x_{i'} : i'\in [n]\setminus Z\}$ to true and $\{x_{i'} : i'\in Z\}$ to false, $x_i$ is not the sole true literal in any clause.

\begin{algorithm}[Fix, Phase 1; \cite{Coj10}]
    \label{alg:fix1}
    On input $\Phi \in \Omega_k(n,m)$, \verb|Fix1| runs as follows.
    \begin{enumerate}[label=(\arabic*), ref=\arabic*]
        \item Set $Z = \emptyset$.
        \item Relabel the clauses $\{\Phi_i : i\in [m]\}$ in a uniformly random order. Also, for each $i\in [m]$, relabel the literals $\{\Phi_{i,j} : j\in [k]\}$ in a uniformly random order.
        \item For $i\in [m]$ in increasing order:
        \begin{enumerate}[label=(\alph*), ref=\alph*]
            \item If $\Phi_i$ is all-negative and contains no variable from $\{x_i : i\in Z\}$:
            \begin{enumerate}[label=(\roman*), ref=\roman*]
                \item If there is $1\le j< \lceil k/2\rceil$ such that the underlying variable of $\Phi_{i,j}$ is $Z$-safe, pick the smallest such $j$ and add the underlying variable of $\Phi_{i,j}$ to $Z$.
                \item Otherwise, add the underlying variable of $\Phi_{i, \lceil k/2\rceil}$ to $Z$.
            \end{enumerate}
        \end{enumerate}
        \item Output $x\in \{\T,\F\}^n$ where $x_i = \F$ if $i\in Z$ and otherwise $x_i = \T$.
    \end{enumerate}
\end{algorithm}
The presentation of \verb|Fix1| in \cite{Coj10} does not rerandomize the clause and literal orders, but of course this makes no difference.
We add this rerandomization so that the algorithm is a local memory algorithm in the sense we define.
For technical reasons having to do with the analysis in \cite{Coj10}, \verb|Fix1| only considers flipping variables $\Phi_{i,j}$ where $j\le \lceil k/2\rceil$.

\begin{fact}
    \label{fac:fix-local-memory}
    \verb|Fix1| is a $3$-local memory algorithm.
\end{fact}
\begin{proof}
    We will construct a $3$-local subroutine $f_1$ and a function $f_2 : \bZ_{\ge 0} \to \{\T,\F\}$ such that $\cA_{f_1,f_2}$ simulates \verb|Fix1|.
    Let $(G,\rho)$ be the factor graph of $\Phi$, and let the i.i.d. randomness of $\varphi$ be sampled from $\unif([0,1])$.

    We will maintain the invariant that for each $v=v_i\in \Va_G$, $\mu(v)=1$ if $i\in Z$, and otherwise $\mu(v)=0$.

    The subroutine $f_1$ runs as follows on $(v,G,\rho,\varphi,\mu)$.
    If $v\in \Va_G$, do nothing.
    The remaining loop over $v\in \Cl_G$ runs over these vertices in a uniformly random order, as desired.
    If $v = c_i\in \Cl_G$, $f_1$ orders the edges $e$ incident to $c_i$ in increasing order of $\varphi(e)$.
    It runs the logic inside the for loop of \verb|Fix1|, with the corresponding literals $\{\Phi_{i,j} :j\in [k]\}$ relabeled in this order, and records the outcome on $\mu$.
    Note that the literals $\{\Phi_{i,j} :j\in [k]\}$ are relabeled in a uniformly random order, and that the logic inside the for loop is $3$-local.
\end{proof}

The proof of Theorem~\ref{thm:achievability}(\ref{itm:achievability-local},\ref{itm:achievability-ldp}) follows immediately from Fact~\ref{fac:fix-local-memory} and our simulation results.
\begin{proof}[Proof of Theorem~\ref{thm:achievability}(\ref{itm:achievability-local},\ref{itm:achievability-ldp})]
    Since $\eta > k^{-12}$, we can find $\eta' > 0$ such that $\eta = k^{-12} + 2\eta'$.
    Theorem~\ref{thm:fix1-guarantee}, Fact~\ref{fac:fix-local-memory} and Proposition~\ref{prop:local-memory-to-local} give $r>0$ and an $r$-local algorithm $\cA$ such that
    \begin{equation}
        \label{eq:achievability-local-alg}
        \P \lt[
            \text{$\cA(\Phi)$ $(k^{-12}+\eta', 0)$-satisfies $\Phi$}
        \rt]
        \ge
        1 - \delta(n)
    \end{equation}
    for $\delta(n) = o(1) + \exp(-\Omega(n^{1/3})) = o(1)$.
    Since $k^{-12}+\eta' < \eta$, this proves part (\ref{itm:achievability-local}).

    Proposition~\ref{prop:local-to-ldp} gives $D,\gamma$ and a random degree-$D$ polynomial that $(\delta(n), \gamma, \eta, 0)$-solves $\Phi_k(n,m)$, for $\delta(n)$ with a larger $\exp(-\Omega(n^{1/3}))$ term.
    Here we use that $\eta = k^{-12} + 2\eta'$.
    Finally, Lemma~\ref{lem:impossibility-assume-deterministic} gives a deterministic degree-$D$ polynomial that $(3\delta(n), 3\gamma, \eta, 0)$-solves $\Phi_k(n,m)$.
    This proves part (\ref{itm:achievability-ldp}).
\end{proof}

\subsection{Concentration of Clauses Satisfied}
\label{subsec:achievability-concentration}

For $x\in \{\T,\F\}^n$, $\Phi \in \Omega_k(n,m)$, and $\eta \in (0,1)$, define the objective
\begin{equation}
    \label{eq:def-sat}
    \Sat_\eta(x,\Phi)
    =
    \max_{y\in \{\T,\F\}^n : \Delta(x,y) \le \eta}
    \lt(
        \text{$\#$ clauses of $\Phi$ satisfied by $y$}
    \rt).
\end{equation}
To prove Theorem~\ref{thm:achievability}(\ref{itm:achievability-local-conc}), we will show that the objective attained by any local algorithm concentrates.
For an assignment $x\in \{\T,\F\}^n$ and a partial assignment $z\in \{\T,\F\}^B$, where $B\subseteq [n]$, it will be useful to define the replacement operator $\Rep(x,z)$ by
\begin{equation}
    \label{eq:def-rep}
    \Rep(x,z) = y \in \{\T,\F\}^n
    \qquad
    \text{where}
    \qquad
    y_i =
    \begin{cases}
        x_i & i\not\in B, \\
        z_i & i\in B.
    \end{cases}
\end{equation}

\begin{proposition}
    \label{prop:local-alg-conc}
    Let $\cA$ be an $r$-local algorithm with internal randomness $\varphi$.
    Let $\eta \in (0,1)$ and $Y = \Sat_\eta(\cA(\Phi, \varphi),\Phi)$.
    Then,
    \[
        \P \lt[|Y-\E Y| \ge \f{m}{\log n}\rt]
        \le
        \exp(-\tOmega(n^{1/3})).
    \]
\end{proposition}
\begin{proof}
    We will show $Y$ has bounded differences with high probability, which implies concentration by Lemma~\ref{lem:mcdiarmid-with-bad}.

    Throughout this section, we will write $\cA(\Phi, \varphi)$ for $\cA$ run with input $\Phi$ and internal randomness $\varphi$.
    Define $z=z(\Phi, \varphi) \in \{\T,\F\}^B$ as the partial assignment maximizing
    \[
        \Sat_0\lt(\Rep(\cA(\Phi, \varphi), z), \Phi\rt),
    \]
    where the maximization is over all $B\subseteq [n]$ with $|B|\le \eta n$.
    We break ties arbitrarily but deterministically.
    By definition of $Y$, $\Rep(\cA(\Phi, \varphi), z)$ satisfies $Y$ clauses of $\Phi$.

    Let $(G, \rho)$ be the factor graph of $\Phi$.
    Let the vertex sets of $G$ be $\Va = \{v_1,\ldots,v_n\}$, $\Cl = \{c_1,\ldots,c_m\}$, and $V = \Va \cup \Cl$, which are fixed across all realizations of $G$.

    All the above random variables are $(\Phi, \varphi)$-measurable.
    We can reformat $(\Phi, \varphi)$ into $km+n+m$ independent parts $\zeta = (\psi_1,\ldots,\psi_{km+n+m})$: for $1\le j\le km$, $\psi_j = (\Phi_{L(j)}, \varphi(e))$ where $e$ is the edge in $G$ corresponding to $\Phi_{L(j)}$, and for $km+1\le j\le km+m+n$, $\psi_j = \varphi(v)$ ranges over $v\in V$.
    We will henceforth write $Y(\zeta)$ to denote the $Y$ corresponding to this realization of $\zeta$, and similarly for other random variables, which are all $\zeta$-measurable.

    Let $S$ denote the set of $\zeta$ such that $\Phi(\zeta)$ is $(r+1)$-locally small (recall Definition~\ref{defn:r-locally-small}).
    By Fact~\ref{fac:r-locally-small}, $\P(S^c) \le \exp(-\Omega(n^{1/3}))$.
    Suppose $\zeta, \zeta'\in S$ differ in only one coordinate.
    We will upper bound $|Y(\zeta) - Y(\zeta')|$.
    For now, assume the differing coordinate is $\psi_j$ for $1\le j\le km$; thus the factor graphs $G(\zeta), G(\zeta')$ differ in one edge.
    Let $c\in \Cl$ be the common endpoint of this edge.

    Assume without loss of generality that $Y(\zeta) \ge Y(\zeta')$.
    Then,
    \begin{align*}
        Y(\zeta) - Y(\zeta')
        &=
        \Sat_0(\Rep(\cA(\zeta), z(\zeta)), \Phi(\zeta))
        -
        \Sat_0(\Rep(\cA(\zeta'), z(\zeta')), \Phi(\zeta')) \\
        &\le
        \Sat_0(\Rep(\cA(\zeta), z(\zeta)), \Phi(\zeta))
        -
        \Sat_0(\Rep(\cA(\zeta'), z(\zeta)), \Phi(\zeta'))
    \end{align*}
    The last inequality holds because $z(\zeta')$ maximizes the number of clauses of $\Phi(\zeta')$ satisfied by $\Rep(\cA(\zeta'), z(\zeta'))$.

    Note that $\cA(\zeta)$ and $\cA(\zeta')$ only differ in coordinates $i\in [n]$ where $v_i \in N_r(c, G(\zeta)) \cup N_r(c, G(\zeta'))$.
    Thus $\Rep(\cA(\zeta), z(\zeta))$ and $\Rep(\cA(\zeta'), z(\zeta))$ differ in only these coordinates.
    So, if
    \[
        \ind{\text{$\Rep(\cA(\zeta), z(\zeta)))$ satisfies clause $\Phi(\zeta)_i$}}
        \neq
        \ind{\text{$\Rep(\cA(\zeta'), z(\zeta)))$ satisfies clause $\Phi(\zeta')_i$}},
    \]
    then $c_i \in N_{r+1}(c, G(\zeta)) \cup N_{r+1}(c, G(\zeta'))$.
    Because $\zeta, \zeta'\in S$, this implies $|Y(\zeta) - Y(\zeta')| \le O(n^{1/3})$.

    We can analogously show the same bounded difference inequality when $\zeta, \zeta'\in S$ differ in coordinate $\psi_j$ for $km+1\le j\le km+m+n$, corresponding to a vertex of the factor graphs.
    Moreover, for all $\zeta, \zeta'$, clearly $|Y(\zeta) - Y(\zeta')| \le m$.
    By Lemma~\ref{lem:mcdiarmid-with-bad},
    \begin{align*}
        \P\lt[|Y - \E Y| \ge \f{m}{\log n} \rt]
        &\le
        2 \exp\lt(\f{m^2/\log^2 n}{8(km+m+n)O(n^{2/3})}\rt)
        +
        \exp(-\Omega(n^{1/3}))
        O\lt(\f{2(km+m+n)m}{n^{1/3}}\rt) \\
        &\le
        \exp(-\tOmega(n^{1/3})).
    \end{align*}
\end{proof}

\begin{proof}[Proof of Theorem~\ref{thm:achievability}(\ref{itm:achievability-local-conc})]
    Equation \eqref{eq:achievability-local-alg} gives an $r$-local algorithm $\cA$ such that $\cA(\Phi,\varphi)$ $(\eta, 0)$-satisfies $\Phi$ with probability $1-o(1)$.
    If $Y = \Sat_\eta(\cA(\Phi, \varphi),\Phi)$, this implies
    \[
        \E Y = (1-o(1))m.
    \]
    Proposition~\ref{prop:local-alg-conc} proves the result with $\nu(n) = \f{1}{\log n} + o(1) = o(1)$.
\end{proof}

Recall that the proof of Proposition~\ref{prop:local-to-ldp} simulates a local algorithm by its $D$-truncation, which can be implemented by a low degree polynomial.
We will prove Theorem~\ref{thm:achievability}(\ref{itm:achievability-ldp-conc}) by showing a concentration result analogous to Proposition~\ref{prop:local-alg-conc} for $D$-truncations of local algorithms.

To formulate this result, we first extend the definition \eqref{eq:def-sat} of $\Sat_\eta$ to allow $x\in \{\T,\F,\Q\}^n$.
Note that the $y$ in the maximum of \eqref{eq:def-sat} must differ from $x$ in all positions where the entry of $x$ is $\Q$.
We define $\Sat_\eta(x,\Phi) = 0$ if $x$ has more than $\eta n$ entries equal to $\Q$.
(In particular, $\Sat_0(x,\Phi)=0$ if $x$ has any entry equal to $\Q$.)
We similarly extend the definition \eqref{eq:def-rep} of $\Rep$ to allow $x\in \{\T,\F,\Q\}^n$.

\begin{proposition}
    \label{prop:local-alg-truncated-conc}
    Let $\cA_g$ be an $r$-local algorithm with internal randomness $\varphi$, where $g : \Lambda \to \{\T,\F\}$ is an $r$-local function, and let $g_{\le D}$ be the $D$-truncation of $g$.
    Let $\eta \in (0,1)$.
    Let $D$ be large enough that
    \begin{equation}
        \label{eq:truncated-conc-D-assumption}
        \P \lt[
            \Delta(
                \cA_g(\Phi, \varphi),
                \cA_{g_{\le D}}(\Phi,\varphi)
            )
            \ge \eta
        \rt]
        \le
        \exp(-\Omega(n^{1/3})).
    \end{equation}
    (Such $D$ exists by Lemma~\ref{lem:local-to-truncation}.)
    Let $Y = \Sat_\eta(\cA_{g_{\le D}}(\Phi, \varphi),\Phi)$.
    Then,
    \[
        \P \lt[|Y-\E Y| \ge \f{m}{\log n}\rt]
        \le
        \exp(-\tOmega(n^{1/5})).
    \]
\end{proposition}
\begin{proof}
    We will again show $Y$ has bounded differences with high probability and use Lemma~\ref{lem:mcdiarmid-with-bad}.
    Define $z=z(\Phi, \varphi) \in \{\T,\F\}^B$ as the partial assignment maximizing
    \[
        \Sat_0(\Rep(\cA_{g_{\le D}}(\Phi, \varphi),z), \Phi).
    \]
    Let $(G,\rho)$ be the factor graph of $\Phi$, with vertex sets $\Va = \{v_1,\ldots,v_n\}$, $\Cl = \{c_1,\ldots,c_m\}$, and $V = \Va \cup \Cl$.

    Define $\zeta$ as in the proof of Proposition~\ref{prop:local-alg-conc}.
    Let $S$ denote the set of $\zeta$ such that:
    \begin{enumerate}[label=(\roman*), ref=\roman*]
        \item \label{itm:truncation-conc-error}
        $\Delta(\cA_g(\zeta),\cA_{g_{\le D}}(\zeta)) \le \eta$ and
        \item \label{itm:truncation-conc-nbd}
        For all $v\in V_G$, $|N_{r+1}(v,G)| \le n^{1/5}$.
    \end{enumerate}
    By the assumption \eqref{eq:truncated-conc-D-assumption} and Lemma~\ref{lem:simulation-dfg-nbd-growth-rate}, $\P(S^c) \le \exp(-\Omega(n^{1/5}))$.
    Note that for $\zeta \in S$, (\ref{itm:truncation-conc-error}) implies that $\Rep(\cA_{g_{\le D}}(\zeta),z(\zeta))$ has no $\Q$ symbols.

    Suppose $\zeta, \zeta' \in S$ differ in only one coordinate.
    We will upper bound $|Y(\zeta) - Y(\zeta')|$.
    Assume the differing coordinate is $\psi_j$ for some $1\le j\le km$.
    (The case $km+1\le j\le km+m+n$ is analogous.)
    Then, the factor graphs $G(\zeta), G(\zeta')$ differ in one edge.
    Let $c \in \Cl$ be the common endpoint of this edge.

    Let $U = N_r(c, G(\zeta)) \cup N_r(c, G(\zeta'))$.
    Because $\zeta,\zeta' \in S$, (\ref{itm:truncation-conc-nbd}) implies $|U| \le 2n^{1/5}$.
    Note that $\cA_{g_{\le D}}(\zeta)$ and $\cA_{g_{\le D}}(\zeta')$ only differ in coordinates $i\in [n]$ where $v_i \in U$.

    Assume without loss of generality that $Y(\zeta) \ge Y(\zeta')$.
    Unlike in the proof of Proposition~\ref{prop:local-alg-conc}, the estimate
    \[
        \Sat_0(\Rep(\cA_{g_{\le D}}(\zeta'), z(\zeta')))
        \ge
        \Sat_0(\Rep(\cA_{g_{\le D}}(\zeta'), z(\zeta)))
    \]
    is not helpful because the right-hand side is $0$ when $\Rep(\cA_{g_{\le D}}(\zeta'), z(\zeta))$ has $\Q$ symbols.
    Instead we note that, because $\cA_{g_{\le D}}(\zeta)$ and $\cA_{g_{\le D}}(\zeta')$ differ in at most $|U|$ positions, there exists $z'$ differing from $z(\zeta)$ in at most $2|U|$ positions ($|U|$ entries in $z(\zeta)$ not in $z'$ and vice versa) such that $\Rep(\cA_{g_{\le D}}(\zeta'), z')$ has no $\Q$ symbols.
    We use the estimate
    \begin{align*}
        Y(\zeta) - Y(\zeta')
        &=
        \Sat_0(\Rep(\cA_{g_{\le D}}(\zeta), z(\zeta)), \Phi(\zeta))
        -
        \Sat_0(\Rep(\cA_{g_{\le D}}(\zeta'), z(\zeta')), \Phi(\zeta')) \\
        &\le
        \Sat_0(\Rep(\cA_{g_{\le D}}(\zeta), z(\zeta)), \Phi(\zeta))
        -
        \Sat_0(\Rep(\cA_{g_{\le D}}(\zeta'), z'), \Phi(\zeta')).
    \end{align*}
    Now, if
    \[
        \ind{\text{$\Rep(\cA_{g_{\le D}}(\zeta), z(\zeta)))$ satisfies clause $\Phi(\zeta)_i$}}
        \neq
        \ind{\text{$\Rep(\cA_{g_{\le D}}(\zeta'), z'))$ satisfies clause $\Phi(\zeta')_i$}},
    \]
    either $c_i \in N_{r+1} (c,G(\zeta)) \cup N_{r+1} (c, G(\zeta'))$ or $c_i$ is adjacent to one of the (at most) $2|U|$ variables where $z(\zeta)$ and $z'$ disagree.
    By definition of $S$, there are at most $2n^{1/5}$ clauses in the former case, and $2|U|n^{1/5}$ clauses in the latter case.
    Thus $|Y(\zeta) - Y(\zeta')| \le O(n^{2/5})$.

    For general $\zeta, \zeta'$, we have $|Y(\zeta) - Y(\zeta')| \le m$.
    By Lemma~\ref{lem:mcdiarmid-with-bad},
    \begin{align*}
        \P\lt[|Y - \E Y| \ge \f{m}{\log n} \rt]
        &\le
        2 \exp\lt(\f{m^2/\log^2 n}{8(km+m+n)O((n^{2/5})^2)}\rt)
        +
        \exp(-\Omega(n^{1/5}))
        O\lt(\f{2(km+m+n)m}{n^{2/5}}\rt) \\
        &\le
        \exp(-\tOmega(n^{1/5})).
    \end{align*}
\end{proof}

\begin{proof}[Proof of Theorem~\ref{thm:achievability}(\ref{itm:achievability-ldp-conc})]
    Set $\eta' > 0$ such that $\eta = k^{-12} + 2\eta'$.
    Let $\cA = \cA_g$ be the $r$-local algorithm achieving \eqref{eq:achievability-local-alg}.
    By Lemma~\ref{lem:local-to-truncation}, there exists $D$ dependent on $\eps, k, \eta$ such that
    \[
        \P \lt[
            \Delta(
                \cA_g(\Phi, \varphi),
                \cA_{g_{\le D}}(\Phi,\varphi)
            )
            \ge \eta'
        \rt]
        \le \exp(-\Omega(n^{1/3})).
    \]
    With \eqref{eq:achievability-local-alg}, this implies
    \[
        \P \lt[
            \text{$\cA_{g_{\le D}}(\Phi,\varphi)$ $(\eta, 0)$-satisfies $\Phi$}
        \rt]
        \ge
        1-o(1).
    \]
    Thus,
    \[
        \E \Sat_\eta(\cA_{g_{\le D}}(\Phi, \varphi),\Phi) = (1-o(1)) m.
    \]
    For $\nu(n) = o(1) + \f{1}{\log n} = o(1)$, Proposition~\ref{prop:local-alg-truncated-conc} implies that
    \[
        \P \lt[
            \Sat_\eta(\cA_{g_{\le D}}(\Phi, \varphi),\Phi) \ge (1-\nu(n)) m
        \rt]
        \ge 1-\exp(-\tOmega(n^{1/5})).
    \]
    Fact~\ref{fac:deg-D-sim-good} implies that the degree-$D$ simulation $f$ of $\cA_g$ satisfies
    \[
        \P \lt[
            \Sat_\eta((\strictround \circ f)(\Phi, \varphi),\Phi) \ge (1-\nu(n)) m
        \rt]
        \ge 1-\exp(-\tOmega(n^{1/5})).
    \]
    In other words,
    \[
        \P \lt[
            \text{$(\strictround \circ f)(\Phi, \varphi)$ $(\eta, \nu(n))$-satisfies $\Phi$}
        \rt]
        \ge 1-\exp(-\tOmega(n^{1/5})).
    \]
    Lemma~\ref{lem:deg-D-sim-2mt} gives $\gamma$ such that
    \[
        \E \norm{f(\Phi, \varphi)}_2^2 \le \gamma n.
    \]
    Thus, $f$ is a degree-$D$ polynomial that $(\exp(-\tOmega(n^{1/5})),\gamma,\eta,\nu(n))$-solves $\Phi_k(n,m)$.
    Finally, Lemma~\ref{lem:impossibility-assume-deterministic} gives a deterministic degree-$D$ polynomial that $(\exp(-\tOmega(n^{1/5})),3\gamma,\eta,\nu(n))$-solves $\Phi_k(n,m)$.
\end{proof}

%% file: tex/10-discussion.tex
\section{Discussion}
\label{sec:discussion}

In this paper we proved that degree $D = o(n/\log n)$ polynomials do not solve random $k$-SAT at clause density $(1+o_k(1)) \kappast 2^k \log k / k$ with success probability $1-\exp(-\Omega(D \log n))$.
We proved that local algorithms cannot solve random $k$-SAT even with success probability $\exp(-\tOmega(n^{1/3}))$ at this clause density, and that at clause density $(1-o_k(1)) 2^k \log k / k$ both computation classes succeed with very high probability.
We now discuss related OGP work, future directions, and some open problems that remain.

\paragraph{The constant factor gap.}
The main open problem is to close the constant factor gap remaining between the clause densities of the positive and negative results.
Because the negative free entropy chaining technique stalls at a clause density lower bounded by $1.716 \cdot 2^k \log k / k$ (see Appendix~\ref{appsec:impossibility-kappast-discussion}), further ideas will be necessary to close this gap.

Two innovations in multi-OGPs appeared recently that may be useful for this task.
\cite{GK21} constructs a large interpolation and uses Ramsey theory to show the existence of a constellation of solutions with pairwise overlaps all approximately equal to a prescribed value.
This gives a finer control on the overlap structure constructed from the algorithm outputs than our approach, which only uses that some algorithm output falls into each moat we construct.
\cite{HS21} uses a \emph{branching OGP}, where the multi-OGP's forbidden structure is an arbitrarily complex ultrametric tree of solutions.
The branching OGP allows their argument to navigate the rich replica symmetry breaking structure of spin glasses and may be useful here.

More speculatively, we expect the limiting clause density for efficient algorithms to coincide with the clustering threshold, even in lower order terms.
\cite[Equation 6]{KMRSZ07} gives the more precise expression $\f{2^k}{k}\lt(\log k + \log \log k + 1 + O(\f{\log \log k}{\log k})\rt)$ for the clustering threshold.
On the algorithmic side, we expect this clause density to be attained (in the large-radius limit) by suitable refinements of \verb|Fix|, where the radius of the neighborhood used to make each decision grows from $3$ to a large constant.

\paragraph{Local Markov chains.}
Proposition~\ref{prop:local-memory-to-local}, our simulation result, can be lightly modified to show that local algorithms (and therefore low degree polynomials, by Proposition~\ref{prop:local-to-ldp}) simulate the following class of local Markov chains run for $O(n)$ time.
Start at a uniformly random initialization $x\in \{\T,\F\}^n$.
At each step, choose a uniformly random vertex $v\in V_G$ of the factor graph and, based on its $r$-neighborhood and the restriction of $x$ to this neighborhood, choose (possibly randomly) to toggle the bits of $x$ in this neighborhood.
This model includes the Glauber dynamics, which corresponds to making a $2$-local decision when $v\in \Va_G$ and doing nothing when $v\in \Cl_G$.
It also includes a lazy version of walksat, where instead of maintaining a list of unsatisfied clauses we choose vertices randomly and do nothing on any $v\in \Va_G$ or any $v\in \Cl_G$ whose clause is already satisfied.
We note that at clause density $O_k(2^k / k^2)$, where walksat is known to succeed, it does succeed in $O(n)$ time \cite{CFFKV09}; the aforementioned lazy version of walksat incurs overhead from laziness, but nonetheless finds a $\nu$-satisfying assignment in $O(n)$ time for any $\nu$ independent of $n$.

This model can be implemented as a variant of an $r$-local memory algorithm where vertices are sampled uniformly and i.i.d. instead of by a random permutation.
The number of iterations may increase from $n$ to any constant multiple of $n$.
The simulation result is proved analogously to Proposition~\ref{prop:local-memory-to-local}, by choosing a simulation radius $R$ such that it is unlikely for $r$-hop dependence chains to escape an $R$-neighborhood.
Consequently, our hardness results also apply to this Markov chain run for $O(n)$ time.
Unfortunately it is much harder to reason about time scales longer than $O(n)$, which is the time scale needed to aggregate global information.
At long time scales the best result is still \cite{CHH17}, which shows walksat fails at clause density $O_k(2^k \log^2 k / k)$.
Showing hardness for local Markov chains at a tighter clause density is another open problem.

\paragraph{Limitations of OGP for low degree hardness.}
Current OGP techniques to show low degree hardness only rule out quite large success probabilities.
This limitation arises because these arguments use that low degree polynomials are stable, which occurs with small but nontrivial probability.
In contrast, OGP techniques to show hardness for (for example) local algorithms leverage these algorithms' concentration properties, which occur with high probability; this allows us to show these algorithms cannot succeed with even small probability.
To see this difference, compare Proposition~\ref{prop:impossibility-prob-bounds}(\ref{itm:impossibility-prob-bound-consecutive}) with Proposition~\ref{prop:impossibility-local-prob-bounds}(\ref{itm:impossibility-local-prob-bound-conc}).
It would be nice to lower the success probability that low degree hardness results rule out, perhaps by leveraging a property stronger than stability.

Concentration style OGPs also allow the construction of more complex forbidden structures such as the branching OGP of \cite{HS21}, which appears difficult to replicate by stability style OGPs.
Allowing the use of these structures is another potential benefit of leveraging a property stronger than stability.

It would also be interesting to prove a low degree hardness result that does not exclude the interval $(-1, 1)$ in the rounding scheme (and thus, does not reference the normalization parameter $\gamma$).
Such a hardness result would be based on the inherent stability of polynomial threshold functions, rather than the stability imposed by a variance condition in $\gamma$.
Note that a generalization of the Gotsman-Linial conjecture \cite{GL92} to non-binary product spaces, plugged in modularly in place of Proposition~\ref{prop:impossibility-total-influence}, would yield a version of Theorem~\ref{thm:impossibility} in this setting at $\delta = \exp(-C D \sqrt{n} \log n)$.
One could hope to devise a different OGP argument that improves this probability.

\paragraph{Other random CSPs.}
Random $k$-SAT is one example of a random constraint satisfaction problem.
\cite{AC08}, the seminal paper linking clustering to algorithmic hardness, predicted that this connection holds in substantial generality for random CSPs.
We believe that recent developments in multi-OGP methodology make it possible to show similar hardness results in other CSPs.
Showing a general hardness result of this type for random CSPs would be a significant advancement of the field.

\paragraph{When clustering does not imply hardness.}
Recent work on the symmetric Ising perceptron \cite{ALS21, PX21} showed that clustering (in the sense that is linked to hardness in random CSPs) does not always imply hardness.
In particular, at any positive constraint density in the symmetric Ising perceptron, all but an exponentially small fraction of solutions are isolated, forming clusters \emph{of size one}, even though efficient algorithms that find a solution exist at some positive constraint densities \cite{BS19}.
Forthcoming work \cite{GK21b} shows that at a constraint density only slightly above where efficient algorithms exist, a multi-OGP rules out stable algorithms.

Thus, while the rigorous connection between multi-OGP and the failure of stable algorithms still holds, the heuristic that clustering implies hardness breaks down.
It would be interesting to clarify this heuristic and identify a refined notion of clustering that does match the limits of algorithms for this problem.

\paragraph{The refutation problem.}
Closely related to the problem of finding a satisfying assignment is the problem of certifying that there is no satisfying assignment.
A well-studied problem is to identify the \emph{refutation threshold}, the clause density above the satisfiability threshold where it is possible to certify the lack of a satisfying assignment with high probability.

It is known \cite{Sch08} that the SOS framework cannot efficiently refute satisfiability of random $k$-SAT with $m = O(n^{k/2 - \eps})$ clauses.
For the DPLL-based Resolution framework, the same fact is known for $k=3$ \cite{BW99}.
On the positive side, \cite{AOW15} showed that refutation is possible with $m=\tOmega(n^{k/2})$ clauses; see \cite[Table 1]{AOW15} for a history.
Thus there is strong evidence that the refutation threshold is at the scale $m=\tTheta(n^{k/2})$.
An open problem is to furnish rigorous evidence for this threshold for more general models of computation.
It would be interesting to identify a signature for refutation hardness at the level of the problem's energy landscape.
It would also be interesting to show hardness for other certification tasks, such as certifying an upper bound on the number of satisfying assignments, see e.g. \cite{HMX21}.

\paragraph{Planted problems.}
Another variant of OGP has been used to study the computational hardness of estimation problems and problems with planted structure \cite{GZ17, GZ19, GJS19, BWZ20}, see also \cite{CM19}.
This notion of OGP tracks the overlap between a single solution and the planted truth, instead of between two or more solutions.
OGP occurs if the best loss attained by a solution at some overlap with the planted truth, as a function of the overlap, is nonmonotone with one minimum at high overlap and another at low overlap.
If this OGP occurs, gradient descent or any local Markov chain with worst-case initialization will be unable to efficiently find the planted truth.
An open problem is to extend these hardness results for planted problems from local Markov chains to arbitrary stable algorithms.

%% file: tex/a-kappa.tex
\section{On Improving the Constant $\kappast$}
\label{appsec:impossibility-kappast-discussion}

In this section, we discuss how the constant $\kappast$ in Theorem~\ref{thm:impossibility} can be improved.
We define a constant $\kappastst$ as the solution to a maximin problem.
We will show that $\kappastst \le \kappast$ and sketch how our proof of Theorem~\ref{thm:impossibility} can be lightly modified to improve the constant $\kappast$ to $\kappastst$.
We heuristically argue that $\kappastst < \kappast$, so that this modification is an improvement.
We also prove that $\kappastst$ is bounded below by a constant larger than $1$, approximately $1.716$.
Further ideas will be needed to prove Theorem~\ref{thm:impossibility} for any $\kappa$ smaller than $\kappastst$.
Because $\kappastst$ remains bounded away from $1$, and we believe $1$ is the optimal constant in Theorem~\ref{thm:impossibility}, we did not attempt to rigorously evaluate $\kappastst$ or optimize $\kappast$.

\subsection{A Maximin Problem}

Let $(\Xi, P_\xi)$ be an arbitrary probability space and let $\cQ$ be the space of functions $q : \Xi \to [0,1]$.
These are abstractions of quantities in the proof of Proposition~\ref{prop:impossibility-ogp-dart-game}: $\xi \sim (\Xi, P_\xi)$ is an abstraction of the random variables $\yp{\le \ell-1}_{I}$ where $I\sim \unif([n])$, and $q(\xi)$ is an abstraction of $\phi_\ell(\T | \yp{\le \ell-1}_{I})$.
We equip $\cQ$ with the metric $d(q, q') = \E_\xi |q(\xi) - q'(\xi)|$.

For $q \in \cQ$, let $D(q)$ be the law of $u$ sampled by the following experiment.
First, sample $\xi \sim (\Xi, P_\xi)$.
Then, set $u = -\log q(\xi)$ with probability $q(\xi)$, and otherwise set $u = -\log(1 - q(\xi))$.
Clearly $\E_{u\sim D(q)} u = \E_{\xi} H(q(\xi))$.
Define
\[
    F(q) =
    \f{1}{\log k} \cdot
    \f{\log 2 + k \E_{\xi} H(q(\xi))}{
        \P_{(u_1,\ldots,u_k)\sim D(q)^{\otimes k}}
        \lt[\sum_{i=1}^k u_i \ge \log k + \log \log k\rt]
    }\,.
\]
Let $\cP$ be the set of functions $p : \Xi \times [0,1] \to [0,1]$, such that $p(\xi, 0) \in \{0,1\}$ and $p(\xi, 1) = \f12$ for all $\xi \in \Xi$, and $p(\cdot, s)$ (which, for fixed $s\in [0,1]$, is an element of $\cQ$) is continuous in $s$ with respect to the topology of $\cQ$.
Consider the maximin problem
\begin{equation}
    \label{eq:def-kappastst}
    \kappastst
    =
    \limsup_{k\to\infty}
    \max_{p\in \cP}
    \min_{s\in [0,1]}
    F(p(\cdot, s)).
\end{equation}
This has the following geometric interpretation: $\kappastst$ is the smallest constant such that the sub-level set $\{q\in \cQ : F(q) \le \kappastst\}$ topologically disconnects the functions $q \equiv 0$ and $q \equiv \f12$ in $\cQ$.
(Note that $\cQ$ is symmetric under replacing $q(\xi)$ with $1 - q(\xi)$ for any subset of the $\xi\in \Xi$, and $F(q) = F(q')$ for any $q,q'$ related by such a symmetry.
Thus, equivalently $\kappastst$ is the smallest constant such that this sub-level set disconnects the function $q\equiv \f12$ from any $q\in \cQ$ with $q(\xi) \in \{0,1\}$ for all $\xi\in \Xi$.)

First, we show that $\kappast$ is an upper bound on the solution to this maximin problem.

\begin{proposition}
    We have that $\kappast \ge \kappastst$.
\end{proposition}
\begin{proof}
    Fix some $p\in \cP$.
    By continuity of $p(\cdot, s)$ in $s$, we can set $s\in [0,1]$ such that $\E_{\xi} H(p(\xi, s)) = \betast \f{\log k}{k}$.
    As in the proof of Proposition~\ref{prop:impossibility-ogp-dart-game}, we apply a Chernoff bound on the random variables $\f{\min(u_i, \log k)}{\log k}$ to show that, for any $\beta > 1$ and $q : \Omega \to [0,1]$ with $\E_\xi H(q(\xi)) = \beta \f{\log k}{k}$, we have
    \begin{equation}
        \label{eq:impossibility-key-chernoff-bound}
        \P_{(u_1,\ldots,u_k)\sim D(q)^{\otimes k}}
        \lt[\sum_{i=1}^k u_i < \log k + \log \log k\rt]
        \le
        \beta e^{-(\beta-1)} + o_k(1).
    \end{equation}
    In particular, for the $s$ we chose,
    \[
        F(p(\cdot, s))
        \le
        \f{\f{2}{\log k} + \betast}{1 - \betast e^{-(\betast - 1)} - o_k(1)}
        \to \iota(\betast) = \kappast.
    \]
\end{proof}

Next, we sketch how the proof of Theorem~\ref{thm:impossibility} can be improved to replace $\kappast$ with $\kappastst$.
The proof of Theorem~\ref{thm:impossibility-local-algs} can be modified similarly.

\begin{proposition}
    Theorem~\ref{thm:impossibility} holds for all $\kappa > \kappastst$.
\end{proposition}
\begin{proof}[Proof Sketch]
    Identically to the original proof of Theorem~\ref{thm:impossibility}, we define the interpolation path $\phip{0}, \ldots, \phip{T}$ and set $\xp{t} = \cA(\phip{t})$ for $0\le t\le T$.
    We define $\Svalid$ as before.
    $\Sconsec$ and $\Sindep$ are analogous to before: $\Sconsec$ is the event that consecutive $\xp{t}$ are close in Hamming distance, and $\Sindep$ is the event that if $0\le t_0\le t_1 \le \cdots \le t_k \le T$ and $t_k \ge t_{k-1} + km$, then any $\nu$-satisfying assignment to $\phip{t_k}$ has large conditional overlap entropy relative to $\xp{t_0},\ldots,\xp{t_{k-1}}$.
    We change the parameters quantifying ``close" and ``large conditional overlap" slightly so that the below proof succeeds; we omit the details.
    Lower bounds on $\P(\Svalid \cap \Sconsec)$ and $\P(\Sindep)$ can be proved analogously to Proposition~\ref{prop:impossibility-prob-bounds}(\ref{itm:impossibility-prob-bound-consecutive},\ref{itm:impossibility-prob-bound-independent}).

    The interesting change will be in the definition of $\Sogp$.
    For $1\le \ell \le k$, the conditional overlap profile $\pi(\yp{\ell} | \yp{0}, \ldots, \yp{\ell-1})$ determines the conditional probabilities\footnote{this is a rewriting of the argument in Subsection~\ref{subsec:impossibility-energy-lb}. The {\XOR}s arise because we no longer assume $\yp{0} = \T^n$.}
    \[
        \phi_\ell (b | \xi)
        =
        \P_{i\sim \unif([n])}
        \lt[
            \yp{\ell}_i \oplus \yp{0}_i = b |
            (\yp{1}_i \oplus \yp{0}_i, \ldots, \yp{\ell-1}_i \oplus \yp{0}_i) = \xi
        \rt],
    \]
    where $\oplus$ denotes \XOR.
    Let $(\Xi, P_\xi)$ be the sample space of $(\yp{1}_i \oplus \yp{0}_i, \ldots, \yp{\ell-1}_i \oplus \yp{0}_i)$, and let $q(\xi) = \phi_\ell (\T | \xi)$.
    Let $\eps>0$ satisfy $\kappa - \eps > \kappastst$.
    $\Sogp$ is now the event that there does not exist $0\le t_0 \le t_1 \le \cdots \le t_k \le T$ and assignments $\yp{0},\ldots,\yp{k} \in \{\T,\F\}^n$ such that
    \begin{enumerate}[label=(OGP-\Alph*), ref=OGP-\Alph*, leftmargin=50pt]
        \item \label{itm:def-sogp-new-satisfy} For all $0\le \ell\le k$, $\yp{\ell}$ $\nu$-satisfies $\phip{t_\ell}$;
        \item \label{itm:def-sogp-new-entropy} For all $1\le \ell\le k$, the conditional overlap profile $H\lt(\pi(\yp{\ell} | \yp{0}, \ldots,\yp{\ell-1}) \rt)$ satisfies that $F(q)\le \kappa - \eps$ for the $q$ defined above.
    \end{enumerate}
    The key point is that our proof that $\P(\Sogp^c) \le \exp(-\Omega(n))$ requires precisely these properties.
    Using the argument in Section~\ref{sec:impossibility-multi-ogp}, we readily prove $\P(\Sogp^c) \le \exp(-\Omega(n))$.

    By a union bound, this gives a positive lower bound on $\P(\Svalid \cap \Sconsec \cap \Sindep \cap \Sogp)$, so $\Svalid \cap \Sconsec \cap \Sindep \cap \Sogp \neq \emptyset$.
    We will show (analogously to Proposition~\ref{prop:impossibility-events-relation}) that $\Svalid \cap \Sconsec \cap \Sindep \cap \Sogp = \emptyset$, yielding a contradiction.
    When $\Svalid, \Sconsec, \Sindep$ simultaneously hold, we will construct an example of the structure forbidden by $\Sogp$.

    We will set $\yp{\ell} = \xp{t_\ell}$ for all $0\le \ell \le k$, for a sequence $0\le t_0 \le t_1 \le \cdots \le t_k \le T$ we now construct.
    We set $t_0 = 0$.
    For $1\le \ell \le k$ we set $t_\ell$ to be the smallest $t>t_{\ell-1}$ such that (\ref{itm:def-sogp-new-entropy}) holds for $\yp{\ell} = \xp{t}$.
    We now sketch why $t_\ell$ exists and satisfies $t_\ell \le t_{\ell-1} + km$.
    Note that this ensures all the $t_\ell$ are well defined because $T = k^2m$.

    By $\Sconsec$, $\xp{t}$ evolves by small steps.
    Thus, for fixed $\yp{0}, \ldots, \yp{\ell-1}$ and varying $\yp{\ell} = \xp{t}$ (varying as we increment $t$), the $q$ defined above moves by small steps in $\cQ$.
    Let $q_t$ denote this $q$ at time $t$.
    By Fact~\ref{fac:overlap-profile-properties}(\ref{itm:overlap-profile-property-duplication}) $q_{t_{\ell-1}}(\xi) \in \{0,1\}$ for all $\xi \in \Xi$.
    $\Sindep$ ensures that $q_{t_{\ell-1}+km}$ is far from $q_{t_{\ell-1}}$ in $\cQ$.
    The evolution of $q_t$ from $t=t_{\ell-1}$ to $t=t_{\ell-1}+km$ can be modeled essentially by a continuous path, and the definition of the maximin $\kappastst$ implies that for some $t$ in this range, $F(q_t) \approx \kappastst \le \kappa-\eps$.
    (Although $q_t$ does not necessarily evolve to the all-$\f12$ function, $\Sindep$ implies that it ends far from where it started, and we can show that over this evolution we already encounter $q_t$ such that $F(q_t)$ is near the maximin value.)
    This shows the existence of $t_\ell$ with $t_\ell \le t_{\ell-1} + km$.

    Since, by $\Svalid$, each $\yp{\ell}$ $\nu$-satisfies $\phip{t_\ell}$, we have constructed an example of the structure forbidden by $\Sogp$.
    This gives the desired contradiction.
\end{proof}

\subsection{Suboptimality of $\kappast$}

We believe that $\kappast > \kappastst$ due to the following heuristic argument.
The Chernoff bound \eqref{eq:impossibility-key-chernoff-bound} is tighest when most of the mass of the random variables $\f{\min(u_i, \log k)}{\log k}$ is near $0$ or $1$.
When this occurs, most of the the mass of $u_i$ is near $0$ or $\log k$.
Then, the event that $\sum_{i=1}^k u_i \ge \log k + \log \log k$ is the event that one or two of the $u_i$ attains a value near $\log k$.
This is a tail probability in a non-asymptotic regime -- approximately, the probability that a Poisson random variable is larger than $1$ or $2$ -- so the Chernoff bound will not get the correct probability.

\subsection{Proof that $\kappastst > 1$}

In this subsection, we will show that $\kappastst$ is bounded below by a constant larger than $1$, approximately $1.716$.
Thus our methods cannot improve the constant $\kappast$ in Theorems~\ref{thm:impossibility} and \ref{thm:impossibility-local-algs} to $1$.

We will first show a weaker lower bound on $\kappastst$.
Define $\psi_1: (0,+\infty) \to \bR$ by
\[
    \psi_1(\lambda) = \f{\lambda/2}{1 - (1 + \lambda)e^{-\lambda}},
\]
and let $\psist_1 = \min_{\lambda > 0} \psi_1(\lambda) \approx 1.675$.
\begin{proposition}
    \label{prop:impossibility-kappastst-lb}
    We have $\kappastst \ge \psist_1$.
\end{proposition}
\begin{proof}
    We will prove this proposition by constructing a suitable function family $p \in \cP$.

    Let $\Xi = [0,1]$ equipped with the uniform measure.
    Let $p: \Xi \times [0,1] \to [0,1]$ be defined by
    \[
        p(\xi, s) =
        \begin{cases}
            \min(s, \f12) & \xi \le s, \\
            0 & \xi \ge s.
        \end{cases}
    \]
    Thus, for fixed $s\in [0,1]$, $p(\xi, s) = \min(s, \f12)$ with probability $s$, and otherwise $p(\xi, s) = 0$.
    We will show that for this $p$,
    \[
        \limsup_{k\to\infty}
        \min_{s\in [0,1]}
        F(p(\cdot, s))
        \ge
        \psist_1,
    \]
    from which the proposition follows.

    Note that if $s = \omega_k(k^{-1/2})$, then $\E_\xi H(p(\xi, s)) = \omega_k(\log k / k)$, and so $F(p(\cdot, s)) = \omega_k(1)$.
    Therefore it suffices to consider $s = O_k(k^{-1/2})$.
    Then,
    \[
        \E_\xi H(p(\xi, s)) =
        (1 + o_k(1)) s^2 \log \f{1}{s}.
    \]
    We now analyze the behavior of the denominator of $F(p(\cdot, s))$.
    Note that a sample $u \sim D(p(\cdot, s))$ equals $\log \f{1}{s}$ with probability $s^2$, $\log \f{1}{1-s} \le \f{s}{1-s}$ with probability $s(1-s)$, and $0$ with probability $1-s$.
    For $i=1,\ldots,k$, define
    \[
        v_i = \log \f{1}{s} \ind{u_i = \log \f{1}{s}},
        \qquad
        \text{and}
        \qquad
        w_i = \f{s}{1-s} \ind{u_i = \log \f{1}{1-s}}.
    \]
    So, $u_i \le v_i + w_i$.
    For $u_1,\ldots,u_k \sim D(p(\cdot, s))^{\otimes k}$, we have
    \[
        \P \lt[\sum_{i=1}^k u_i \ge \log k + \log \log k\rt]
        \le
        \P \lt[\sum_{i=1}^k v_i \ge \log k\rt] +
        \P \lt[\sum_{i=1}^k w_i \ge \log \log k\rt].
    \]
    Let $1 + \delta = \f{\log \log k}{k \E w_1} = \f{\log \log k}{ks^2}$.
    Because $s = O_k(k^{-1/2})$, we have $1+\delta = \omega_k(1)$, and so $\f{\delta^2}{2 + \delta} \ge \f12 (1+\delta)$ for sufficiently large $k$.
    By a Chernoff bound,
    \begin{align*}
        \P \lt[\sum_{i=1}^k w_i \ge \log \log k\rt]
        &\le
        \P \lt[\sum_{i=1}^k \f{1-s}{s} w_i \ge \f{1-s}{s} \log \log k\rt]
        \le \exp\lt(-\f{\delta^2}{2 + \delta} \cdot ks(1-s) \rt) \\
        &\le \exp\lt(- \f12 (1+\delta) ks(1-s)\rt)
        \le \exp\lt(- \f{(1-s)\log \log k}{2s}\rt)
        \le \exp\lt(- \Omega_k(k^{1/2})\rt).
    \end{align*}
    To analyze the other probability, we consider cases $s > \f{1}{k}$ and $s\le \f{1}{k}$.
    We first consider $s > \f{1}{k}$.
    In order to have $\sum_{i=1}^k v_i \ge \log k$, at least two $v_i$ must be nonzero.
    This occurs with probability
    \[
        1 - (1 - s^2)^k - s^2k (1 - s^2)^{k-1}
        \le 1 - (1 + s^2k) (1-s^2)^k.
    \]
    Thus,
    \[
        F(p(\cdot, s))
        \ge
        \f{1}{\log k} \cdot \f{\log 2 + (1 + o_k(1)) s^2k \log \f{1}{s}}{1 - (1 + s^2k) (1-s^2)^k + \exp(-\Omega_k(k^{1/2}))}.
    \]
    If $s^2k = o_k(1)$, then $1 - (1 + s^2k) (1-s^2)^k = O_k(s^4k^2)$, and the right-hand side is $\omega_k(1)$.
    So, this bound is minimized at $s = \lambda k^{-1/2}$ for constant $\lambda$, in which case
    \[
        \f{1}{\log k} \cdot \f{\log 2 + (1 + o_k(1)) s^2k \log \f{1}{s}}{1 - (1 + s^2k) (1-s^2)^k + \exp(-\Omega_k(k^{1/2}))}
        \to
        \f{\lambda/2}{1 - (1 + \lambda)\exp(-\lambda)} = \psi_1(\lambda) \ge \psist_1.
    \]

    We now consider $s \le \f{1}{k}$.
    In order to have $\sum_{i=1}^k v_i \ge \log k$, at least one $v_i$ must be nonzero.
    This occurs with probability
    \[
        1 - (1 - s^2)^k
        = (1 + o_k(1)) s^2k,
    \]
    and so
    \[
        F(p(\cdot, s))
        \ge
        \f{1}{\log k} \cdot \f{\log 2 + (1 + o_k(1)) s^2k \log \f{1}{s}}{(1+o_k(1)) s^2k + \exp(-\Omega_k(k^{1/2}))}.
    \]
    The right-hand side is $\omega_k(1)$ because $s \le \f{1}{k}$.
\end{proof}

For any nonnegative integer $N$, we may further define
\[
    \psi_N(\lambda) =
    \f{\lambda/(N+1)}{1-\lt(\sum_{k\le N} \f{\lambda^k}{k!}\rt)\exp(-\lambda)}
\]
and $\psist_N = \inf_{\lambda > 0} \psi_N (\lambda)$.
Over positive integers $N$, the largest $\psist_N$ is $\psist_2 \approx 1.716$.
The following corollary gives the lower bound on $\kappastst$ alluded to above.
\begin{corollary}
    \label{cor:impossibility-kappastst-lb}
    We have that $\kappastst \ge \psist_2$.
\end{corollary}
\begin{proof}
    We will construct a suitable function family $p$.
    For any nonnegative integer $N$, we can define
    \begin{equation}
        \label{eq:def-p-n}
        p_N(\xi, s) =
        \begin{cases}
            \min(s, \f12) & \xi \le s^N, \\
            0 & \xi \ge s.
        \end{cases}
    \end{equation}
    By a similar analysis to Proposition~\ref{prop:impossibility-kappastst-lb}, we can show for this $p$ that
    \[
        \limsup_{k\to\infty}
        \min_{s\in [0,1]}
        F(p_N(\cdot, s))
        \ge
        \psist_N.
    \]
    Taking $N=2$ yields the result.
\end{proof}

Due to Corollary~\ref{cor:impossibility-kappastst-lb}, a proof of Theorem~\ref{thm:impossibility} improving the constant $\kappast$ below $\psist_2$ will require new conceptual insights.
Finally, we conjecture that Corollary~\ref{cor:impossibility-kappastst-lb} is in fact sharp.
\begin{conjecture}
    \label{conj:kappastst-value}
    We have that $\kappastst = \psist_2$.
    In particular, Theorem~\ref{thm:impossibility} holds for all $\kappa > \psist_2$.
\end{conjecture}

The following evidence supports this conjecture.
In the maximin problem \eqref{eq:def-kappastst}, if we restrict the maximum over $p$ to functions such that for every $s$, $p(\xi, s)$ attains at most one nonzero value, then we can show by explicit computation that the maximin problem has value $\psist_2$.
The idea of this proof is that for each such $p$, at the $s$ minimizing $F(p(\cdot, s))$, $p(\cdot, s)$ equals (up to isomorphism of the probability space $(\Xi, P_\xi)$) $p_N(\cdot, s')$ for some $s'$ and some (possibly fractional) $N$.
We can show that fractional $N$ do not maximize $\min_{s\in [0,1]} F(p_N(\cdot, s))$.
Thus the candidate maxima are $p_N$ for integer $N$, and of these $p_2$ is maximal, attaining value $\psist_2$.
We believe that the maximum of \eqref{eq:def-kappastst} over $p\in \cP$ is attained by $p$ with this property.

%% file: ksat.bbl
\newcommand{\etalchar}[1]{$^{#1}$}
\begin{thebibliography}{KMRT{\etalchar{+}}07}

\bibitem[ABM04]{ABM04}
Dimitris Achlioptas, Paul Beame, and Michael Molloy.
\newblock Exponential bounds for {DPLL} below the satisfiability threshold.
\newblock In {\em Proceedings of 15th SODA}, pages 139--140, 2004.

\bibitem[Ach09]{Ach09}
Dimitris Achlioptas.
\newblock {\em Random satisfiability}, volume 185, pages 245--270.
\newblock IOS Press, 2009.

\bibitem[ACO08]{AC08}
Dimitris Achlioptas and Amin Coja-Oghlan.
\newblock Algorithmic barriers from phase transitions.
\newblock In {\em Proceedings of 49th FOCS}, pages 793--802, 2008.

\bibitem[Ajt83]{Ajt83}
Mikl{\'o}s Ajtai.
\newblock $\sigma_1^1$-formulae on finite structures.
\newblock {\em Annals of Pure and Applied Logic}, 24(1):1--48, 1983.

\bibitem[ALS21]{ALS21}
Emmanuel Abbe, Shuangping Li, and Allan Sly.
\newblock Proof of the contiguity conjecture and lognormal limit for the
  symmetric perceptron.
\newblock {\em arXiv preprint 2102.13069}, 2021.

\bibitem[AMS20]{AMS20}
Ahmed~El Alaoui, Andrea Montanari, and Mark Sellke.
\newblock Optimization of mean-field spin glasses.
\newblock {\em arXiv preprint arXiv:2001.00904}, 2020.

\bibitem[AOW15]{AOW15}
Sarah~R. Allen, Ryan O’Donnell, and David Witmer.
\newblock How to refute a random csp.
\newblock In {\em Proceedings of 56th FOCS}, pages 689--708, 2015.

\bibitem[AS00]{AS00}
Dimitris Achlioptas and Gregory~B. Sorkin.
\newblock Optimal myopic algorithms for random $3$-{SAT}.
\newblock In {\em Proceedings of 41st FOCS}, pages 590--600, 2000.

\bibitem[BAWZ20]{BWZ20}
G{\'e}rard Ben~Arous, Alexander~S. Wein, and Ilias Zadik.
\newblock Free energy wells and overlap gap property in sparse {PCA}.
\newblock In {\em Proceedings of 33rd COLT}, pages 479--482, 2020.

\bibitem[BB20]{BB20}
Matthew Brennan and Guy Bresler.
\newblock Reducibility and statistical-computational gaps from secret leakage.
\newblock In {\em Proceedings of 33rd COLT}, pages 648--847, 2020.

\bibitem[BBH{\etalchar{+}}21]{BBHLS21}
Matthew Brennan, Guy Bresler, Samuel~B. Hopkins, Jerry Li, and Tselil Schramm.
\newblock Statistical query algorithms and low-degree tests are almost
  equivalent.
\newblock {\em Proceedings of 34th COLT}, page 774, 2021.

\bibitem[BBK{\etalchar{+}}21]{BBKMW21}
Afonso~S. Bandeira, Jess Banks, Dmitriy Kunisky, Cristopher Moore, and
  Alexander~S. Wein.
\newblock Spectral planting and the hardness of refuting cuts, colorability,
  and communities in random graphs.
\newblock {\em Proceedings of 34th COLT}, pages 410--473, 2021.

\bibitem[BCN20]{BCN20}
Charles Bordenave, Simon Coste, and Raj~Rao Nadakuditi.
\newblock Detection thresholds in very sparse matrix completion.
\newblock {\em arXiv preprint arXiv:2005.06062}, 2020.

\bibitem[BGT10]{BGT10}
Mohsen Bayati, David Gamarnik, and Prasad Tetali.
\newblock Combinatorial approach to the interpolation method and scaling limits
  in sparse random graphs.
\newblock In {\em Proceedings of 42nd STOC}, pages 105--114, 2010.

\bibitem[BHK{\etalchar{+}}19]{BHKKMP19}
Boaz Barak, Samuel~B. Hopkins, Jonathan Kelner, Pravesh~K. Kothari, Ankur
  Moitra, and Aaron Potechin.
\newblock A nearly tight sum-of-squares lower bound for the planted clique
  problem.
\newblock {\em SIAM Journal on Computing}, 48(2):687--735, 2019.

\bibitem[BKW20]{BKW20}
Afonso~S. Bandeira, Dmitriy Kunisky, and Alexander~S. Wein.
\newblock Computational hardness of certifying bounds on constrained {PCA}
  problems.
\newblock In {\em Proceedings of 11th ITCS}, 2020.

\bibitem[BM11]{BM11}
Mohsen Bayati and Andrea Montanari.
\newblock The dynamics of message passing on dense graphs, with applications to
  compressed sensing.
\newblock {\em IEEE Transactions on Information Theory}, 57(2):764--785, 2011.

\bibitem[BMZ05]{BMZ05}
Alfredo Braunstein, Marc M{\'e}zard, and Riccardo Zecchina.
\newblock Survey propagation: an algorithm for satisfiability.
\newblock {\em Random Structures \& Algorithms}, 27(2):201--226, 2005.

\bibitem[BS19]{BS19}
Nikhil Bansal and Joel~H. Spencer.
\newblock On-line balancing of random inputs.
\newblock {\em arXiv preprint arXiv:1903.06898}, 2019.

\bibitem[BSW99]{BW99}
Eli Ben-Sasson and Avi Wigderson.
\newblock Short proofs are narrow -- resolution made simple.
\newblock In {\em Proceedings of 31st STOC}, pages 517--526, 1999.

\bibitem[CGPR19]{CGPR19}
Wei-Kuo Chen, David Gamarnik, Dmitry Panchenko, and Mustazee Rahman.
\newblock Suboptimality of local algorithms for a class of max-cut problems.
\newblock {\em Annals of Probability}, 47(3):1587--1618, 2019.

\bibitem[CHK{\etalchar{+}}20]{CHKRT20}
Yeshwanth Cherapanamjeri, Samuel~B. Hopkins, Tarun Kathuria, Prasad
  Raghavendra, and Nilesh Tripuraneni.
\newblock Algorithms for heavy-tailed statistics: Regression, covariance
  estimation, and beyond.
\newblock In {\em Proceedings of 52nd STOC}, pages 601--609, 2020.

\bibitem[CM19]{CM19}
Michael Celentano and Andrea Montanari.
\newblock Fundamental barriers to high-dimensional regression with convex
  penalties.
\newblock {\em arXiv preprint arXiv:1903.10603}, 2019.

\bibitem[CMM09]{CMM09}
Moses Charikar, Konstantin Makarychev, and Yury Makarychev.
\newblock Integrality gaps for sherali-adams relaxations.
\newblock In {\em Proceedings of 41st STOC}, pages 283--292, 2009.

\bibitem[CO10]{Coj10}
Amin Coja-Oghlan.
\newblock A better algorithm for random $k$-{SAT}.
\newblock {\em SIAM Journal on Computing}, 39:2823--2864, 2010.

\bibitem[COE15]{CE15}
Amin Coja-Oghlan and Charilaos Efthymiou.
\newblock On independent sets in random graphs.
\newblock {\em Random Structures \& Algorithms}, 47(3):436--486, 2015.

\bibitem[COFF{\etalchar{+}}09]{CFFKV09}
Amin Coja-Oghlan, Uriel Feige, Alan Frieze, Michael Krivelevich, and Dan
  Vilenchik.
\newblock On smoothed $k$-{CNF} formulas and the walksat algorithm.
\newblock In {\em Proceedings of 20th SODA}, pages 451--460, 2009.

\bibitem[COHH17]{CHH17}
Amin Coja-Oghlan, Amir Haqshenas, and Samuel Hetterich.
\newblock Walksat stalls well below the satisfiability threshold.
\newblock {\em SIAM Journal on Discrete Mathematics}, 31:1160--1173, 2017.

\bibitem[Coo71]{Coo71}
Stephen Cook.
\newblock The complexity of theorem proving procedures.
\newblock In {\em Proceedings of 3rd STOC}, pages 151--158, 1971.

\bibitem[COP16]{CP16}
Amin Coja-Oghlan and Konstantinos Panagiotou.
\newblock The asymptotic $k$-{SAT} threshold.
\newblock {\em Advances in Mathematics}, 288:985--1068, 2016.

\bibitem[CR92]{CR92}
V{\'a}clav Chv{\'a}tal and Bruce Reed.
\newblock Mick gets some (the odds are on his side).
\newblock In {\em Proceedings of 33th FOCS}, pages 620--627, 1992.

\bibitem[CSS18]{CSS18}
Ruiwen Chen, Rahul Santhanam, and Srikanth Srinivasan.
\newblock Average-case lower bounds and satisfiability algorithms for small
  threshold circuits.
\newblock {\em Theory of Computing}, 14(9):1--55, 2018.

\bibitem[DKWB20]{DKWB19}
Yunzi Ding, Dmitriy Kunisky, Alexander~S. Wein, and Afonso~S. Bandeira.
\newblock Subexponential-time algorithms for sparse {PCA}.
\newblock {\em arXiv preprint arXiv:1907.11635}, 2020.

\bibitem[DLL61]{DLL61}
Martin Davis, George Logemann, and Donald Loveland.
\newblock A machine program for theorem proving.
\newblock {\em Communications of the ACM}, 5(7):394--397, 1961.

\bibitem[DMM09]{DMM09}
David~L. Donoho, Arian Maleki, and Andrea Montanari.
\newblock Message-passing algorithms for compressed sensing.
\newblock {\em Proceedings of the National Academy of Sciences},
  106(45):18914--18919, 2009.

\bibitem[DMMZ08]{DMMZ08}
Herv{\'e} Daud{\'e}, Marc M{\'e}zard, Thierry Mora, and Riccardo Zecchina.
\newblock Pairs of {SAT} assignment in random boolean formulae.
\newblock {\em Theoretical Computer Science}, 393:260--279, 2008.

\bibitem[DP60]{DP60}
Martin Davis and Hilary Putnam.
\newblock A computing procedure for quantification theory.
\newblock {\em Journal of the ACM}, 7(3):201--205, 1960.

\bibitem[DSS15]{DSS15}
Jian Ding, Allan Sly, and Nike Sun.
\newblock Proof of the satisfiability conjecture for large $k$.
\newblock In {\em Proceedings of 47th STOC}, pages 59--68, 2015.

\bibitem[FP83]{FP83}
John Franco and Marvin Paull.
\newblock Probabilistic analysis of the {D}avis-{P}utnam procedure for solving
  the satisfiability problem.
\newblock {\em Discrete Applied Mathematics}, 5(1):77--87, 1983.

\bibitem[Fri90]{Fri90}
Alan Frieze.
\newblock On the independence number of random graphs.
\newblock {\em Discrete Mathematics}, 81(2):171--175, 1990.

\bibitem[FS96]{FS96}
Alan Frieze and Stephen Suen.
\newblock Analysis of two simple heuristics on a random instance of $k$-{SAT}.
\newblock {\em Journal of Algorithms}, 20:312--355, 1996.

\bibitem[FSS84]{FSS84}
Merrick~L. Furst, James~B. Saxe, and Michael Sipser.
\newblock Parity, circuits, and the polynomial-time hierarchy.
\newblock {\em Mathematical Systems Theory}, 17(1):13--27, 1984.

\bibitem[Gam21]{Gam21}
David Gamarnik.
\newblock The overlap gap property: A topological barrier to optimizing over
  random structures.
\newblock {\em Proceedings of the National Academy of Sciences}, 118(41), 2021.

\bibitem[GJ21]{GJ21}
David Gamarnik and Aukosh Jagannath.
\newblock The overlap gap property and approximate message passing algorithms
  for $p$-spin models.
\newblock {\em The Annals of Probability}, 49(1):180--205, 2021.

\bibitem[GJS19]{GJS19}
David Gamarnik, Aukosh Jagannath, and Subhabrata Sen.
\newblock The overlap gap property in principal submatrix recovery.
\newblock {\em arXiv preprint arXiv:1908.09959}, 2019.

\bibitem[GJW20]{GJW20}
David Gamarnik, Aukosh Jagannath, and Alexander~S. Wein.
\newblock Low-degree hardness of random optimization problems.
\newblock In {\em Proceedings of 61st FOCS}, pages 131--140, 2020.

\bibitem[GJW21]{GJW21}
David Gamarnik, Aukosh Jagannath, and Alexander~S. Wein.
\newblock Circuit lower bounds for the $p$-spin optimization problem.
\newblock {\em arXiv preprint arXiv:2109.01342}, 2021.

\bibitem[GK21a]{GK21}
David Gamarnik and Eren~C. K{\i}z{\i}lda{\u{g}}.
\newblock Algorithmic obstructions in the random number partitioning problem.
\newblock {\em arXiv preprint arXiv:2103.01369}, 2021.

\bibitem[GK21b]{GK21b}
David Gamarnik and Eren~C. K{\i}z{\i}lda{\u{g}}.
\newblock A curious case of symmetric binary perceptron model: algorithms and
  barriers.
\newblock Simons Institute presentation \url{https://youtu.be/io2OXE1Xw04},
  October 2021.

\bibitem[GL92]{GL92}
Craig Gotsman and Nathan Linial.
\newblock The equivalence of two problems on the cube.
\newblock {\em Journal of Combinatorial Theory, Series A}, 61(1):142--146,
  1992.

\bibitem[GL18]{GL18}
David Gamarnik and Quan Li.
\newblock Finding a large submatrix of a {G}aussian random matrix.
\newblock {\em The Annals of Statistics}, 46(6A):2511--2561, 2018.

\bibitem[GPB82]{GPB82}
Allen~T. Goldberg, Paul~W. Purdom, and Cynthia Brown.
\newblock Average time analysis of simplified {D}avis-{P}utnam procedures.
\newblock {\em Information Processing Letters}, 15:72--75, 1982.

\bibitem[Gri01]{Gri01}
Dima Grigoriev.
\newblock Linear lower bound on degrees of positivstellensatz calculus proofs
  for the parity.
\newblock {\em Theoretical Computer Science}, 259(1-2):613--622, 2001.

\bibitem[GS14]{GS14}
David Gamarnik and Madhu Sudan.
\newblock Limits of local algorithms over sparse random graphs.
\newblock In {\em Proceedings of 5th ITCS}, pages 369--376, 2014.

\bibitem[GS17]{GS17}
David Gamarnik and Madhu Sudan.
\newblock Performance of sequential local algorithms for the random
  {NAE}-$k$-{SAT} problem.
\newblock {\em SIAM Journal on Computing}, 46(2):590--619, 2017.

\bibitem[GZ17]{GZ17}
David Gamarnik and Ilias Zadik.
\newblock High-dimensional regression with binary coefficients. estimating
  squared error and a phase transition.
\newblock In {\em Proceedings of 30th COLT}, pages 948--953, 2017.

\bibitem[GZ19]{GZ19}
David Gamarnik and Ilias Zadik.
\newblock The landscape of the planted clique problem: dense subgraphs and the
  overlap gap property.
\newblock {\em arXiv preprint arXiv:1904.07174}, 2019.

\bibitem[H{\r{a}}s86]{Has86}
Johan H{\r{a}}stad.
\newblock Almost optimal lower bounds for small depth circuits.
\newblock In {\em Proceedings of 18th STOC}, pages 6--20, 1986.

\bibitem[Het16]{Het16}
Samuel Hetterich.
\newblock Analysing {S}urvey {P}ropagation guided decimation on random
  formulas.
\newblock In {\em Proceedings of 43rd ICALP}, 2016.

\bibitem[HKP{\etalchar{+}}17]{HKPRSS17}
Samuel~B. Hopkins, Pravesh~K. Kothari, Aaron Potechin, Prasad Raghavendra,
  Tselil Schramm, and David Steurer.
\newblock The power of sum-of-squares for detecting hidden structures.
\newblock In {\em Proceedings of 58th FOCS}, pages 720--731, 2017.

\bibitem[HMX21]{HMX21}
Jun-Ting Hsieh, Sidhanth Mohanty, and Jeff Xu.
\newblock Certifying solution geometry in random csps: counts, clusters and
  balance.
\newblock {\em arXiv preprint arXiv:2106.12710}, 2021.

\bibitem[Hop18]{Hop18}
Samuel~B. Hopkins.
\newblock {\em Statistical Inference and the Sum of Squares Method}.
\newblock PhD thesis, Cornell University, 2018.

\bibitem[HS17]{HS17}
Samuel~B. Hopkins and David Steurer.
\newblock Efficient {B}ayesian estimation from few samples: community detection
  and related problems.
\newblock In {\em Proceedings of 58th FOCS}, pages 379--390, 2017.

\bibitem[HS21]{HS21}
Brice Huang and Mark Sellke.
\newblock Tight {L}ipschitz hardness for optimizing mean field spin glasses.
\newblock {\em arXiv preprint arXiv:2110.07847}, 2021.

\bibitem[JM13]{JM13}
Adel Javanmard and Andrea Montanari.
\newblock State evolution for general approximate message passing algorithms,
  with applications to spatial coupling.
\newblock {\em Information and Inference: A Journal of the IMA}, 2(2):115--144,
  2013.

\bibitem[Kar76]{Kar76}
Richard~M. Karp.
\newblock {\em The probabilistic analysis of some combinatorial search
  algorithms}, pages 1--19.
\newblock Academic Press, 1976.

\bibitem[KKKS98]{KKKS98}
Lefteris~M. Kirousis, Evangelos Kranakis, Danny Krizanc, and Yannis~C.
  Stamatiou.
\newblock Approximating the unsatisfiability threshold of random formulas.
\newblock {\em Random Structures \& Algorithms}, 12(3):253--269, 1998.

\bibitem[KMOW17]{KMOW17}
Pravesh~K. Kothari, Ryuhei Mori, Ryan O’Donnell, and David Witmer.
\newblock Sum of squares lower bounds for refuting any {CSP}.
\newblock In {\em Proceedings of 49th STOC}, pages 132--145, 2017.

\bibitem[KMRT{\etalchar{+}}07]{KMRSZ07}
Florent Krzakala, Andrea Montanari, Federico Ricci-Tersenghi, Guilhem
  Semerjian, and Lenka Zdeborov{\'a}.
\newblock Gibbs states and the set of solutions of random constraint
  satisfaction problems.
\newblock {\em Proceedings of the National Academy of Sciences},
  104:10318--10323, 2007.

\bibitem[Kut02]{Kut02}
Samuel Kutin.
\newblock Extensions to {M}c{D}iarmid's inequality when differences are bounded
  with high probability.
\newblock Technical report, University of Chicago, Department of Computer
  Science, 2002.

\bibitem[KWB19]{KWB19}
Dmitriy Kunisky, Alexander~S. Wein, and Afonso~S. Bandeira.
\newblock Notes on computational hardness of hypothesis testing: predictions
  using the low-degree likelihood ratio.
\newblock {\em arXiv preprint arXiv:1907.11636}, 2019.

\bibitem[LMS98]{LMS98}
Michael~G. Luby, Michael Mitzenmacher, and M.~Amin Shokrollahi.
\newblock Analysis of random processes via and-or tree evaluation.
\newblock In {\em Proceedings of 9th SODA}, pages 364--373, 1998.

\bibitem[LZ20]{LZ20}
Yuetian Luo and Anru~R. Zhang.
\newblock Tensor clustering with planted structures: statistical optimality and
  computational limits.
\newblock {\em arXiv preprint arXiv:2005.10743}, 2020.

\bibitem[Mon19]{Mon19}
Andrea Montanari.
\newblock Optimization of the {S}herrington-{K}irkpatrick hamiltonian.
\newblock In {\em Proceedings of 60th FOCS}, pages 1417--1433, 2019.

\bibitem[MPZ02]{MPZ02}
Marc M{\'e}zard, Giorgio Parisi, and Riccardo Zecchina.
\newblock Analytic and algorithmic solution of random satisfiability problems.
\newblock {\em Science}, 297:812--815, 2002.

\bibitem[MRTS07]{MRS07}
Andrea Montanari, Federico Ricci-Tersenghi, and Guilhem Semerjian.
\newblock Solving constraint satisfaction problems through {B}elief
  {P}ropagation-guided decimation.
\newblock In {\em Proceedings of 45th Allerton}, pages 352--359, 2007.

\bibitem[MTF90]{CF90}
Chao Ming-Te and John Franco.
\newblock Probabilistic analysis of a generalization of the unit-clause literal
  selection heuristic for the $k$-satisfiability problem.
\newblock {\em Information Sciences}, 51:289--314, 1990.

\bibitem[NSS20]{NSS20}
Danny Nam, Allan Sly, and Youngtak Sohn.
\newblock One-step replica symmetry breaking of random regular {NAE}-$k$-{SAT}.
\newblock {\em arXiv preprint arXiv:2011.14270}, 2020.

\bibitem[O'D14]{Odo14}
Ryan O'Donnell.
\newblock {\em Analysis of {B}oolean functions}.
\newblock Cambridge University Press, 2014.

\bibitem[Pap91]{Pap91}
Christos~H. Papadimitriou.
\newblock On selecting a satisfying truth assignment.
\newblock In {\em Proceedings of 32nd FOCS}, pages 163--169, 1991.

\bibitem[PX21]{PX21}
Will Perkins and Changji Xu.
\newblock Frozen $1$-{RSB} structure of the symmetric {I}sing perceptron.
\newblock {\em arXiv preprint arXiv:2102.05163}, 2021.

\bibitem[RV17]{RV17}
Mustazee Rahman and B{\'a}lint Vir{\'a}g.
\newblock Local algorithms for independent sets are half-optimal.
\newblock {\em The Annals of Probability}, 45(3):1543--1577, 2017.

\bibitem[Sch08]{Sch08}
Grant Schoenebeck.
\newblock Linear level lasserre lower bounds for certain $k$-{CSP}s.
\newblock In {\em Proceedings of 49th FOCS}, pages 593--602, 2008.

\bibitem[Sel21]{Sel21}
Mark Sellke.
\newblock Optimizing mean field spin glasses with external field.
\newblock {\em arXiv preprint arXiv:2105.03506}, 2021.

\bibitem[SW20]{SW20}
Tselil Schramm and Alexander~S. Wein.
\newblock Computational barriers to estimation from low-degree polynomials.
\newblock {\em arXiv preprint arXiv:2008.02269}, 2020.

\bibitem[Wei20]{Wei20}
Alexander~S. Wein.
\newblock Optimal low-degree hardness of maximum independent set.
\newblock {\em arXiv preprint arXiv:2010.06563}, 2020.

\end{thebibliography}
